\documentclass[nofootinbib,eqsecnum,superscriptaddress,amsmath,amssymb,aps,11pt]{revtex4-2}

\usepackage{amsmath,amssymb,amsfonts}
\usepackage{algorithmic}
\usepackage{graphicx}
\usepackage{textcomp}
\usepackage{color} 
\usepackage[]{hyperref}
\usepackage{comment}
\usepackage{braket}
\usepackage{amsthm}
\newtheorem{lemma}{Lemma}[section]
\usepackage{tikz}
\usetikzlibrary{quantikz2}
\usepackage{booktabs} 
\usepackage{multirow} 

\newcommand{\secref}[1]{Sec.~\ref{#1}}
\newcommand{\figref}[1]{Fig.~\ref{#1}}
\newcommand{\appref}[1]{Appendix~\ref{#1}}
\newcommand{\tabref}[1]{Table~\ref{#1}}

\DeclareMathOperator{\Tr}{Tr} 

\newcommand{\corr}[1]{\textcolor{blue}{#1}} 

\begin{document}

\title{ShotQC: Reducing Sampling Overhead in Quantum Circuit Cutting}

\author{Po-Hung Chen}
\affiliation{Department of Electrical and Computer Engineering, The University of Texas at Austin, TX 78703, USA}
\affiliation{Department of Electrical Engineering, National Taiwan University, Taipei 10617, Taiwan}

\author{Bo-Hung Chen}
\affiliation{Department of Electrical Engineering, National Taiwan University, Taipei 10617, Taiwan}
\affiliation{Graduate Institute of Electronics Engineering, National Taiwan University, Taipei 10617, Taiwan}
\affiliation{Center for Quantum Science and Engineering, National Taiwan University, Taipei 10617, Taiwan}
\affiliation{Physics Division, National Center for Theoretical Sciences, Taipei 10617, Taiwan}

\author{Dah-Wei Chiou}
\affiliation{Department of Electrical Engineering, National Taiwan University, Taipei 10617, Taiwan}
\affiliation{Graduate Institute of Electronics Engineering, National Taiwan University, Taipei 10617, Taiwan}
\affiliation{Center for Quantum Science and Engineering, National Taiwan University, Taipei 10617, Taiwan}
\affiliation{Physics Division, National Center for Theoretical Sciences, Taipei 10617, Taiwan}

\author{Jie-Hong Roland Jiang}
\email{jhjiang@ntu.edu.tw}
\affiliation{Department of Electrical Engineering, National Taiwan University, Taipei 10617, Taiwan}
\affiliation{Graduate Institute of Electronics Engineering, National Taiwan University, Taipei 10617, Taiwan}
\affiliation{Center for Quantum Science and Engineering, National Taiwan University, Taipei 10617, Taiwan}
\affiliation{Physics Division, National Center for Theoretical Sciences, Taipei 10617, Taiwan}

\begin{abstract}
The recently developed \emph{quantum circuit cutting} technique greatly extends the capabilities of current noisy intermediate-scale quantum (NISQ) hardware. However, it introduces substantial overhead in both classical postprocessing and quantum  resources, as the postprocessing complexity and sampling cost scale exponentially with the number of circuit cuts. In this work, we propose an enhanced circuit cutting framework, \emph{ShotQC}, which effectively reduces the sampling overhead through two key optimizations: \emph{shot distribution} and \emph{cut parameterization}. The former employs an adaptive Monte Carlo strategy to dynamically allocate more quantum shots to subcircuit configurations that contribute more to the variance in the final outcome. The latter exploits additional degrees of freedom in postprocessing to further suppress variance. Integrating these optimizations, ShotQC significantly reduces the sampling overhead without increasing classical postprocessing complexity, as demonstrated across a range of benchmark circuits.
\end{abstract}
\maketitle

\section{Introduction}


Despite the development of many promising quantum algorithms~\cite{365700,10.1145/237814.237866}, current quantum computers lack the scalability and reliability required for practically relevant applications. Most devices today fall into the category of noisy intermediate-scale quantum (NISQ) systems~\cite{Preskill2018quantumcomputingin}, featuring a moderate number of qubits and limited fidelity due to noise~\cite{10487283}. As an intermediate stage between classical and fully error-corrected quantum computers, their scalability remains fundamentally constrained by both qubit count and noise.

To mitigate these limitations, noise-tolerant algorithms such as the Quantum Approximate Optimization Algorithm (QAOA)~\cite{farhi2014quantumapproximateoptimizationalgorithm} and the Variational Quantum Eigensolver (VQE) \cite{Peruzzo2014} have been developed to solve practical optimization problems on today’s noisy quantum hardware. Despite these advances, scalability challenges remain a significant obstacle, restricting the advantages of quantum algorithms over classical approaches, particularly for solving larger-scale, real-world problems.

While there are various promising quantum algorithms, current quantum computers lacks the scale and reliability to apply them on practically relevant problems. Instead, present quantum computers are noisy intermediate-scale quantum (NISQ) devices~\cite{Preskill2018quantumcomputingin}. As an intermediate hardware between classical computers and fully error-corrected quantum computers, NISQ devices have a moderate number of qubits but its scale is inherently limited by insufficient qubit fidelity and high error rates. Adapting to these limitations, some noise-tolerant algorithms such as Quantum Approximate Optimization Algorithm (QAOA)~\cite{farhi2014quantumapproximateoptimizationalgorithm} and Variational Quantum Eigensolver (VQE)~\cite{Peruzzo2014} are developed to solve practical optimization problems with today's noisy quantum devices. Nevertheless, the scalability issue still hinders these algorithms from offering significant benefits over classical devices, particularly for larger problem instances.

To address these obstacles, various hybrid approaches that utilize both classical and quantum resources have been proposed. Among these, the technique of \emph{quantum circuit cutting} has emerged as a powerful tool for enabling the simulation of large quantum circuits on smaller devices by partitioning the circuit into disconnected subcircuits, which can be executed separately and then recombined through post-processing~\cite{Peng_2020}. This method not only facilitates the distribution of computational workloads across multiple smaller devices but also offers distinct advantages over executing a single large circuit. For instance, it enhances fidelity~\cite{Perlin2021,PhysRevLett.130.110601} and reduces noise~\cite{Bechtold_2023}, as smaller devices are typically less prone to noise. These benefits make circuit cutting a promising approach for addressing scalability challenges, particularly in the NISQ era.


The cost of circuit cutting, however, can be substantial. Specifically, the quantum resources required for repeated sampling\footnote{To accurately estimate the probability distribution of the outcomes, the circuit must be run numerous times, commonly referred to as ``shots.''} and the classical resources needed for postprocessing computations both scale exponentially with the number of cuts~\cite{brenner2023optimalwirecuttingclassical}. Consequently, identifying optimal cut points to minimize the total number of cuts is crucial~\cite{Tang_2021,10374226}. 
Other techniques, such as optimized channel decomposition~\cite{harada2023doublyoptimalparallelwire}, randomized measurements~\cite{Lowe_2023}, and sampling-based reconstruction~\cite{Lian2023}, have been proposed to mitigate the associated overhead. Furthermore, the overhead can be reduced by decreasing the size of the observable set required for each cut~\cite{Tang_2021}. State-of-the-art methods achieve this through approaches like rotating the quantum state~\cite{10025537}, designing circuits with cutting efficiency in mind~\cite{10196555}, or detecting optimal cut points dynamically during execution~\cite{10313822}.


In this work, we propose an enhanced circuit cutting framework, \emph{ShotQC}, to reduce sampling overhead by incorporating two novel techniques: \emph{shot distribution optimization} and \emph{cut parameter optimization}. 
The former dynamically allocates shot numbers across subcircuits using an adaptive Monte Carlo-inspired strategy, while the latter exploits additional degrees of freedom in the postprocessing formula.
Together, these techniques enable ShotQC to reduce sampling overhead, achieving an improvement rate of up to $80{,}000\times$ in empirical benchmarks. 
Theoretical analysis and a benchmark case study further suggest that this improvement rate tend to increase with the degree of probability mass concentration.

Additionally, we implement the ShotQC framework with GPU acceleration, demonstrating its robustness and quantifying the trade-off between classical and quantum resources across various settings.




\section{Circuit Cutting}
This section presents the essential theory and an overview of circuit cutting,
which decompose an identity channel into multiple \textit{measure-and-prepare} channels.

\subsection{Theory} \label{subsec:theory}
In circuit cutting, the primary goal is to reconstruct the probability distribution of the original (uncut) circuit from those of the subcircuits. This can be achieved by computing the value of the tensor network via the lemma presented below.

If any $2 \times 2$ matrix $A$ can always be decomposed as
\begin{equation}
    \label{eq:decomp}
    A = \sum_{i=1}^\ell r_i \Tr[A O_i]\rho_i,
\end{equation}
in terms of a given set of coefficients $r_i$, observables $O_i$, and single-qubit states $\rho_i$, then the following lemma holds.
\begin{lemma}
\label{lemma:cut}
Let $(G(E, V), \mathcal{A})$ be a tensor network of a quantum circuit, and let $(u\rightarrow v)\in E$ be an edge in $G$ (with $u,v\in V$). Replacing the edge $(u\rightarrow v)$ with the summation of measure-and-prepare channels as
\begin{equation}\label{eq:cut}
\begin{tikzpicture}

\begin{scope}[shift={(0,0)}]
\draw [thick, ->] (0,0) -- (1-0.05,0);
\draw[fill] (0,0) circle [radius=0.05];
\draw[fill] (1,0) circle [radius=0.05];
\node [above] at (0,0) {$u$};
\node [above] at (1,0) {$v$};
\end{scope}

\begin{scope}[shift={(3.5,0)}]
\draw [thick,->] (0,0) -- (0.5,0);
\draw[fill] (0,0) circle [radius=0.05];
\node [above] at (0,0) {$u$};
\draw [thick] (0.5,-0.4) -- (1.1,0) -- (0.5,0.4) -- cycle;
\node [] at (0.75,0) {{\small $O_i$}};
\draw [thick] (1.3,0) -- (1.9,0.4) -- (1.9,-0.4) -- cycle;
\node [] at (1.7,0) {{\small $\rho_i$}};
\draw [thick,<-] (1.9,0) -- (1.9+0.5,0);
\draw[fill] (1.9+0.5,0) circle [radius=0.05];
\node [above] at (1.9+0.5,0) {$v$};
\end{scope}

\node [] at (1.8,0) {$\equiv$};
\node [] at (2.8,0) {$\sum_{i=1}^\ell r_i$};

\end{tikzpicture}
\end{equation}
keeps the value of the tensor network unchanged, provided that $r_i$, $O_i$, and $\rho_i$ satisfy \eqref{eq:decomp}. More precisely, by cutting the edge $(u\rightarrow v)$ and attaching a pair of $\lhd$ and $\rhd$ vertices to the cut endpoints, the graph $G$ is transformed into $G'$, and each term in the summation corresponds to a distinct tensor network $(G', \mathcal{A}_i)$. The value of the original tensor network $T(G, \mathcal{A})$ is then identical to
\begin{equation}\label{eq:tensor value}
    T(G, \mathcal{A}) = \sum_{i=1}^\ell r_i T(G', \mathcal{A}_i).
\end{equation}
\end{lemma}
This lemma provides the theoretical foundation for \emph{wire cutting}, a technique that decomposes a quantum wire into a series of measure-and-prepare channels. For a brief background on the tensor network representation of quantum circuits and the proof of Lemma~\ref{lemma:cut}, see Appendix~\ref{sec: back}.

In many quantum computing applications, a large circuit can be partitioned into disconnected subcircuits by cutting only a few wires. This enables the simulation of large circuits beyond the capacity of current quantum hardware by executing subcircuits with different $O_i$ and $\rho_i$, and reconstructing the result of the original circuit via~\eqref{eq:tensor value}~\cite{Peng_2020}.

The mathematical identity
\begin{equation}\label{eq:mitA}
    A=\frac{\Tr(AI)I + \Tr(AX)X + \Tr(AY)Y + \Tr(AZ)Z}{2}
\end{equation}
for any arbitrary $2\times2$ matrix $A$ straightforwardly suggests the wire cutting scheme, as suggested in \cite{Peng_2020}, with
\begin{subequations} \label{eq:mit}
\begin{align}
r_1&=+1/2, & O_1&=I, & \rho_1&=\ket{0}\bra{0},\\
r_2&=+1/2, & O_2&=I, & \rho_2&=\ket{1}\bra{1},\\
r_3&=+1/2, & O_3&=X, & \rho_3&=\ket{+}\bra{+},\\
r_4&=-1/2, & O_4&=X, & \rho_4&=\ket{-}\bra{-},\\
r_5&=+1/2, & O_5&=Y, & \rho_5&=\ket{+i}\bra{+i},\\
r_6&=-1/2, & O_6&=Y, & \rho_6&=\ket{-i}\bra{-i},\\
r_7&=+1/2, & O_1&=Z, & \rho_7&=\ket{0}\bra{0},\\
r_8&=-1/2, & O_1&=Z, & \rho_8&=\ket{1}\bra{1},
\end{align}
\end{subequations}
where $\ket{\pm}:=(\ket{0}\pm\ket{1})/\sqrt{2}$ and $\ket{\pm i}:=(\ket{0}\pm i\ket{1})/\sqrt{2}$ are eigenstates of $X$ and $Y$, respectively.
This can be rearranged into an $\ell=6$ scheme that involves iterating over three different measurement bases ($X$, $Y$, and $Z$) and six preparation states ($\ket{0}$, $\ket{1}$, $\ket{+}$, $\ket{-}$, $\ket{+i}$, and $\ket{-i}$).

There are infinitely many possible constructions of $\{r_i\}$, $\{O_i\}$, and $\{\rho_i\}$. However, since the space of arbitrary $2 \times 2$ matrices is four-dimensional (over the complex numbers), the decomposition in \eqref{eq:decomp} requires at least four preparation states. One such example is the $\ell=4$ scheme, given by
\begin{subequations}\label{eq:cutqc}
\begin{align}
    \{r_i\} &= \{1/2, 1/2, 1/2, 1/2\},\label{eq:decomp_c}\\
    \{O_i\} &= \left\{2\ket{0}\bra{0}-X, 2\ket{1}\bra{1}-Y, X, Y\right\},\label{eq:decomp_O}\\
    \{\rho_i\} &= \left\{\ket{0}\bra{0}, \ket{1}\bra{1}, \ket{+}\bra{+}, \ket{+i}\bra{+i}\right\}\label{eq:decomp_rho},
\end{align}
\end{subequations}
which involves three measurement bases ($X$, $Y$, and $Z$) and four preparation states ($\ket{0}$, $\ket{1}$, $\ket{+}$, and $\ket{+i}$), matching the minimum number of states required. This decomposition was proposed in ~\cite{Tang_2021} and will be used in this section as an illustrative example.

By restricting ourselves to Pauli bases and states, which are easier to prepare and measure, we propose a generic parameterized decomposition in Section \ref{sec:cut para}, of which the schemes in both \eqref{eq:mit} and \eqref{eq:cutqc} can be viewed as special cases.

The observable $O_i$ corresponds to measuring the expectation value of $O_i$ under specific tensor constraints. For example, the observable $\ket{0}\bra{0}$ corresponds to the probability of obtaining the $\ket{0}$ outcome, conditioned on the state of other qubits, when the qubit is measured in the $Z$ basis.
Similarly, each state $\rho_i$ corresponds to preparing the qubit in the respective quantum state. The construction of $\{c_i\}$, $\{O_i\}$, and $\{\rho_i\}$ collectively resembles multiple measure-and-prepare channels, which form the core mechanism behind circuit cutting.

\subsection{Workflow}
We now shift our perspective from a theoretical view to a practical one by describing the procedure for circuit cutting. We must determine the measurement set $S_\text{meas}$ and the preparation set $S_\text{prep}$, which specify the measurement bases and preparation states used at each measure-and-prepare channel. These sets can be inferred from a particular choice made for \eqref{eq:decomp}. For instance, given the observable set $\{O_i\}$ and the state set $\{\rho_d\}$ in \eqref{eq:cutqc}, the measurement set 
\begin{equation}
    \label{eq:smeas}
    S_\text{meas} = \{ X, Y, Z\}
\end{equation}
and the preparation set 
\begin{equation}
    \label{eq:sprep}
    S_\text{prep} = \{\ket{0}, \ket{1}, \ket{+}, \ket{+i}\}
\end{equation}
are sufficient, since the measurements in the $X$, $Y$, $Z$ bases provide the results for the observables $\ket{0/1}\bra{0/1}$, $\ket{\pm}\bra{\pm}$, and$\ket{\pm i}\bra{\pm i}$, respectively.

After specifying the sets, the workflow for circuit cutting can generally be divided into three steps:

\begin{enumerate}
    \item \textbf{Finding cuts}. Choose the desired cut points that divide the original circuit into multiple subcircuits. These cut points are optimally selected to minimize metrics such as postprocessing cost and sampling overhead, subject to width or depth constraints. Once the subcircuits are defined, generate the \emph{subcircuit configurations}, i.e., all possible variants of the subcircuits applying different measurement bases from $S_\text{meas}$ and preparation states from $S_\text{prep}$.
    
    \item \textbf{Executing subcircuits}.
    Evaluate the subcircuit configurations using classical simulation or quantum hardware. When using quantum hardware, the number of shots must be predetermined based on the desired accuracy and selected strategies. 
    Note that the number of subcircuit configurations scales exponentially with the number of cuts, and thus the runtime for evaluating all subcircuits will also scale exponentially unless parallel computing is applied.
    
    \item \textbf{Postprocessing}. Combine the probability distributions obtained from the previous step according to the circuit structure and the cut parameters. Depending on what the outcome of interest is (e.g., the user may be interested in the full spectrum of all outcomes or the probability of a particular outcome), different postprocessing strategies can be applied. For an $\ell=4$ decomposition, a total of $4^K$ Kronecker products are required for classical postprocessing, where $K$ is the number of cuts~\cite{Tang_2021}.
\end{enumerate}

\begin{figure}
    \centering
    \includegraphics[width=0.9\columnwidth]{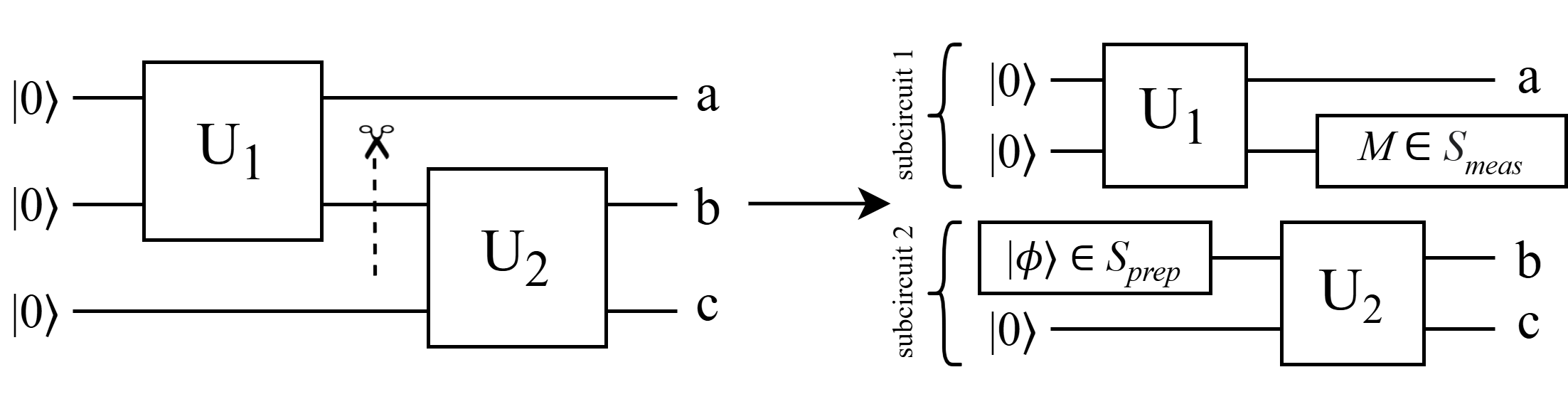}
    \caption{A 3-qubit circuit is cut into two 2-qubit subcircuits.}
    \label{fig:qc_example}
\end{figure}

\subsection{Illustration} \label{subsec:example}

We now demonstrate these steps using the example shown in Figure \ref{fig:qc_example}, where a 3-qubit circuit is cut into two 2-qubit quantum subcircuits. We employ the decomposition \eqref{eq:decomp} with the corresponding $S_\text{meas}$ in \eqref{eq:smeas} and $S_\text{prep}$ in \eqref{eq:sprep}.

First, since the cut has already been located, no additional search is required. The original wire being cut gives rise to a measurement operation in subcircuit 1 and an initial state preparation in subcircuit 2. Once the subcircuits are defined by the cut, the subcircuit configurations are generated by applying different elements from $S_\text{meas}$ and $S_\text{prep}$. In this case, three subcircuit configurations are generated for subcircuit 1, and four subcircuit configurations are generated for subcircuit 2.

Second, all subcircuit configurations must be evaluated. Suppose that each subcircuit configuration is evaluated on a quantum computer with 1000 shots, leading to a total of 7000 shots across the seven subcircuit configurations. If multiple quantum computers are available, these configurations can be executed in parallel without the need for classical communication, as they are independent of one another. The output probabilities for each subcircuit configuration are computed as the ratio of the number of occurrences of each outcome to the 1000 shots.

Finally, the output probability of the original circuit can be recovered using Lemma \ref{lemma:cut}. Let $(a,b,c) \in \{0,1\}^3$ represent a bitstring outcome of the original circuit. Based on the decomposition in \eqref{eq:cutqc}, the probability $p(a,b,c)$ of obtaining the bitstring outcome $(a,b,c)$ is given by
\begin{equation}\label{eq:sump}
     p(a,b,c) = \sum_{i=1}^4 r_i\,  p^{(1)}_i(a) \, p^{(2)}_i(b,c),
\end{equation}
where $r_i$ are given by \eqref{eq:decomp_c}, and $p^{(1)}_i$ and $p^{(2)}_i$, defined below, are derived from the probabilities associated with the subcircuit configurations of subcircuit 1 and subcircuit 2, respectively.

In accordance with \eqref{eq:decomp_O}, we define
\begin{subequations}\label{eq:p1}
\begin{align}
     p^{(1)}_1(a) &= 2p^{(1)}(a0|Z) - [p^{(1)}(a0|X) - p^{(1)}(a1|X)], \\
     p^{(1)}_2(a) &= 2p^{(1)}(a1|Z) - [p^{(1)}(a0|Y) - p^{(1)}(a1|Y)], \\
     p^{(1)}_3(a) &= p^{(1)}(a0|X) - p^{(1)}(a1|X), \\
     p^{(1)}_4(a) &= p^{(1)}(a0|Y) - p^{(1)}(a1|Y),
\end{align}
\end{subequations}
where $p^{(1)}(a0|Z)$ represents the probability of obtaining the outcome $(a,0)$ when measuring in the $Z$ basis at the cut location in subcircuit 1, with similar notation applying to the other terms.

In accordance with \eqref{eq:decomp_rho}, we define
\begin{subequations}\label{eq:p2}
\begin{align}
    p^{(2)}_1(bc) &= p^{(2)}(bc|\ket{0}), \\
    p^{(2)}_2(bc) &= p^{(2)}(bc|\ket{1}), \\
    p^{(2)}_3(bc) &= p^{(2)}(bc|\ket{+}), \\
    p^{(2)}_4(bc) &= p^{(2)}(bc|\ket{+i}),
\end{align}
\end{subequations}
where $p^{(2)}(bc|\ket{0})$ represents the probability of obtaining the outcome $(b,c)$ when preparing the initial state $\ket{0}$ at the cut location in subcircuit 2, with similar notation applying to the other terms.

By iterating over all 8 possible bitstrings $(a,b,c)$, we obtain the full probability distribution of the original circuit. If only certain outcomes are of interest, we can restrict the query to those specific bitstrings. The more shots executed for each subcircuit configuration in the second step, the lower the expected error in the reconstructed probabilities. Consequently, the calculated probability of obtaining $(a,b,c)$ will converge to the ground-truth value, provided that a sufficient number of shots are performed for each subcircuit configuration.

\subsection{Sampling Overhead}
In addition to the classical postprocessing cost, the sampling overhead is a major concern in the circuit-cutting scenario. The sampling overhead refers to the multiplicative factor that indicates how much the total number of shots must be increased to achieve the same accuracy for the output probability as in the original, uncut circuit.

The sampling overhead scales as $O(16^K)$, where $K$ is the number of cuts. In the work~\cite{Lowe_2023}, this overhead was reduced to $O(4^K)$ using randomized measurements within the Local Operations and Classical Communication (LOCC) framework. However, the LOCC framework introduces additional resource demands and potential errors due to the need for synchronization.

In this work, we adhere to the no-communication setting. Instead of reducing the theoretical sampling overhead of $O(16^K)$, we propose an adaptive approach that strategically allocates quantum shots, significantly reducing the sampling overhead in most practical cases.
With appropriate modifications, it may be possible to combine the techniques proposed in~\cite{Lowe_2023} with our approach.

\section{Problem Formulation} \label{sec:problem}
Circuit cutting allows the simulation of large circuits using smaller ones. However, the exponential overhead in both classical resources (e.g., postprocessing costs) and quantum resources (e.g., the cost of performing measurements and state preparation) presents a significant challenge for large-scale applications. To address this, we propose a novel strategy that significantly reduces the sampling overhead associated with state preparation and measurement.



We aim for seamless integration with existing methods that minimize postprocessing costs by optimizing the locations of cuts. Accordingly, we assume that the circuit partition has been optimally determined in advance by cut-finding algorithms.

We formulate the problem as follows: given a quantum circuit partitioned into multiple subcircuits, the objective is to minimize the total sampling overhead required to reconstruct the original circuit's output probability distribution to a target accuracy. This is equivalent to minimizing the statistical error under constrained quantum resources, since increasing the number of shots consistently decreases the expected statistical error.

After performing a sufficient number of shots for all subcircuit configurations, we obtain estimated values for all subcircuit probabilities, such as $p^{(1)}(\cdot|\cdot)$ and $p^{(2)}(\cdot|\cdot)$ in \eqref{eq:p1} and \eqref{eq:p2}, along with their variance-covariance relations, which characterize the statistical uncertainties in these estimates and the dependencies between them.

Through \eqref{eq:tensor value}, exemplified by \eqref{eq:sump}, we then compute the \emph{estimated} probabilities $\{p_a\}$ for all $2^n$ possible bitstring outcomes $\{a=0,\dots,2^n-1\}$ of the original circuit’s $n$ qubits. The uncertainties are also propagated to $\{p_a\}$, described by a $2^n \times 2^n$ variance-covariance matrix $\Sigma$, in accordance with \eqref{eq:sump}. Here, $\Sigma_{ab} \equiv \sigma_{p_a p_b}$ represents the covariance between $p_a$ and $p_b$, and in particular, $\Sigma_{aa} \equiv \text{Var}(p_a)$ is the variance of $p_a$. This matrix $\Sigma$ defines the shape of a high-dimensional Gaussian distribution centered at $\{p_a\}$ in the hyperspace constrained by $\sum_a p_a = 1$. The Gaussian distribution represents the probability distribution for the \emph{ground-truth} probabilities $\{\hat{p}_a\}$, given the estimates $\{p_a\}$. The eigenvalues of $\Sigma$ quantify how closely $\{p_a\}$ approximates $\{\hat{p}_a\}$.\footnote{Note that one of the eigenvalues must vanish as a consequence of the constraint $\sum_a p_a = 1$. The corresponding eigenvector is the normal vector of the constrained hyperspace.}

To quantify the overall statistical uncertainty due to finite sampling, we define the \emph{total variance} as
\begin{equation}
    \label{eq:error}
    \text{Err} = \text{Tr}\, \Sigma \equiv \sum_{a=0}^{2^n-1} \text{Var}(p_a),
\end{equation}
which provides a measure of the cumulative uncertainty across all probabilities $\{p_a\}$.

\section{ShotQC: Framework}

\begin{figure}
    \centering
    \includegraphics[width=0.35\linewidth]{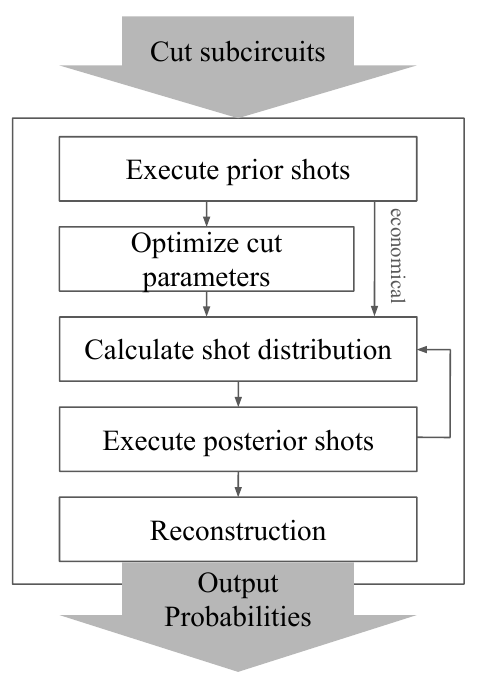}
    \caption{The workflow of ShotQC. The cut subcircuits is obtained by applying cut-finding methods on the original circuit.}
    \label{fig:overview}
\end{figure}


\subsection{ShotQC: Overview}

The workflow of ShotQC is depicted in Figure~\ref{fig:overview}. ShotQC incorporates two novel techniques: \emph{cut parameter optimization} and \emph{shot distribution optimization}. The former expands the parameter space for postprocessing without incurring additional costs, while the latter finds the optimal allocation of quantum shots to minimize the statistical error.
Once the optimal cut points are identified by cut-finding algorithms, the workflow proceeds to the third block in Figure~\ref{fig:overview}.


Given the disconnected subcircuits, subcircuit configurations are generated by applying different measurement bases from $S_\text{meas}$ and preparation states from $S_\text{prep}$. These configurations are then executed for a predefined number of repetitions, referred to as ``prior shots,'' to produce a rough estimate $P$ of the outcome probabilities for all subcircuit configurations.

Next, depending on the circuit size or the user's preference, the cut parameter optimization step is either performed---optimizing the cut parameters $\Theta$ with respect to the loss function defined in \eqref{eq:loss}---or skipped in the economical route, with $\Theta$ set to fixed values (e.g., all parameters in \tabref{table:final} set to zero).

Once $\Theta$ and $P$ are determined, ShotQC begins the iterative process of shot distribution optimization. In each iteration, the optimal shot distribution is computed using \eqref{eq:opt_distribute}, and then the subcircuit configurations are executed for a designated number of repetitions, referred to as ``posterior shots,'' following the optimal distribution. Based on the outcomes of these posterior shots, the estimate $P$ is then updated.

This iterative process continues until the total number of shots reaches a predefined limit. Finally, the outcomes of all shots are aggregated, and the outcome probabilities of the original circuit are reconstructed through classical postprocessing.

\subsection{ShotQC: cut parameter Optimization}\label{sec:cut para}
In Section~\ref{subsec:theory}, we establish that any decomposition of an arbitrary $2\times 2$ matrix satisfying \eqref{eq:decomp} corresponds to a valid wire-cutting scheme. This is subsequently reflected in the postprocessing phase, where the probability distribution of the original circuit is reconstructed based on the chosen decomposition, as illustrated in Section~\ref{subsec:example}. These findings naturally raise the question: How many valid combinations of $r_i$, $O_i$, and $\rho_i$ are possible? Furthermore, if we extend our analysis to more general decompositions (parameterized by some ``cut parameters''), can we devise a more efficient postprocessing method that further reduces statistical errors, even when utilizing the same subcircuit outcomes?

\subsubsection{Problem Reduction}
It turns out there are infinitely many possible choices for $r_i$, $O_i$, and $\rho_i$, even when $O_i$ and $\rho_i$ are constrained to special types, as shown below.

Mathematically, $\rho_i$ do not necessarily need to be valid density matrices. However, in practice, $\rho_i$ must be easy to prepare. Consequently, we restrict $\rho_i$ to be eigenstates of the Pauli matrices $X$, $Y$, and $Z$. That is, we consider
\begin{equation}
    S_\text{prep} = \{\ket{0}, \ket{1}, \ket{+}, \ket{-}, \ket{+i}, \ket{-i}\},
\end{equation}
which is the set of states used in ~\cite{Peng_2020}.
Similarly, we impose a restriction on the observables, limiting them to the commonly available measurement set
\begin{equation}
    S_\text{meas} = \{ X, Y, Z\}.
\end{equation}

With these restrictions in place, the decomposition \eqref{eq:decomp} can be recast as
\begin{equation}\label{eq:decomp generalized}
    A = \sum_{i=1}^6 \sum_{j=1}^6 r_{i,j}\text{Tr}[AO_j]\,\ket{\rho_i}\bra{\rho_i},
\end{equation}
where $\ket{\rho_i}\in S_\text{prep}$ and $O_j\in\{X_\pm=\ket{\pm}\bra{\pm}, Y_\pm=\ket{\pm i}\bra{\pm i},$ $Z_+=\ket{0}\bra{0}, Z_-=\ket{1}\bra{1}\}$, corresponding to the positive- and negative-eigenvalue outcomes of the measurement observables in $S_\text{meas}$.
The decomposition then can be represented by a $6\times6$ table of $r_{i,j}$, where the column indexes ($j$) correspond to the measurement outcomes, the row indexes ($i$) correspond to the initial states. For instance, the decomposition \eqref{eq:mitA} can be represented by Table \ref{table: mit}. Note that we have eliminated the $I$ observable since
\begin{equation}\label{eq:I2ZZ}
    I = I_++I_- = \ket{0}\bra{0}+\ket{1}\bra{1} =Z_++Z_-.
\end{equation}

\begin{table}[h!]
\caption{The tabular representation of the decomposition described in \eqref{eq:mitA}.}
\label{table: mit}
\centering
\scalebox{0.9}{
\begin{tabular}{|c|c|c|c|c|c|c|}
\hline
i\textbackslash j & $X_{+}$ & $X_{-}$ & $Y_{+}$ & $Y_{-}$ & $Z_{+}$ & $Z_{-}$ \\ \hline
$\ket{0}$ & $1/2$ & $-1/2$ & $0$ & $0$ & $0$ & $0$ \\ \hline
$\ket{1}$ & $-1/2$ & $1/2$ & $0$ & $0$ & $0$ & $0$ \\ \hline
$\ket{+}$ & $0$ & $0$ & $1/2$ & $-1/2$ & $0$ & $0$ \\ \hline
$\ket{-}$ & $0$ & $0$ & $-1/2$ & $1/2$ & $0$ & $0$ \\ \hline
$\ket{+i}$ & $0$ & $0$ & $0$ & $0$ & $1$ & $0$ \\ \hline
$\ket{-i}$ & $0$ & $0$ & $0$ & $0$ & $0$ & $1$ \\ \hline
\end{tabular}
}
\end{table}

\begin{table}[h!]
\caption{The tabular representation of a generalized decomposition described in \eqref{eq:decomp generalized}.}
\label{table:final}
\centering
\scalebox{0.9}{
\begin{tabular}{|c|c|c|c|c|c|c|}
\hline
i\textbackslash j & $X_{+}$ & $X_{-}$ & $Y_{+}$ & $Y_{-}$ & $Z_{+}$ & $Z_{-}$ \\ \hline
$\ket{0}$ & $1/2-a_1-c_1$ & $-1/2-a_1-c_2$ & $-b_1-c_3$ & $-b_1-c_4$ & $a_1+b_1-c_5$ & $a_1+b_1-c_6$ \\ \hline
$\ket{1}$ & $-1/2-a_2-c_1$ & $1/2-a_2-c_2$ & $-b_2-c_3$ & $-b_2-c_4$ & $a_2+b_2-c_5$ & $a_2+b_2-c_6$ \\ \hline
$\ket{+}$ & $-a_3-d_1$ & $-a_3-d_2$ & $1/2-b_3-d_3$ & $-1/2-b_3-d_4$ & $a_3+b_3-d_5$ & $a_3+b_3-d_6$ \\ \hline
$\ket{-}$ & $-a_4-d_1$ & $-a_4-d_2$ & $-1/2-b_4-d_3$ & $1/2-b_4-d_4$ & $a_4+b_4-d_5$ & $a_4+b_4-d_6$ \\ \hline
$\ket{+i}$ & $-a_5+c_1+d_1$ & $-a_5+c_2+d_2$ & $-b_5+c_3+d_3$ & $-b_5+c_4+d_4$ & $1+a_5+b_5+c_5+d_5$ & $a_5+b_5+c_6+d_6$ \\ \hline
$\ket{-i}$ & $-a_6+c_1+d_1$ & $-a_6+c_2+d_2$ & $-b_6+c_3+d_3$ & $-b_6+c_4+d_4$ & $a_6+b_6+c_5+d_5$ & $1+a_6+b_6+c_6+d_6$ \\ \hline
\end{tabular}
}
\end{table}

\subsubsection{Degrees of Freedom}

In addition to replacing $I$ with $Z_+ + Z_-$, the decomposition allows for additional degrees of freedom, even when the set of observables and states is fixed. In particular, we observe that 
\begin{equation}
    X_++X_-=Y_++Y_-=Z_++Z_- = 
\begin{pmatrix}
1 & 0 \\
0 & 1
\end{pmatrix}=I,
\end{equation}
which is an extension of \eqref{eq:I2ZZ} to all Pauli bases. By linearity, we have
\begin{align}
    \text{Tr}[AX_+]+\text{Tr}[AX_-] &=\text{Tr}[AY_+]+\text{Tr}[AY_-] \nonumber \\&=\text{Tr}[AZ_+]+\text{Tr}[AZ_-]
\end{align}
for any $2\times2$ matrix $A$. This equation shows that we can adjust the coefficients $r_{i,j}$ for different $O_j$ while preserving the validity of the decomposition, provided the equation above is satisfied. In the tabular representation, without loss of generality, for any row $i$, we may subtract $a_i$ from $r_{i,1/2}$, subtract $b_i$ from $r_{i,3/4}$, and add $a_i + b_i$ to $r_{i,5/6}$, where $a_i$ and $b_i$ are arbitrary real numbers.

Similarly, the same reasoning can be applied to the initial states $\ket{\rho_i}$. Since
\begin{align}
    \ket{0}\bra{0} + \ket{1}\bra{1} &= \ket{+}\bra{+} + \ket{-}\bra{-} \nonumber\\
    &= \ket{+i}\bra{+i} + \ket{-i}\bra{-i},
\end{align}
any adjustments to $r_{i,j}$ that preserve these relations will not compromise the validity of the decomposition. The permissible modifications in the tabular representation are analogous to those described previously, with rows and columns exchanged and $a_i, b_i$ replaced by $c_j, d_j$. Incorporating all such adjustments results in \tabref{table:final}, which contains a total of 24 parameters that can be arbitrarily set. We refer to this parameterized form as a ``parameterized cut'' and denote the set of adjustable parameters as
\begin{equation}
    \mathbf{\theta} = \{a_i, b_i, c_j, d_j \mid i,j = 1, \dots, 6\},
\end{equation}
whose significance will become evident during postprocessing. By expanding the design space for postprocessing---without the need for prior knowledge or complex constructions of initial states---we can thus explore cut parameters that minimize the statistical error during postprocessing.

\subsubsection{Design Trade-offs for Overhead Reduction}
\label{sssec:tradeoff}
While using all six states as possible options provides a total of 24 degrees of freedom, it comes at the expense of increased postprocessing cost. Instead of the original cost of $4^K$ Kronecker products in the $\ell=4$ scenario, we require a total of $6^K$ Kronecker products for reconstruction, as $\ell=6$. This introduces a trade-off between postprocessing cost and the number of possible states, which is directly correlated with the size of the design space. To reduce the six states to four, such as $\{\ket{0}, \ket{1}, \ket{+}, \ket{+i}\}$, certain parameters must be modified, while others are fixed. Specifically, this can be achieved by setting
\begin{subequations}
    \begin{align}
        &(a_2, a_4, b_2, b_4) = (0,0,0,0),\\
        &(c_1, ..., c_6) = (-1/2,1/2,0,0,0,0),\\
        &(d_1, ..., d_6) = (0,0,-1/2,1/2,0,0),
    \end{align}
\end{subequations}
which reduces the degrees of freedom from 24 to 8.


Depending on the properties of the circuit cutting task, including number of cuts, circuit scale, etc., different design decisions can be made regarding this trade-off between design space and classical postprocessing overhead. Since the minimum states required for a decomposition is 4, we cannot trade more design space for overhead reduction. Therefore, in this work we will mainly discuss the $\ell=6$ option that have a larger design space and more classical overhead, and the $\ell=4$ option that have a smaller design space and less classical overhead. 
\subsection{ShotQC: Shot Distribution Optimization}
In the previous sections, we discussed how the output probability of the original circuit can be reconstructed using data obtained from simulating subcircuit configurations. However, not all data contribute equally in the reconstruction. For instance, consider the example shown in \eqref{eq:sump} through \eqref{eq:p2}. For a given outcome bitstring $\{a, b, c\}$, if $p^{(1)}(a0|X) - p^{(1)}(a1|X) = 0$, then the outcome probability $p(a, b, c)$ is theoretically independent of $p^{(2)}(bc|\ket{+})$, as $p^{(2)}(bc|\ket{+})$ is multiplied by $0$ in \eqref{eq:sump}. In such cases, fewer shots are required for the configuration that generates $p^{(2)}(bc|\ket{+})$, since its statistical uncertainty will be suppressed by the estimated value of $p^{(1)}(a0|X) - p^{(1)}(a1|X)$, which is expected to be close to $0$. Instead, the quantum shots can be reallocated to other configurations that have a greater impact on the final result. Our main objective is to explore this concept and calculate the optimal shot distribution across all subcircuit configurations using the limited information obtained from partial sampling.


\subsubsection{Calculating the Error Function}

Consider the general case where an $n$-qubit circuit is partitioned into $s$ disconnected subcircuits by cutting $K$ wires. Generalizing \eqref{eq:sump} through \eqref{eq:p2}, the outcome probability $p_a$ for a bitstring $a \in \{0, \dots, 2^n - 1\}$ in the original circuit is reconstructed via
\begin{equation}\label{eq:sump general}
p_a =  \sum_{i_1,\dots,i_K=1}^\ell \sum_{j_1,\dots,j_K=1}^\ell \left(\prod_{k=1}^K  r_{i_k,j_k}\right) \prod_{x=1}^s p^{(x)}\left(a; \{\rho_{i_{k'}}\}_{k'=1}^K, \{O_{j_{k''}}\}_{k''=1}^K\right),
\end{equation}
where each cut $k$ utilizes the general decomposition given in \eqref{eq:decomp generalized}, with cut parameters that may vary across different cuts. Here, $p^{(x)}\big(a; \{\rho_{i_{k'}}\}, \{O_{j_{k''}}\}\big)$ denotes the probability of obtaining the outcome in the subcircuit indexed by $x$ that matches the corresponding registers of bitstring $a$, evaluated under the condition that each cut input register indexed by $k'$ (if appearing in $x$) is prepared in state $\rho_{i_{k'}}$, each cut output register indexed by $k''$ (if appearing in $x$) yields the outcome corresponding to $O_{j_{k''}}$, and uncut input registers of $x$ remain as in the original circuit.

We can formally recast \eqref{eq:sump general} in terms of a deterministic function $F_a$:
\begin{equation}\label{eq:sump formal}
p_a = F_a\left(P\equiv \{P^{(e)}\}_{e\in E}, \Theta\right),
\end{equation}
where each configuration index $e \in E$ specifies a subcircuit indexed by $x$ with a measurement setting from $S_\text{meas}$ and a state preparation setting from $S_\text{prep}$ for each cut of $x$. The set $\Theta$ denotes all parameters as listed in \tabref{table:final} for parameterizing $\{r_{i_k,j_k}\}$, and $P^{(e)} = \{p^{(e)}_b \mid b = 0, \dots, 2^{\#(e)}-1\}$ represents the bitstring outcome probability distribution for configuration $e$ (with $\#(e)$ denoting the qubit count of configuration $e$).
 
Note that once a circuit cutting scheme with fixed parameters $\Theta$ is established, $F_a(P, \Theta)$ becomes a multilinear polynomial in $p^{(e)}_b$ with coefficients parameterized by $\Theta$, as directly determined by \eqref{eq:sump general} and illustrated in \eqref{eq:sump}--\eqref{eq:p2}. Importantly, this structural dependence is purely a function of the partition scheme and remains independent of the specific unitary functionality of the original circuit.

For each configuration $e$, the measured outcomes follow a multinomial distribution for $\{p^{(e)}_b\}$. Given $N_e$ execution shots for configuration $e$, the variance-covariance matrix $\Sigma^{(e)}$ describing the uncertainties and correlations among $\{p^{(e)}_b\}$ is given by
\begin{subequations}
\begin{align}
\Sigma^{(e)}_{bb} &\equiv \mathrm{Var}(p^{(e)}_b) = \frac{p^{(e)}_b(1 - p^{(e)}_b)}{N_e}, \\
\Sigma^{(e)}_{b_1 b_2; b_1 \neq b_2} &\equiv \mathrm{Cov}(p^{(e)}_{b_1}, p^{(e)}_{b_2}) = -\frac{p^{(e)}_{b_1} p^{(e)}_{b_2}}{N_e}.
\end{align}
\end{subequations}

Assuming $N_e$ is sufficiently large, the estimation error $\delta p^{(e)}_b$ remains small, justifying a first-order Taylor expansion of \eqref{eq:sump formal}:
\begin{equation}
\delta p_a \approx \sum_{e \in E} \sum_{b} \frac{\partial F_a\left(\{P^{(e')}\}_{e'\in E}, \Theta\right)}{\partial p^{(e)}_b} \delta p^{(e)}_b 
\equiv \sum_{e \in E} \sum_{b} h^{(e)}_{a,b} \, \delta p^{(e)}_b.
\end{equation}
Using standard error propagation, the covariance matrix $\Sigma$ for the reconstructed probabilities $p_a$ becomes
\begin{subequations}
\begin{align}
\Sigma_{a_1 a_2} &\approx \sum_{e \in E} \sum_{b_1, b_2} h^{(e)}_{a_1, b_1} h^{(e)}_{a_2, b_2} \Sigma^{(e)}_{b_1 b_2} \\
&= \sum_{e \in E} \frac{1}{N_e} g^{(e)}_{a_1 a_2}(P, \Theta),
\end{align}
\end{subequations}
where $g^{(e)}_{a_1 a_2}(P, \Theta)$ absorbs the shot-independent terms. Like $F_a\left(P, \Theta\right)$, the function $g^{(e)}_{a_1a_2}(P, \Theta)$ is directly determined once a circuit cutting scheme is established. Also note that $g^{(e)}_{a_1a_2}(P, \Theta)\geq0$. This yields the following linearized variance result.

\begin{lemma}\label{lemma:linear variance}
Given a sufficient number of shots $N_e$ performed for each configuration $e \in E$ with observed probability distributions $P^{(e)}$, the total error metric defined in \eqref{eq:error} for the original circuit is approximated by
\begin{equation}\label{eq:linear variance}
\mathrm{Err} \approx \sum_{e\in E} \frac{1}{N_e} f^{(e)}(P, \Theta),
\end{equation}
where $f^{(e)}(P, \Theta) = \sum_a g^{(e)}_{aa}(P, \Theta) \geq 0$ depends only on the empirical distributions $P\equiv\{P^{(e)}\}_{e\in E}$ and the cutting coefficients $\Theta$.
\end{lemma}

The post-processing step requires evaluating $\ell^K$ terms, where each term involves querying $O(6^{\#\text{meas}})$ observable probabilities per subcircuit ($\#\text{meas}$ denoting the number of measurement channels on that subcircuit). This yields an overall time complexity of $O(\ell^K \cdot 6^{\#\text{meas}})$. 

In contrast, evaluating the variance and covariance terms for each subcircuit requires processing only $\ell^{\#\mathrm{prep}} \cdot 3^{\#\mathrm{meas}}$ configurations ($\#\mathrm{prep}$ denoting the number of state preparation channels). Because each configuration requires $(2^{\#\mathrm{meas}})^2 = 4^{\#\mathrm{meas}}$ arithmetic operations, the overall computational effort per subcircuit scales as $O(12^{\#\mathrm{meas}} \cdot \ell^{\#\mathrm{prep}})$. 

Since every subcircuit satisfies $\#\mathrm{meas} + \#\mathrm{prep} \leq K$, we can bound the operational complexity of computing $f^{(e)}(P^{(e)}, \Theta)$ as follows:
\begin{equation}
\begin{aligned}
12^{\#\mathrm{meas}} \cdot \ell^{\#\mathrm{prep}} 
&\leq 12^{\#\mathrm{meas}} \cdot \ell^{K - \#\mathrm{meas}} \\
&= \ell^K \cdot \left(\frac{12}{\ell}\right)^{\#\mathrm{meas}} \\
&\leq \ell^K \cdot 3^{\#\mathrm{meas}},
\end{aligned}
\end{equation}
where the final inequality follows from $\ell \geq 4$ as established in Section~\ref{subsec:theory}. Consequently, the overall computational complexity of constructing the coefficient functions $f^{(e)}(P^{(e)}, \Theta)$ across all configurations $e \in E$ does not exceed that of the primary post-processing routine itself.

\subsubsection{Shot Distribution}\label{sec:shot distribution}

Because each configuration $e \in E$ is executed independently, the total sampling overhead---defined as the total number of shots across all configurations---is given by
\begin{equation}
N = \sum_{e\in E} N_e.
\end{equation}
Applying the Cauchy--Schwarz inequality to $N$ and \eqref{eq:linear variance} yields
\begin{equation}
\left(\sum_{e\in E} N_e\right) \left(\sum_{e\in E} \frac{1}{N_e} f^{(e)}(P, \Theta)\right) \ge \left(\sum_{e\in E} \sqrt{f^{(e)}(P, \Theta)}\right)^2,
\end{equation}
which establishes a lower bound on the product of the total shot count and the reconstruction error:
\begin{equation}\label{eq:shot-err}
N \cdot \mathrm{Err} \ge \left(\sum_{e\in E} \sqrt{f^{(e)}(P, \Theta)}\right)^2.
\end{equation}

For sufficiently large $N_e$, the empirical probability distributions $P^{(e)}$ converge to their ground-truth values, making $f^{(e)}(P, \Theta)$ asymptotically stable. Consequently, \eqref{eq:shot-err} demonstrates that the expected error function scales inversely with the total number of shots $N$.

Furthermore, the Cauchy--Schwarz inequality attains equality if and only if the constituent vectors are linearly dependent, which directly specifies the optimal allocation of shots across configurations. Specifically, for every $e \in E$, the optimal shot allocation $N_e$ must satisfy
\begin{equation}
\frac{N_e}{\sqrt{f^{(e)}(P, \Theta)}} = \mathrm{const}.
\end{equation}

Given an accurate \textit{a priori} estimate of each coefficient function $f^{(e)}(P, \Theta)$ (for instance, obtained via preliminary sampling runs), the optimal shot distribution subject to a constrained total budget $N$ is
\begin{equation}\label{eq:opt_distribute}
N_e = N \cdot \frac{\sqrt{f^{(e)}(P, \Theta)}}{\sum_{e'\in E} \sqrt{f^{(e')}(P, \Theta)}}.
\end{equation}
Expression \eqref{eq:opt_distribute} constitutes the central formula for our optimal shot allocation strategy.

\subsubsection{Optimizing the cost function}
As shown in \eqref{eq:shot-err}, the total number of shots $N$ required to achieve the same $\mathrm{Err}$ decreases if the loss function defined as 
\begin{equation}  
    \label{eq:loss}  
    \mathrm{Loss} = \sum_{e\in E} \sqrt{f^{(e)}(P, \Theta)}
\end{equation}
is further reduced.
Since the cut parameters $\Theta$ can be chosen arbitrarily, the loss function can be minimized by optimizing $\Theta$. Standard optimization techniques, such as stochastic gradient descent, the Adam optimizer, or simulated annealing, can be employed.

As discussed earlier in \ref{sssec:tradeoff}, the design space for $\Theta$ can be expanded at the cost of increased classical computation. A larger design space allows for a greater range of possible parameters, enabling further reduction of the loss function's minimum value.


There are additional considerations when calculating the coefficient functions $f^{(e)}(P^{(e)}, \Theta)$ for determining the optimal shot allocation or minimizing the loss function. To obtain the true values of the coefficient functions, the ground-truth probabilities $\hat{P}$ are required, or alternatively, $\hat{P}$ must be approximated using the estimated probabilities $P$. Consequently, the reliability of $P$ can affect the quality of the solutions.  
To address this issue, we divide the total number of shots into two stages: \emph{prior shots} and \emph{posterior shots}. The prior shots are performed first, with a predefined number of shots allocated to each circuit configuration to generate an initial estimate of $P$. This estimate is then used to calculate the optimal shot distribution for the posterior shots, which are executed in the subsequent step.

\subsubsection{Assessing the benefit of shot distribution optimization}
To evaluate the advantage of the optimized shot distribution, we compare it with the uniform shot distribution (i.e., the total shot budget is evenly distributed across all subcircuit configurations), leading to the following lemma.

\begin{lemma}
\label{lemma:scalability}
Let $N_0$ and $N_\mathrm{opt}$ denote the total number of shots required to achieve estimation error $\epsilon$ under the uniform and optimal shot distribution schemes, respectively.  
Given a set of cut parameters $\Theta$, ground-truth probabilities $P$, and a total of $|E|$ subcircuit configurations, the ratio between $N_0$ and $N_\mathrm{opt}$ is bounded as follows:
\begin{equation}
    1 \leq \frac{N_0}{N_\mathrm{opt}} \leq |E|,
\end{equation}
if the optimal shot distribution is allocated according to \eqref{eq:opt_distribute} with $P$.
\end{lemma}
\begin{proof}
The total number of shots required to achieve estimation error $\epsilon$ under the uniform shot distribution scheme, according to \eqref{eq:linear variance} with $N_e=N/|E|$, is
\begin{equation}
N_0 = \frac{|E|}{\epsilon} \sum_{e\in E} f^{(e)}(P,\Theta),
\end{equation}
while under the optimal shot distribution scheme, according to \eqref{eq:linear variance} with \eqref{eq:opt_distribute}, it is
\begin{equation}
N_\mathrm{opt}  \equiv \sum_{e\in E} N_e = \frac{1}{\epsilon} \left( \sum_{e\in E} \sqrt{f^{(e)}(P,\Theta)} \right)^2.
\end{equation}

We first establish the lower bound. Applying Sedrakyan's inequality gives
\begin{equation}
\left( \sum_{e\in E} 1 \right) \sum_{e\in E} f^{(e)}(P,\Theta) \geq \left( \sum_{e\in E} \sqrt{f^{(e)}(P,\Theta)} \right)^2,
\end{equation}
which implies
\begin{equation}\label{eq:R}
R := \frac{N_0}{N_\mathrm{opt}} 
= \frac{E \sum_{e\in E} f^{(e)}(P,\Theta)}{\left( \sum_{e\in E} \sqrt{f^{(e)}(P,\Theta)} \right)^2} \geq 1.
\end{equation}
Equality holds when $f^{(e)}(P,\Theta)$ is uniform across all $e$, a condition rarely satisfied in practical circuits. Thus, $R > 1$ in almost all realistic scenarios.

For the upper bound, we use the inequality
\begin{equation}
\left( \sum_{e\in E} \sqrt{f^{(e)}(P,\Theta)} \right)^2 \geq \sum_{e\in E} f^{(e)}(P,\Theta),
\end{equation}
which leads to
\begin{equation}\label{eq:R and E}
R\equiv \frac{N_0}{N_\mathrm{opt}} \leq |E|.
\end{equation}
Equality is attained if and only if $f^{(e)}(P,\Theta) \neq 0$ for only a single $e\in E$, a condition that is typically unattainable unless the circuit is purposely constructed.
In more realistic cases where the original $n$-qubit circuit is designed to yield only a limited number of outcomes and $K\ll n$, the values $f^{(e)}(P,\Theta)$ may remain close to zero for most $e\in E$, making the attainable upper bound much greater than unity.
\end{proof}

Since $0 \leq \sum_{e\in E} f^{(e)}(P, \Theta) \leq 1$, the total shot budget required to achieve an estimation variance $\epsilon$ scales as $\mathcal{O}(|E|/\epsilon)$ under a uniform allocation scheme. Under the optimal adaptive allocation scheme, this cost is reduced to $\mathcal{O}(|E|/(\epsilon R))$, where $1 \leq R \leq |E|$ according to Lemma~\ref{lemma:scalability}.

In general, the theoretical maximum $R = |E|$ is rarely attained, with $R$ often remaining well below $|E|$ across typical circuit families (see the benchmark summary in \tabref{tab:benchmark best results}). However, when the output probability distributions of the uncut circuit or its cut subcircuits exhibit high probability mass concentration (as detailed in \appref{sec:profile analysis}), the improvement factor $R$ increases significantly. In fact, $R$ can, in principle, scale exponentially with the number of cuts $K$---matching the scaling of $|E|$---for specially constructed edge cases (see \appref{sec:variance reduction scaling}).

Importantly, relation \eqref{eq:R and E} is derived under the assumption that $P$ represents the exact, ground-truth probability distributions. In practical adaptive shot optimization, $P$ is instead constructed empirically from prior sampling iterations. When ground-truth distributions are extremely concentrated, the empirical estimate of $P$ obtained from a small prior sample can diverge substantially from the true distribution. Under such conditions, the theoretical upper bound in \eqref{eq:R and E} can not only be violated, but exceeded by multiple orders of magnitude (see \tabref{tab:benchmark best results} and \appref{sec:variance reduction scaling}). Consequently, adaptive shot allocation under severe mass concentration can potentially yield enormous practical performance gains.

\begin{figure}
\centering
\includegraphics[width=0.75\columnwidth]{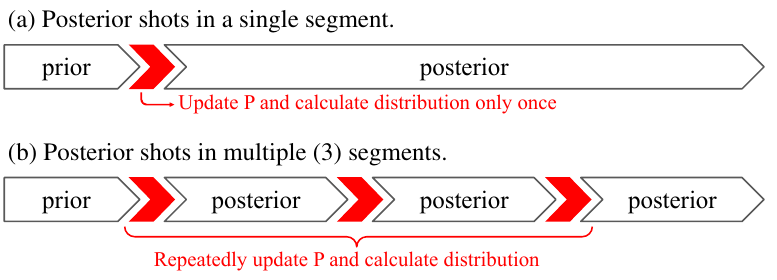}
\caption{Posterior shot segmentation. 
}
\label{fig:segment}
\end{figure}

\subsection{ShotQC: Additional Strategies}
\subsubsection{Posterior Shot Segmentation}

Calculating the optimal shot distribution requires a rough estimate $P$ of $\hat{P}$ to begin with. Naturally, a more accurate estimate $P$ leads to a more precise calculation of the optimal shot distribution, resulting in greater reductions in sampling overhead. However, improving the reliability of $P$ necessitates increasing the number of prior shots. Since prior and posterior shots share the same total shot budget, allocating more shots to the prior phase reduces the number of posterior shots available for distribution. Therefore, an appropriate balance must be struck when selecting the number of prior shots to maximize overall effectiveness.

On the other hand, potential statistical bias resulting from the prior shots may cause the shot distribution to deviate from the optimal allocation. To address this issue, \emph{posterior shot segmentation} is introduced. Instead of basing the shot distribution solely on the rough probabilities $P$ inferred from the prior shots, the posterior shots are divided into multiple stages. At each stage, the optimal shot distribution is refined using the updated estimates obtained from the previous stage. As illustrated in Figure \ref{fig:segment}, this segmentation helps mitigate the risk of prior shot bias misdirecting the overall shot distribution, ensuring a more reliable allocation process.


\subsubsection{Economical approach}\label{sec:economical}
Optimizing the cut parameters $\Theta$ can often be time-consuming due to the computationally intensive cost function that have to be calculated repeatedly with gradients for optimization. Hence, when the quantum circuit is too large, the parameters can simply be set to fixed values as described earlier. While not being able to exploit the benefits offered by cut parameter optimization, this \emph{economical} approach can still benefit from shot distribution optimization with minimum computational load. 

\subsubsection{Post-shot cut parameter optimization}
\label{sec:PSCP}

After all shots have been performed and all measurement outcomes have been collected, one can still perform a post-shot cut parameter (PSCP) optimization immediately before reconstructing the output probability distribution of the original circuit. By adjusting the cut parameters based on the collected data, PSCP can further reduce the estimation variance without requiring any additional quantum shots. The only cost is the additional classical computation for finding the optimized cut parameters.

In practice, as demonstrated by the benchmark results in \secref{sec:benchmark results}, the benefit of PSCP optimization varies substantially: it can range from nearly negligible to highly significant, depending on whether and how other optimization strategies have been applied, as well as on the circuit structure and the circuit-cutting scheme employed.


\section{Benchmarks}\label{sec:benchmarks}
In this section, we present the benchmark results of ShotQC.

\subsection{Overview and Setup}\label{sec:overview and setup}

To evaluate the performance of ShotQC across circuits with distinct output statistics, we consider five benchmark families: adder circuits~\cite{cuccaro2004newquantumripplecarryaddition}, approximate quantum Fourier transform (AQFT) circuits~\cite{PhysRevA.54.139}, QAOA circuits on Erd\"{o}s--R\'{e}nyi graphs and random regular graphs, and Google quantum supremacy circuits~\cite{Boixo2018}.

These families span a spectrum of output probability distributions with varying degrees of mass concentration: a \emph{single-output} profile for the adder, a \emph{uniform} profile for AQFT, and \emph{diffuse, non-uniform} profiles for the QAOA and quantum supremacy circuits. (See \appref{sec:profile analysis} for more details) Evaluating across these distinct profiles allows us to test ShotQC under qualitatively different sampling dynamics.



Using the CutQC library~\cite{Tang_2021}, we generate the benchmark instances and partition each circuit into two nearly balanced subcircuits by cutting 2 to 5 wires. To evaluate the performance of ShotQC in different scales, we select three circuit sizes (small, medium, and large) for each benchmark family. The complete benchmark suite---including total qubit counts, subcircuit sizes, and the number of cuts---is summarized in \tabref{tab:benchmark_circuits}.

\begin{table}[t]
\caption{Benchmark circuits used in the evaluation. Size entries denote $n\,(n_1, n_2)$, where $n$ is the total qubit count of the original circuit, and $n_1, n_2$ are the qubit counts of the two subcircuits.}
\label{tab:benchmark_circuits}
\centering
\begin{tabular}{l cc cc cc}
\toprule
& \multicolumn{2}{c}{\textbf{Small}} & \multicolumn{2}{c}{\textbf{Medium}} & \multicolumn{2}{c}{\textbf{Large}} \\
\cmidrule(lr){2-3} \cmidrule(lr){4-5} \cmidrule(lr){6-7}
\textbf{Family} & Size & Cuts & Size & Cuts & Size & Cuts \\
\midrule
Adder & 24\,(14, 12) & 2 & 32\,(16, 18) & 2 & 42\,(20, 24) & 2 \\
AQFT & 15\,(9, 10) & 4 & 24\,(18, 11) & 5 & 31\,(21, 15) & 5 \\
Erd{\"o}s--R\'enyi (QAOA) & 22\,(17, 9) & 4 & 27\,(19, 11) & 3 & 29\,(10, 22) & 3 \\
Regular (QAOA) & 20\,(11, 13) & 4 & 24\,(14, 14) & 4 & 30\,(18, 16) & 4 \\
Supremacy & 16\,(12, 8) & 4 & 24\,(11, 17) & 4 & 30\,(15, 20) & 5 \\
\bottomrule
\end{tabular}
\end{table}

Each benchmark instance is evaluated under various optimization settings by sweeping over shot allocation strategies, cut parameter optimization schedules, segmentation schemes, and wire-cutting decompositions.

With respect to cut parameter optimization, these configurations are categorized into three primary groups:
\begin{enumerate}
    \item \textbf{No cut optimization}: Performs only shot distribution optimization while omitting cut parameter optimization. This is the economical approach described in \secref{sec:economical} and serves as a baseline to measure the marginal benefit of optimizing cut parameters.
    
    \item \textbf{First-segment cut optimization}: Performs cut parameter optimization exactly once---immediately after prior shots are executed (i.e., directly preceding the first segment of posterior shots).
    
    \item \textbf{Per-segment cut optimization}: Re-evaluates and optimizes cut parameters dynamically before every individual segment of posterior shots.
\end{enumerate}

Within each group, specific optimization routines are further parameterized by:
\begin{enumerate}
    \item \textbf{Shot allocation strategy}: Whether shots are distributed \emph{uniformly} across all subcircuit configurations or allocated \emph{adaptively} per segment via the optimization detailed in \secref{sec:shot distribution}.
    \item \textbf{Segmentation depth}: The total number of posterior shot segments, as illustrated in \figref{fig:segment}.
    \item \textbf{Post-shot cut parameter (PSCP) optimization}: Whether the PSCP refinement step (\secref{sec:PSCP}) is enabled.
\end{enumerate}

Finally, we evaluate two distinct wire-cutting schemes:
\begin{itemize}
    \item \textbf{L4}: Replaces each cut wire with the $\ell=4$ channel decomposition as defined in \eqref{eq:cutqc}.
    \item \textbf{L6}: Replaces each cut wire with the expanded $\ell=6$ decomposition referenced following \eqref{eq:mit}.
\end{itemize}
Specifically, the set of state preparations $S_\mathrm{prep}$ iteratively injected into each cut wire is defined as:
\begin{equation}\label{eq:L4 and L6}
S_\mathrm{prep} = 
\begin{cases}
\{\ket{0}, \ket{1}, \ket{+}, \ket{+i}\}, & \text{for L4,} \\[2pt]
\{\ket{0}, \ket{1}, \ket{+}, \ket{-}, \ket{+i}, \ket{-i}\}, & \text{for L6.}
\end{cases}
\end{equation}

For a circuit partitioned into $s$ subcircuits indexed by $x$, with reference to the L4 scheme as the baseline wire-cutting scheme, a given total shot budget $N$ is conveniently expressed in terms of a rescaled reference factor $r$ as
\begin{equation}
N = r \sum_{x=1}^s 2^{n_x} 4^{p_x} 3^{m_x},
\end{equation}
where $n_x$ is the qubit count of subcircuit $x$, while $p_x$ and $m_x$ denote its number of preparation and measurement cuts, respectively. Here, $r$ serves as a normalized scaling factor representing the overall sampling density across the output bitstrings of all subcircuit configurations (based on the L4 scheme). Under adaptive allocation, a fixed total budget $N$ specified by a chosen $r$ is dynamically redistributed across subcircuit configurations $e \in E$ according to the ShotQC allocation rule.

Unless stated otherwise, we set $r=1$ for all benchmark evaluations. For adaptive allocation schemes, the prior shot budget is initialized to the shot count corresponding to $r=0.2$. All cut parameter optimizations are executed in double precision using the Adam optimizer with a learning rate of $0.1$ over $100$ optimization steps per phase. The remaining budget for posterior shots is partitioned and adaptively reallocated based on the specified number of posterior sampling segments.

All experiments are conducted on an ideal, noiseless quantum simulator. The software stack is implemented in \texttt{Python}, leveraging \texttt{Qiskit} and \texttt{Qiskit Aer} for quantum circuit execution, \texttt{PyTorch} for variance estimation and shot allocation, and Adam for cut parameter optimization. Computations are run on a server equipped with an Intel Xeon w7-2495X CPU and NVIDIA RTX A4000 GPUs, using a single GPU per instance for tensor contractions and gradient updates.

\subsection{Benchmark Results}\label{sec:benchmark results}

For each circuit instance listed in \tabref{tab:benchmark_circuits}, we evaluate $12$ distinct optimization settings---as detailed in \secref{sec:overview and setup}, configured with segmentation depths of $1$ and $5$---across both the L4 and L6 wire-cutting schemes, using a total shot budget of $N_{r=1}$.
The benchmark results for the small, medium, and large instances of the five circuit families (Adder, AQFT, Erd{\"o}s--R{\'e}nyi QAOA, Random Regular QAOA, and Quantum Supremacy) are summarized in \tabref{tab:shotqc_small_results}, \tabref{tab:shotqc_medium_results}, and \tabref{tab:shotqc_large_results}, respectively. For completeness and to facilitate visual comparison across configurations, the corresponding detailed plots---including statistical error bars---are provided in \appref{sec:sup_data}.


\begin{table}
\centering
\caption{Estimated sampling variance $\mathrm{Err}$ across different optimization settings for the small circuits.
}
\label{tab:shotqc_small_results}
\setlength{\tabcolsep}{2.9pt}
\renewcommand{\arraystretch}{1.12}
\scriptsize
\begin{tabular}{lcc cc cc cc cc cc}
\toprule
& & & \multicolumn{2}{c}{\textbf{Adder}} & \multicolumn{2}{c}{\textbf{AQFT}} & \multicolumn{2}{c}{\textbf{Erd\"{o}s--R\'{e}nyi}} & \multicolumn{2}{c}{\textbf{Regular}} & \multicolumn{2}{c}{\textbf{Supremacy}} \\
& & & \multicolumn{2}{c}{Fig.~\ref{fig:adder_small}} & \multicolumn{2}{c}{Fig.~\ref{fig:aqft_small}} & \multicolumn{2}{c}{Fig.~\ref{fig:erdos_small}} & \multicolumn{2}{c}{Fig.~\ref{fig:regular_small}} & \multicolumn{2}{c}{Fig.~\ref{fig:supremacy_small}} \\
\cmidrule(lr){4-5} \cmidrule(lr){6-7} \cmidrule(lr){8-9} \cmidrule(lr){10-11} \cmidrule(lr){12-13}
\textbf{Allocation} & \textbf{Segments} & \textbf{PSCP} & L4 & L6 & L4 & L6 & L4 & L6 & L4 & L6 & L4 & L6 \\
\midrule
\multicolumn{13}{l}{\textit{No cut optimization}} \\
Uniform & -- & -- & 0.234 & 91.7 & 5.49 & 6.79 & 5.21 & 1.15 & 4.2 & 1.24 & 1.91 & 0.759 \\
Uniform & -- & $\checkmark$ & 0.148 & 0.0802 & 1.47 & 1.08 & 3.46 & 0.355 & 3.07 & 0.465 & 1.3 & 0.283 \\
Adaptive & 1 & -- & 0.18 & 4.07 & 3.72 & 1.45 & 1.22 & 0.379 & 3.66 & 1.06 & 1.72 & 0.601 \\
Adaptive & 1 & $\checkmark$ & \textbf{0.0826} & 0.282 & 0.975 & 0.45 & 0.674 & 0.157 & 2.53 & 0.537 & 1.13 & 0.306 \\
Adaptive & 5 & -- & 0.125 & 3.71 & 3.71 & 1.15 & 0.749 & 0.297 & 3.57 & 1.06 & 1.72 & 0.585 \\
Adaptive & 5 & $\checkmark$ & 0.0849 & 0.234 & 0.918 & 0.333 & 0.454 & 0.132 & 2.51 & 0.517 & 1.16 & 0.294 \\
\addlinespace[2pt]
\multicolumn{13}{l}{\textit{First-segment cut optimization}} \\
Adaptive & 1 & -- & 0.268 & 0.195 & 0.485 & 0.197 & 0.683 & 0.147 & 2.38 & \textbf{0.461} & \textbf{1.06} & 0.265 \\
Adaptive & 1 & $\checkmark$ & 0.105 & 0.0718 & 0.454 & 0.18 & 0.681 & 0.147 & 2.38 & \textbf{0.461} & \textbf{1.06} & 0.265 \\
Adaptive & 5 & -- & 0.205 & 0.102 & 0.38 & 0.132 & 0.449 & 0.121 & \textbf{2.31} & 0.464 & 1.07 & 0.261 \\
Adaptive & 5 & $\checkmark$ & 0.0999 & \textbf{0.0517} & 0.353 & 0.117 & 0.442 & 0.12 & \textbf{2.31} & 0.464 & 1.07 & \textbf{0.26} \\
\addlinespace[2pt]
\multicolumn{13}{l}{\textit{Per-segment cut optimization}} \\
Adaptive & 5 & -- & 0.111 & 0.0963 & 0.32 & 0.0896 & 0.433 & 0.118 & \textbf{2.31} & 0.463 & \textbf{1.06} & 0.261 \\
Adaptive & 5 & $\checkmark$ & 0.0846 & 0.0532 & \textbf{0.315} & \textbf{0.0884} & \textbf{0.431} & \textbf{0.117} & \textbf{2.31} & 0.463 & \textbf{1.06} & 0.261 \\
\cmidrule(lr){4-5} \cmidrule(lr){6-7} \cmidrule(lr){8-9} \cmidrule(lr){10-11} \cmidrule(lr){12-13}
\multicolumn{3}{c}{\textbf{Variance scale}} &  \multicolumn{2}{c}{$\times 10^{-6}$} & \multicolumn{2}{c}{$\times 10^{-4}$} & \multicolumn{2}{c}{$\times 10^{-6}$} & \multicolumn{2}{c}{$\times 10^{-5}$} & \multicolumn{2}{c}{$\times 10^{-4}$} \\
\bottomrule
\end{tabular}
\end{table}

\begin{table}
\centering
\caption{Estimated sampling variance $\mathrm{Err}$ across different optimization settings for the medium circuits.
}
\label{tab:shotqc_medium_results}
\setlength{\tabcolsep}{2.9pt}
\renewcommand{\arraystretch}{1.12}
\scriptsize
\begin{tabular}{lcc cc cc cc cc cc}
\toprule
& & & \multicolumn{2}{c}{\textbf{Adder}} & \multicolumn{2}{c}{\textbf{AQFT}} & \multicolumn{2}{c}{\textbf{Erd\"{o}s--R\'{e}nyi}} & \multicolumn{2}{c}{\textbf{Regular}} & \multicolumn{2}{c}{\textbf{Supremacy}} \\
& & & \multicolumn{2}{c}{Fig.~\ref{fig:adder_medium}} & \multicolumn{2}{c}{Fig.~\ref{fig:aqft_medium}} & \multicolumn{2}{c}{Fig.~\ref{fig:erdos_medium}} & \multicolumn{2}{c}{Fig.~\ref{fig:regular_medium}} & \multicolumn{2}{c}{Fig.~\ref{fig:supremacy_medium}} \\
\cmidrule(lr){4-5} \cmidrule(lr){6-7} \cmidrule(lr){8-9} \cmidrule(lr){10-11} \cmidrule(lr){12-13}
\textbf{Allocation} & \textbf{Segments} & \textbf{PSCP} & L4 & L6 & L4 & L6 & L4 & L6 & L4 & L6 & L4 & L6 \\
\midrule
\multicolumn{13}{l}{\textit{No cut optimization}} \\
Uniform & -- & -- & 0.0946 & 143 & 39.2 & 3.18 & 3.49 & 0.586 & 8.4 & 2.84 & 2.53 & 1.81 \\
Uniform & -- & $\checkmark$ & 0.0857 & 0.0432 & 4.34 & 0.338 & 3.05 & 0.21 & 7.03 & 0.543 & 1.82 & 0.689 \\
Adaptive & 1 & -- & 0.0788 & 17.5 & 10.4 & 1.32 & 0.526 & 0.175 & 7.41 & 1.16 & 0.791 & 0.318 \\
Adaptive & 1 & $\checkmark$ & 0.0698 & 0.431 & 1.09 & 0.218 & 0.369 & 0.085 & 6.3 & 0.297 & 0.612 & 0.167 \\
Adaptive & 5 & -- & 0.0487 & 17.3 & 6.9 & 1.13 & 0.268 & 0.135 & 7.29 & 1.04 & 0.687 & 0.252 \\
Adaptive & 5 & $\checkmark$ & \textbf{0.046} & 0.304 & 0.739 & 0.188 & 0.196 & 0.0723 & 6.01 & 0.262 & 0.553 & 0.131 \\
\addlinespace[2pt]
\multicolumn{13}{l}{\textit{First-segment cut optimization}} \\
Adaptive & 1 & -- & 0.199 & 0.135 & 0.175 & 0.0708 & 0.352 & 0.086 & 5.82 & 0.206 & 0.612 & 0.146 \\
Adaptive & 1 & $\checkmark$ & 0.135 & 0.0393 & 0.17 & 0.0627 & 0.352 & 0.0859 & 5.81 & 0.204 & 0.61 & 0.146 \\
Adaptive & 5 & -- & 0.0573 & 0.195 & 0.148 & 0.0607 & 0.192 & 0.0713 & 5.57 & 0.163 & 0.546 & 0.115 \\
Adaptive & 5 & $\checkmark$ & 0.0522 & \textbf{0.0299} & \textbf{0.144} & 0.053 & 0.189 & 0.0705 & 5.56 & 0.159 & \textbf{0.545} & 0.115 \\
\addlinespace[2pt]
\multicolumn{13}{l}{\textit{Per-segment cut optimization}} \\
Adaptive & 5 & -- & 0.0807 & 0.669 & 0.148 & 0.0497 & 0.186 & 0.07 & 5.56 & 0.154 & \textbf{0.545} & 0.114 \\
Adaptive & 5 & $\checkmark$ & 0.0682 & 0.0305 & \textbf{0.144} & \textbf{0.0482} & \textbf{0.184} & \textbf{0.069} & \textbf{5.55} & \textbf{0.153} & \textbf{0.545} & \textbf{0.113} \\
\cmidrule(lr){4-5} \cmidrule(lr){6-7} \cmidrule(lr){8-9} \cmidrule(lr){10-11} \cmidrule(lr){12-13}
\multicolumn{3}{c}{\textbf{Variance scale}} & \multicolumn{2}{c}{$\times10^{-8}$} & \multicolumn{2}{c}{$\times 10^{-7}$} & \multicolumn{2}{c}{$\times10^{-7}$} & \multicolumn{2}{c}{$\times10^{-6}$} & \multicolumn{2}{c}{$\times10^{-6}$} \\
\bottomrule
\end{tabular}
\end{table}

\begin{table}
\centering
\caption{Estimated sampling variance $\mathrm{Err}$ across different optimization settings for the large circuits.
}
\label{tab:shotqc_large_results}
\setlength{\tabcolsep}{2.9pt}
\renewcommand{\arraystretch}{1.12}
\scriptsize
\begin{tabular}{lcc cc cc cc cc cc}
\toprule
& & & \multicolumn{2}{c}{\textbf{Adder}} & \multicolumn{2}{c}{\textbf{AQFT}} & \multicolumn{2}{c}{\textbf{Erd\"{o}s--R\'{e}nyi}} & \multicolumn{2}{c}{\textbf{Regular}} & \multicolumn{2}{c}{\textbf{Supremacy}} \\
& & & \multicolumn{2}{c}{Fig.~\ref{fig:adder_large}} & \multicolumn{2}{c}{Fig.~\ref{fig:aqft_large}} & \multicolumn{2}{c}{Fig.~\ref{fig:erdos_large}} & \multicolumn{2}{c}{Fig.~\ref{fig:regular_large}} & \multicolumn{2}{c}{Fig.~\ref{fig:supremacy_large}} \\
\cmidrule(lr){4-5} \cmidrule(lr){6-7} \cmidrule(lr){8-9} \cmidrule(lr){10-11} \cmidrule(lr){12-13}
\textbf{Allocation} & \textbf{Segments} & \textbf{PSCP} & L4 & L6 & L4 & L6 & L4 & L6 & L4 & L6 & L4 & L6 \\
\midrule
\multicolumn{13}{l}{\textit{No cut optimization}} \\
Uniform & -- & --         & 0.00791 & 224     & 30.8  & 3.16   & 9.85 & 16    & 3.29  & 1.64  & 8.1  & 1.02  \\
Uniform & -- & $\checkmark$ & 0.00791 & 0.00467 & 4.23  & 0.563  & 6.65 & 5.3   & 1.75  & 0.256 & 5.59 & 0.284 \\
Adaptive & 1 & --        & 0.00505 & 29.1    & 11.6  & 1.6    & 4.31 & 2.52  & 2.02  & 1.07  & 3.27 & 0.623 \\
Adaptive & 1 & $\checkmark$ & 0.00505 & 0.0733  & 2.07  & 0.56   & 4.26 & 1.23  & 1.07  & 0.264 & 2.05 & 0.256 \\
Adaptive & 5 & --        & 0.0023  & 29.2    & 8.98  & 1.47   & 3.82 & 2.02  & 1.88  & 1.04  & 2.55 & 0.584 \\
Adaptive & 5 & $\checkmark$ & 0.0023  & 0.0912  & 1.78  & 0.526  & 3.79 & 1.04  & 0.986 & 0.259 & 1.6  & 0.233 \\
\addlinespace[2pt]
\multicolumn{13}{l}{\textit{First-segment cut optimization}} \\
Adaptive & 1 & --        & 0.00843 & 0.236   & 0.202 & 0.09   & 4.04 & 0.828 & 0.931 & 0.117 & 1.82 & 0.211 \\
Adaptive & 1 & $\checkmark$ & 0.00843 & 0.0031  & 0.198 & 0.0868 & 4.03 & 0.825 & 0.931 & 0.116 & 1.82 & 0.211 \\
Adaptive & 5 & --        & 0.00301 & 0.501   & 0.183 & 0.0759 & 3.85 & 0.604 & 0.877 & 0.11  & \textbf{1.39} & 0.197 \\
Adaptive & 5 & $\checkmark$ & 0.00301 & 0.0055  & 0.179 & 0.0735 & 3.83 & 0.587 & 0.875 & 0.11  & \textbf{1.39} & 0.197 \\
\addlinespace[2pt]
\multicolumn{13}{l}{\textit{Per-segment cut optimization}} \\
Adaptive & 5 & --         & \textbf{0.00162} & 0.439   & 0.181 & 0.0606 & 3.62 & 0.551 & 0.874 & 0.108 & \textbf{1.39} & 0.197 \\
Adaptive & 5 & $\checkmark$ & \textbf{0.00162} & \textbf{0.00265} & \textbf{0.179} & \textbf{0.0598} & \textbf{3.59} & \textbf{0.535} & \textbf{0.872} & \textbf{0.107} & \textbf{1.39} & \textbf{0.196} \\
\cmidrule(lr){4-5} \cmidrule(lr){6-7} \cmidrule(lr){8-9} \cmidrule(lr){10-11} \cmidrule(lr){12-13}
\multicolumn{3}{c}{\textbf{Variance scale}} &  \multicolumn{2}{c}{$\times 10^{-10}$} & \multicolumn{2}{c}{$\times 10^{-9}$} & \multicolumn{2}{c}{$\times 10^{-8}$} & \multicolumn{2}{c}{$\times 10^{-7}$} & \multicolumn{2}{c}{$\times 10^{-8}$} \\
\bottomrule
\end{tabular}
\end{table}

Each table reports the estimated sampling variance, $\mathrm{Err}$ as defined in \eqref{eq:linear variance}, obtained from simulated sampling across different optimization settings for both the L4 and L6 wire-cutting schemes. To account for statistical fluctuations in variance estimation, each tabular entry represents the sample mean over $5$ independent runs under the same total shot budget $N_{r=1}$. Numerical values are scaled by family-specific multiplication factors indicated at the bottom of each column. Boldfaced entries highlight the optimal (minimum) variance achieved within each column, including ties. The topmost entry of each sub-table presents the baseline setting (uniform shot allocation without optimization), corresponding to the standard implementations in \cite{Tang_2021} for L4 and \cite{Peng_2020} for L6, respectively.

\begin{table}
\caption{Improvement factors $R := \mathrm{Err}_0 / \mathrm{Err}_\mathrm{opt}$ achieved under the best optimization setting for each benchmark instance. Each entry $R\ (|E|)$ reports the maximal $R$ alongside the total number of subcircuit configurations $|E|$. Checkmarks (\checkmark) highlight whether the L4 or L6 scheme yields the smaller absolute variance under its respective optimal setting. Underlined $R$ values indicate instances where the optimal performance is achieved not by the most aggressive optimization setting. Overlined $|E|$ values highlight instances where the theoretical bound $R \leq |E|$ is violated.}
\label{tab:benchmark best results}
\centering
\scalebox{0.9}{
\begin{tabular}{l cc cc cc}
\toprule
& \multicolumn{2}{c}{\textbf{Small}} & \multicolumn{2}{c}{\textbf{Medium}} & \multicolumn{2}{c}{\textbf{Large}} \\
& \multicolumn{2}{c}{\tabref{tab:shotqc_small_results}} & \multicolumn{2}{c}{\tabref{tab:shotqc_medium_results}} & \multicolumn{2}{c}{\tabref{tab:shotqc_large_results}} \\
\cmidrule(lr){2-3} \cmidrule(lr){4-5} \cmidrule(lr){6-7}
\textbf{Family} & L4 & L6 & L4 & L6 & L4 & L6 \\
\midrule
Adder & \underline{2.83}\,(25) & \checkmark\underline{1773.69}\,($\overline{45}$) & \underline{2.06}\,(25) & \checkmark\underline{4782.61}\,($\overline{45}$) & \checkmark4.88\,(25) & 84528.33\,($\overline{45}$)\\
AQFT & 17.43\,(337) & \checkmark76.81\,(1377) & 272.22\,(1267) & \checkmark65.98\,(8019) & 172.07\,(1267) & \checkmark52.84\,(8019) \\
Erd{\"o}s--R{\'e}nyi & 12.09\,(337) & \checkmark9.83\,(1377) & 18.97\,(1728) & \checkmark8.49\,(5832) & 2.78\,(91) & \checkmark29.91\,(243) \\
Regular & 1.82\,(337) & \checkmark\underline{2.69}\,(1377) & 1.51\,(337) & \checkmark18.56\,(1377) & 3.77\,(337) & \checkmark15.33\,(1377) \\
Supremacy & 1.80\,(337) & \checkmark\underline{2.92}\,(1377) & 4.64\,(337) & \checkmark16.02\,(1377) & 5.82\,(1267) & \checkmark5.20\,(8019) \\
\bottomrule
\end{tabular}
}
\end{table}

The performance gain of a given optimization setting is quantified by the improvement factor relative to the baseline, defined as $R := \mathrm{Err}_0 / \mathrm{Err}_\mathrm{opt}$, where $\mathrm{Err}_0$ and $\mathrm{Err}_\mathrm{opt}$ denote the variance under the baseline and optimized settings, respectively. Note that this factor directly corresponds to $R$ in \eqref{eq:R}, representing the ratio of required shots to achieve a target precision $\epsilon$. The maximum improvement factors $R$ achieved across all $12$ optimization settings for each benchmark instance are summarized in \tabref{tab:benchmark best results}.

As shown in \tabref{tab:benchmark best results}, the optimal improvement factor varies substantially---ranging from $1.51$ to $84528.33$---depending on the underlying circuit profile, problem size, cut count, and wire-cutting scheme (L4 vs.\ L6). Notably, except for some small circuits and the medium Adder, the maximum performance gain across all benchmark families is consistently attained by the most aggressive optimization strategy (i.e., adaptive allocation combined with per-segment cut parameter optimization over $5$ segments and active PSCP, corresponding to the bottommost rows of \tabref{tab:shotqc_small_results}, \tabref{tab:shotqc_medium_results}, and \tabref{tab:shotqc_large_results}). 
This confirms that shot distribution optimization, cut parameter optimization, and PSCP deliver distinct, complementary benefits. Furthermore, increasing the segmentation depth from $1$ to $5$ facilitates more frequent adaptive updates, generally enhancing the efficacy of both shot distribution and cut parameter optimization.

However, the magnitude of variance reduction provided by each individual optimization component depends heavily on its synergy with other active techniques as well as the underlying circuit structure, cut configuration, and decomposition scheme.

For example, for the large AQFT, Regular, and Supremacy instances in \tabref{tab:shotqc_large_results} (see also \figref{fig:aqft_large}, \figref{fig:regular_large}, and \figref{fig:supremacy_large}, respectively), applying PSCP alone within the \emph{no cut optimization} group yields significant improvements. Conversely, within the \emph{first-segment} and \emph{per-segment cut optimization} groups, PSCP offers only marginal gains, indicating that \textit{en route} cut parameter optimization renders PSCP largely redundant. The large Erdős–Rényi instance exhibits a similar trend (see also \figref{fig:erdos_large}), with one notable exception: under \emph{no cut optimization}, the marginal benefit of PSCP diminishes rapidly for L4 once adaptive allocation is introduced, whereas it remains substantial for L6. This suggests that the broader parameter space of L6 offers greater capacity for cut parameter optimization, whereas the constrained parameter space of L4 allows shot distribution optimization to dominate performance gains. 


Importantly, as indicated by the checkmarks in \tabref{tab:benchmark best results}, the minimum absolute variance ($\mathrm{Err}_\mathrm{opt}$) is achieved by L6 rather than L4 in nearly all instances (with the large Adder circuit as the sole exception). Even in the unoptimized baseline setting (topmost rows in the tables), L6 frequently outperforms L4, as seen in \tabref{tab:shotqc_small_results}, \tabref{tab:shotqc_medium_results}, and \tabref{tab:shotqc_large_results} (see also the corresponding figures for the comparison). 
This result is counterintuitive at first glance, given that L4 was specifically designed to reduce sampling overhead by reducing the required state preparations from $6$ to $4$ per cut wire, as described in \eqref{eq:L4 and L6}.

The conventional intuition that L4 should outperform L6 relies on the implicit assumption that the probability distributions of the cut subcircuits and the original uncut circuit are approximately uniform. In practice, realistic quantum circuits exhibit varying degrees of probability mass concentration (see \appref{sec:profile analysis}). The concentration profile of a cut subcircuit depends directly on the specific preparation state injected at the cut locations. Because L6 iterates over a richer basis of preparation states than L4, it explores a broader state space and is more likely to sample configurations with higher mass concentration. When reconstructing the target state distribution, standard error propagation leads to a lower overall variance $\mathrm{Err}$ for L6, as higher probability concentration reduces statistical sampling uncertainty. This structural advantage of L6 is amplified further when shot distribution and cut parameter optimizations are applied, as both routines actively exploit probability mass concentration.

This insight suggests that alternative wire-cutting decompositions beyond L4 and L6 could offer even greater efficiency, even if they mandate iterating over a larger set of preparation states. Specifically, injecting states outside those in \eqref{eq:L4 and L6} may further exploit structural imbalances in subcircuit outputs. However, systematically evaluating the space of all valid wire-cutting schemes remains an open challenge that warrants future investigation.

Among the evaluated benchmark families, the Adder circuit exhibits several unique characteristics. First, for a fixed shot budget $N_{r=1}$, the statistical fluctuation in the variance estimate $\mathrm{Err}$ is significantly larger than in any other family (as evidenced by the visible error bars in \figref{fig:adder_large}, \figref{fig:adder_small}, and \figref{fig:adder_medium}). Second, in small and medium Adder instances, the highest improvement factor $R$ is achieved prior to the most aggressive optimization configuration (highlighted in bold in \tabref{tab:benchmark best results}), indicating that overly frequent optimization updates can become counterproductive. Third, while $R$ remains well below the theoretical subcircuit configuration count $|E|$ for most benchmarks---consistent with the upper bound $R \leq |E|$ discussed after \eqref{eq:R and E}---this bound is dramatically violated in the L6 Adder instances (highlighted by boldfaced $|E|$ values in \tabref{tab:benchmark best results}).

These anomalies stem from the extreme probability mass concentration of the Adder circuit, whose output distribution is concentrated entirely on a \emph{single} computational basis state, with its cut subcircuits retaining similarly highly concentrated profiles (see \appref{sec:profile analysis}). When reconstructing a strictly 1-hot target distribution from finite subcircuit samples, empirical probability estimates $P^{(e)}$ frequently yield small non-zero or even unphysical negative values for zero-probability outcomes. This causes substantial deviations from ground-truth probabilities and inflates the variance of variance estimation under finite sampling. When $P^{(e)}$ estimates are noisy, frequent adaptive updates to shot distributions and cut parameters based on intermediate sample estimates can lead to suboptimal reallocations (overkill). Once the prior shot allocation is sufficiently large to establish reliable initial probability estimates---as in the large Adder instance—aggressive optimization resumes its high performance.

Furthermore, because the theoretical bound $R \leq |E|$ in \eqref{eq:R and E} is derived assuming exact ground-truth subcircuit probabilities, severe sampling deviations can cause $R$ to vastly exceed $|E|$. In regimes with extreme mass concentration, uniform shot allocation performs exceptionally poorly, enabling ShotQC's adaptive reallocation to achieve massive variance reductions ($R > 80,000$). This result is highly relevant for practical quantum algorithms (e.g., arithmetic, search, and domain-specific state preparation), which naturally output highly concentrated probability distributions. In \appref{sec:variance reduction scaling}, we demonstrate that more massive improvement factors are achievable for Bernstein--Vazirani circuits under a synthetic wire-cutting configuration designed specifically for stress-testing.

In summary, our benchmark results confirm that each optimization component in ShotQC delivers tangible variance reductions, with more aggressive multi-segment strategies generally yielding the highest performance gains. The improvement factor $R$ can reach extraordinary levels ($>10^4$) when evaluating circuits with concentrated output profiles. Counterintuitively, the L6 wire-cutting scheme routinely outperforms L4 by leveraging state preparations that reveal and exploit probability mass concentration. While overall performance gains are governed by a complex interplay between circuit structure, cut choices, and optimization settings, ShotQC consistently provides robust sampling overhead reductions across all tested regimes. In the following subsection, we analyze the impact of varying ShotQC configurations in greater detail.

\subsection{Impact of Varying ShotQC Configurations}

We next evaluate the impact of varying three key configuration parameters of ShotQC: the total shot budget, the prior-sampling ratio, and the segmentation depth. To isolate the effects of shot distribution optimization in synergy with \emph{en route} cut parameter optimization, post-shot cut parameter (PSCP) optimization is disabled in these experiments. Because circuit-size scaling was evaluated in the comprehensive benchmark suite, we select one representative circuit instance from each benchmark family for this parametric analysis.
As before, each value of $\mathrm{Err}$ and its corresponding standard error (SE) are computed from $5$ independent sampling runs executed under identical configurations.

The results for the small AQFT and Quantum Supremacy instances are presented in \figref{fig:aqft_15} and \figref{fig:supremacy_16}, respectively. Corresponding results for the medium Adder, Erd{\"o}s--R{\'e}nyi QAOA, and Random Regular QAOA instances are shown in \figref{fig:adder_32}, \figref{fig:erdos_27}, and \figref{fig:regular_24}, respectively. In these plots, the visual error bars represent magnified standard errors to ensure legibility across all panels.

\begin{figure}
\centering
\includegraphics[width=\linewidth]{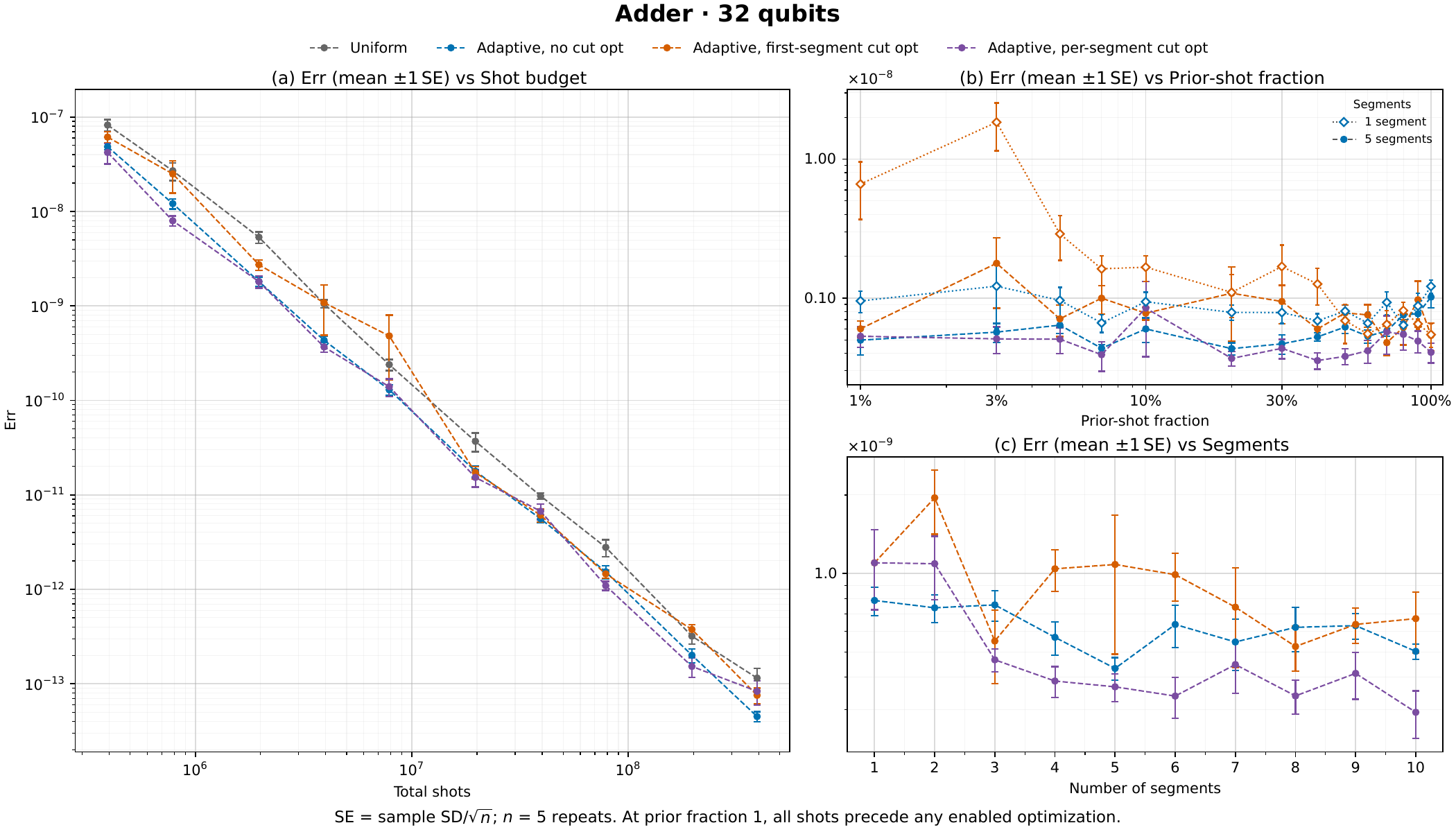}
\caption{Impact of varying ShotQC parameters for the 32-qubit Adder circuits.}
\label{fig:adder_32}
\end{figure}

\begin{figure}
\centering
\includegraphics[width=\linewidth]{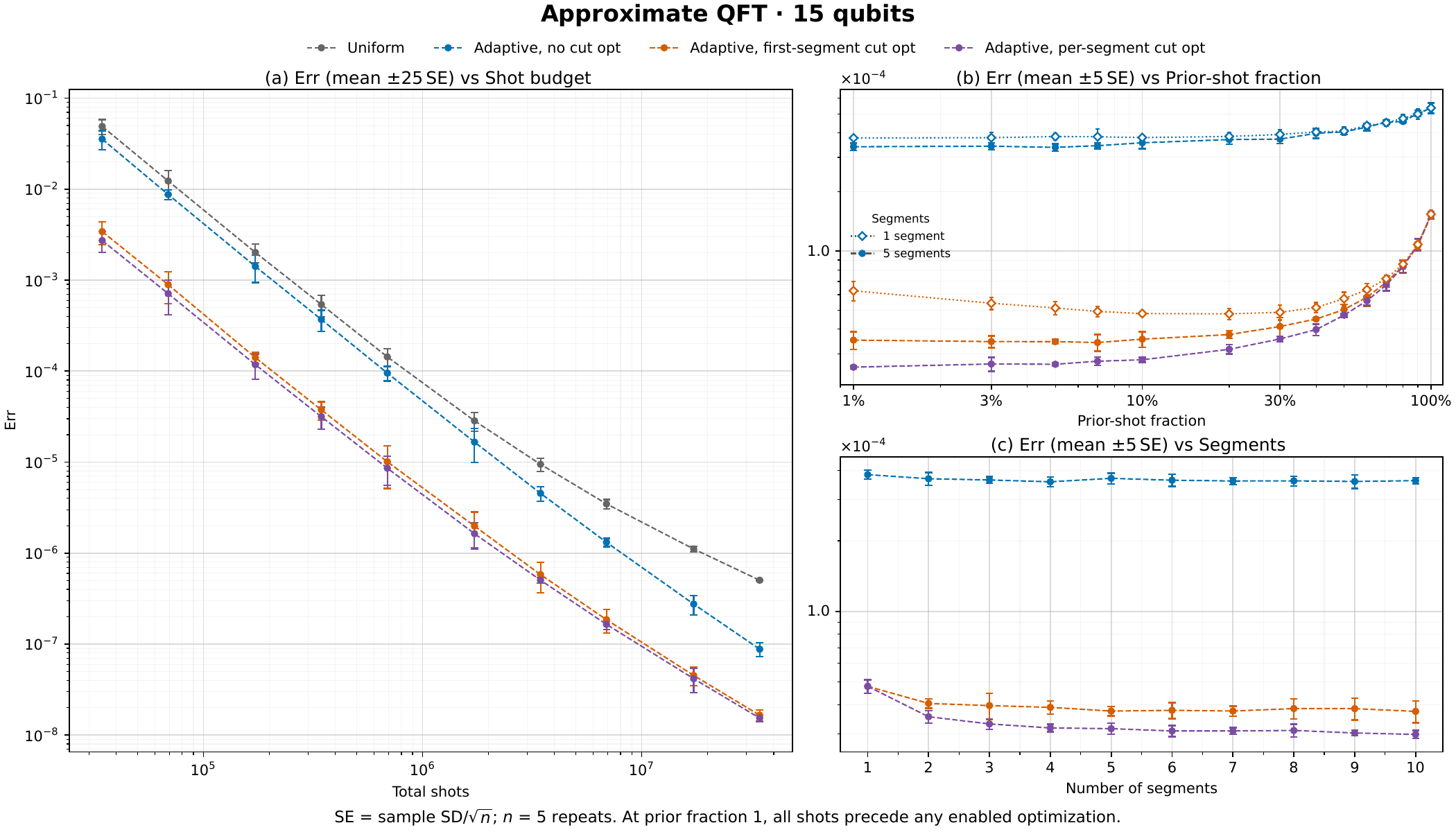}
\caption{Impact of varying ShotQC parameters for the 15-qubit AQFT circuit.}
\label{fig:aqft_15}
\end{figure}

\begin{figure}
\centering
\includegraphics[width=\linewidth]{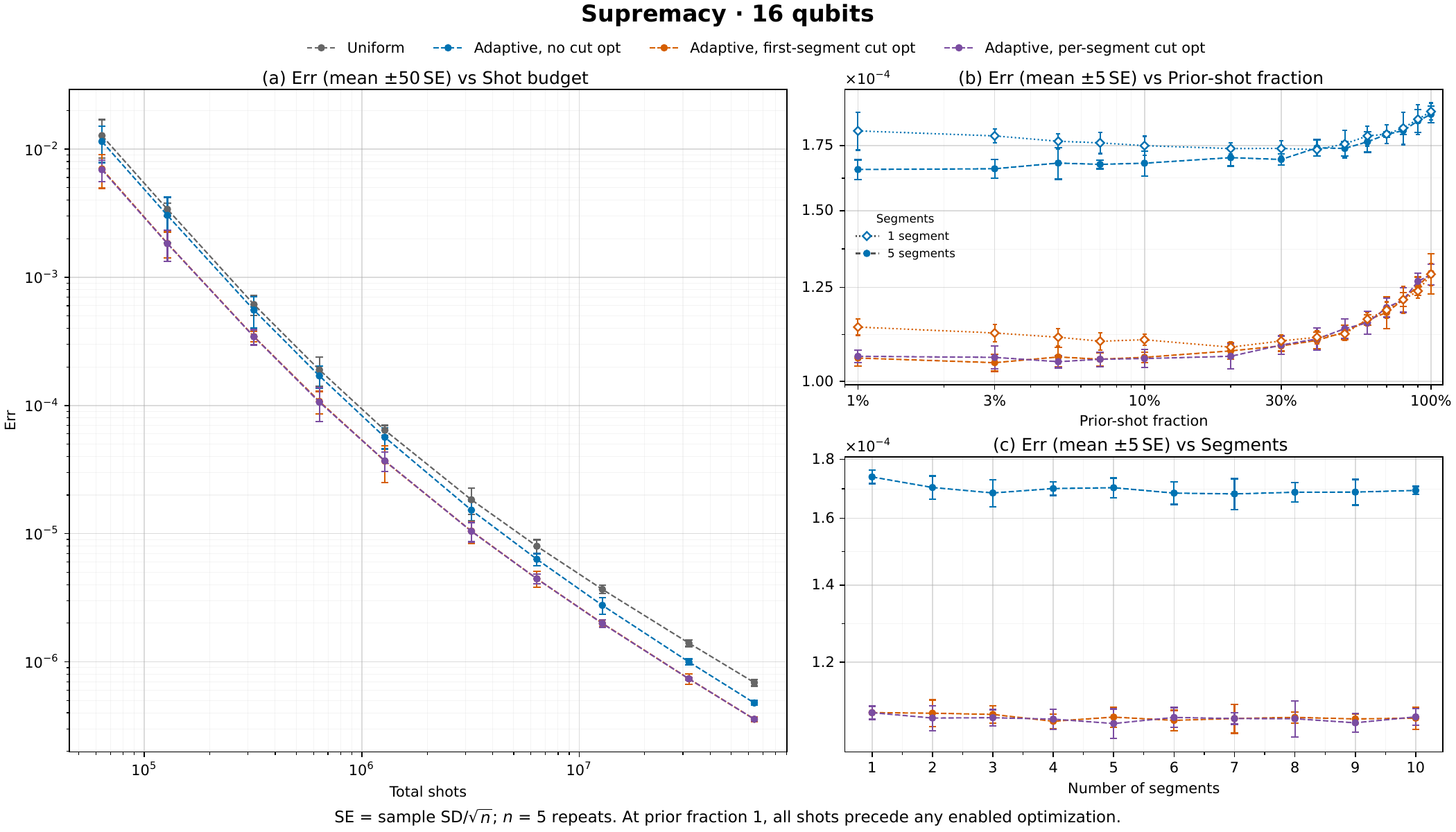}
\caption{Impact of varying ShotQC parameters for the 16-qubit Quantum Supremacy circuit.}
\label{fig:supremacy_16}
\end{figure}

\begin{figure}
\centering
\includegraphics[width=\linewidth]{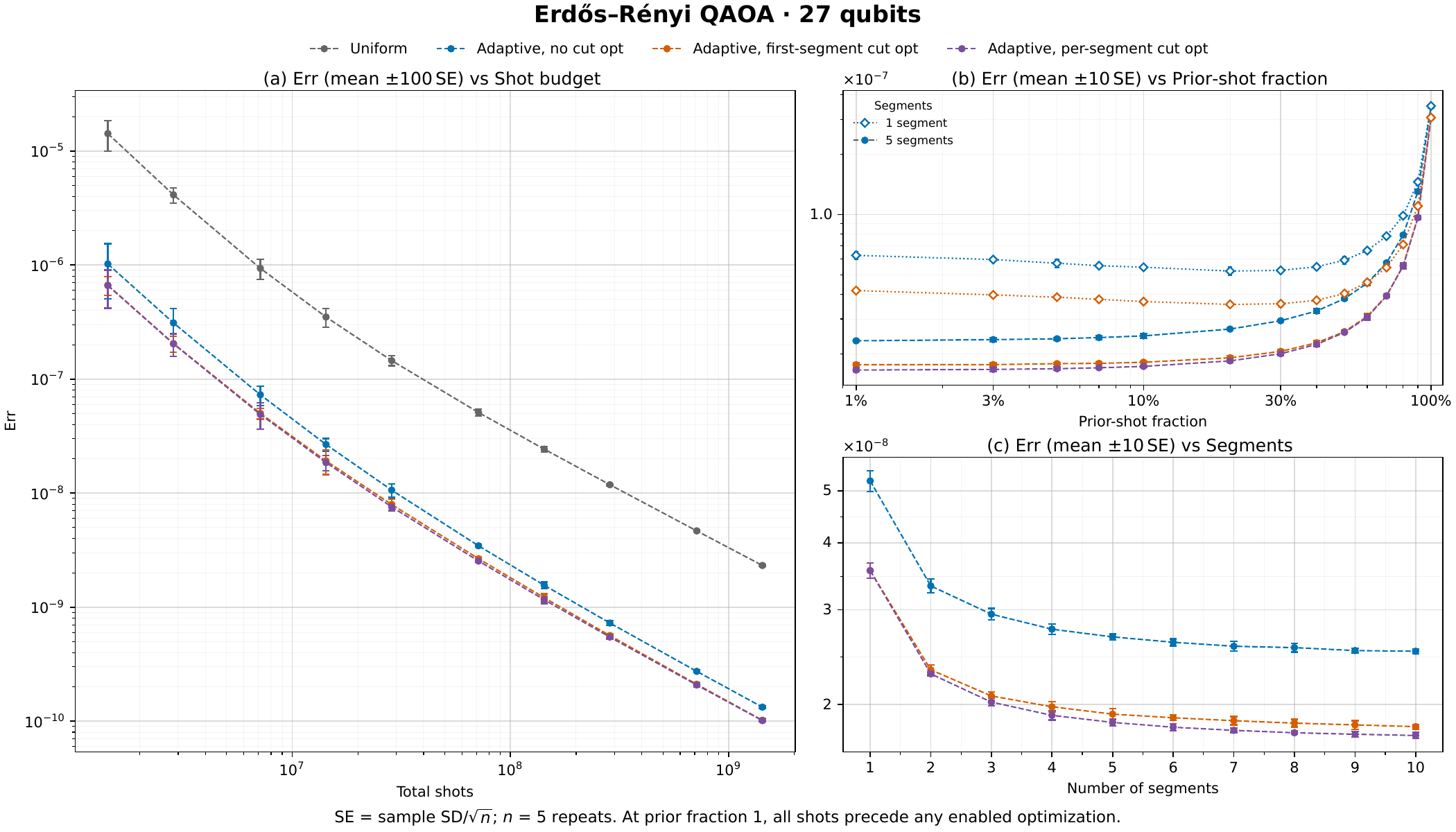}
\caption{Impact of varying ShotQC parameters for the 27-qubit Erdős–Rényi QAOA circuit.}
\label{fig:erdos_27}
\end{figure}

\begin{figure}
\centering
\includegraphics[width=\linewidth]{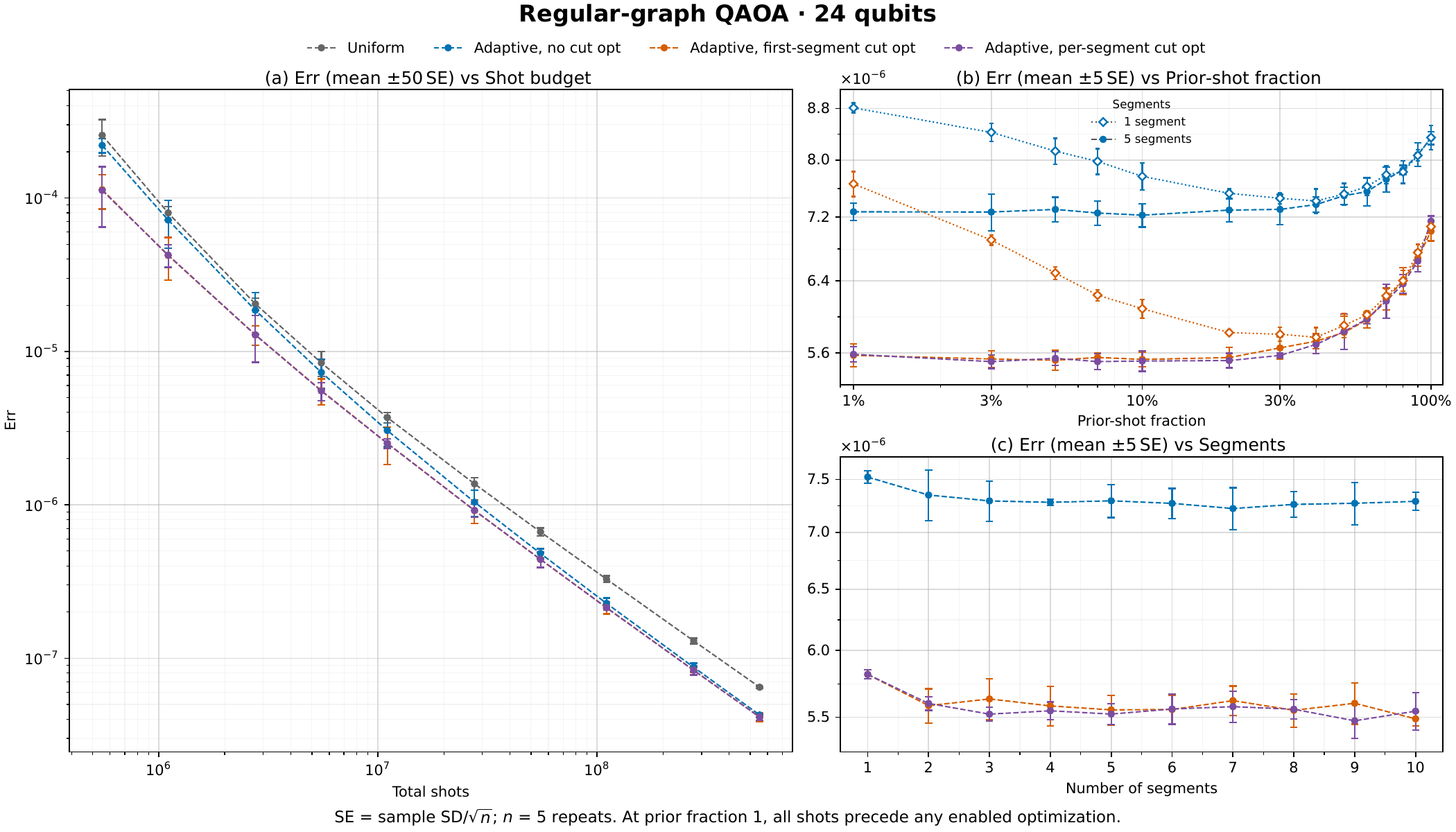}
\caption{Impact of varying ShotQC parameters for the 24-qubit Random Regular QAOA circuit.}
\label{fig:regular_24}
\end{figure}

We first vary the total shot budget by sweeping over $r \in \{0.1, 0.2, 0.5, 1, 2, 5, 10, 20, 50, 100\}$, while fixing the prior-sampling ratio to $0.2$ and the segmentation depth to $5$. As depicted in panel~(a) of each figure, the estimated variance $\mathrm{Err}$ decreases monotonically with increasing shot budget across all benchmark families and optimization settings. Adaptive allocation consistently yields lower sampling variance than uniform allocation for AQFT and Erd{\"o}s--R{\'e}nyi QAOA, whereas the relative improvement is more modest for the Random Regular QAOA and Quantum Supremacy circuits. Cut parameter optimization further amplifies variance reduction for AQFT, whereas its marginal benefit is comparatively limited for the latter two families. The Adder circuit exhibits less regular behavior due to its extremely low baseline variance and higher relative statistical fluctuations in the estimation of $\mathrm{Err}$ (note that the error bars in \figref{fig:adder_32} are not magnified).

We next vary the fraction of the total budget assigned to prior sampling while maintaining a fixed total shot budget of $r = 1$. (At a prior-sampling ratio of $1.0$, the entire budget is consumed prior to any posterior segment without any optimization.) The results in panel~(b) demonstrate that no single optimal prior-sampling ratio exists across all circuit families and optimization settings. Under a 5-segment schedule, relatively small prior fractions ($0.1$--$0.3$) perform exceptionally well for AQFT, Erd{\"o}s--R{\'e}nyi QAOA, Random Regular QAOA, and Quantum Supremacy circuits, indicating that a modest initial sample suffices to generate accurate probability estimates for subsequent adaptive reallocation phases. Allocating an excessively large budget to prior sampling leaves fewer shots available for posterior sampling, thereby dampening the gains of adaptive optimization. Conversely, the Adder circuit in general favors substantially larger prior fractions. The preferred prior-sampling fraction also depends on the number of segments: 1-segment schedules generally require a larger prior sample than 5-segment schedules, as they offer fewer opportunities to refine the allocation using newly gathered measurement data. 

Finally, we sweep the number of segments from $1$ to $10$ while maintaining $r = 1$ and a prior-sampling ratio of $0.2$. Panel~(c) reveals that increasing segmentation depth improves ShotQC performance by enabling repeated reallocation of remaining shots as empirical probability estimates are refined. This effect is particularly pronounced for Erd{\"o}s--R{\'e}nyi QAOA, where $\mathrm{Err}$ decreases substantially with additional allocation updates even when cut parameters remain static. AQFT also benefits from multi-segment execution, especially when paired with per-segment cut parameter optimization. However, these gains display diminishing returns as the segment count grows. For Random Regular QAOA and Quantum Supremacy circuits, the performance gain from extra segments is minor, and per-segment cut parameter optimization does not guarantee monotonic improvements with higher segment counts. The Adder circuit again exhibits larger statistical fluctuations, making optimal segment choices more empirical.

Overall, this parameter sweep demonstrates that ShotQC configuration parameters should be selected with regard to the target circuit family and optimization strategy. While expanding the shot budget consistently reduces sampling variance, the optimal prior-sampling ratio and the ideal number of posterior segments are highly benchmark-dependent. In general, a small prior fraction combined with a moderate number of posterior segments ($3$--$5$) is highly effective for non-adder circuits, with diminishing returns observed beyond a few updates. These empirical insights validate our default configuration---a prior-sampling ratio of $0.2$ and $5$ posterior segments—used in the primary benchmark evaluation, while highlighting avenues for fine-tuning performance across specialized circuit profiles.

\subsection{Runtime and Convergence of Cut Parameter Optimization}

Finally, we characterize the computational cost and convergence behavior of cut parameter optimization in ShotQC through two dedicated experiments. First, we quantify the runtime overhead associated with different cut parameter optimization schedules. Second, we evaluate the convergence profile of the Adam optimizer to verify whether the $100$ update iterations used throughout our study are sufficient.

The dominant computational overhead of ShotQC stems from iterative cut parameter optimization. In contrast, shot distribution optimization and variance estimation rely on closed-form expressions, rendering them computationally negligible by comparison. To quantify optimization overhead, we evaluate adaptive ShotQC across $5$ segments under three distinct optimization schedules: first-segment, per-segment, and post-shot cut parameter (PSCP) optimizations. We focus on the L6 scheme because its cut parameterization encompasses more degrees of freedom than L4, representing the more computationally demanding regime.

\tabref{tab:cut_opt_runtime} summarizes the accumulated cut parameter optimization runtime for the large instance of each benchmark family. Within each benchmark, a single first-segment optimization and a single post-shot optimization require virtually identical runtimes. In contrast, per-segment optimization incurs approximately five times the accumulated runtime, as cut parameters are re-optimized prior to each of the five posterior-shot segments. The absolute optimization cost also varies substantially across benchmark families, reflecting differences in the sampling distributions that enter the evaluation of $f^{(e)}(P,\Theta)$.


\begin{table}[t]
\centering
\caption{Runtime of cut parameter optimization for the large benchmark circuits under the L6 cutting scheme. All adaptive schedules utilize $5$ posterior-shot segments. The per-segment runtime reports the total accumulated time across all five optimization phases.}
\label{tab:cut_opt_runtime}
\begin{tabular}{l c r r rrr}
\toprule
& & & \multicolumn{3}{c}{\textbf{Cut Parameter Optimization Runtime (s)}} \\
\cmidrule(lr){4-6}
\textbf{Benchmark} & \textbf{Cuts} & \textbf{Configs} & \textbf{First-segment} & \textbf{Per-segment} & \textbf{Post-shot (PSCP)} \\
\midrule
AQFT-31 & 5 & 8,019 & 374.7 & 1874.8 & 374.0 \\
Adder-42 & 2 & 36 & 18.6 & 93.9 & 18.6 \\
Erd{\"o}s--R{\'e}nyi-29 & 3 & 243 & 36.7 & 180.2 & 36.8 \\
Regular-30 & 4 & 810 & 9.6 & 45.6 & 9.4 \\
Supremacy-30 & 5 & 4,374 & 364.2 & 1812.7 & 365.4 \\
\bottomrule
\end{tabular}
\end{table}

We next examine the convergence trajectory of cut parameter optimization to determine whether $100$ Adam update iterations provide adequate convergence. For this evaluation, we consider the large instance of each benchmark family under adaptive ShotQC with $5$ segments and active PSCP optimization across both L4 and L6 schemes. The results are shown in \figref{fig:convergence}, where the variance at iteration $i$, denoted $\mathrm{Err}_i$, is normalized by its initial value prior to PSCP optimization, $\mathrm{Err}_0$. Across all benchmark instances in both L4 and L6, the normalized variance drops rapidly during early iterations and plateaus well before the $100$th update. This consistent convergence behavior confirms that $100$ Adam iterations are sufficient for cut parameter optimization, providing a safe margin beyond the point of diminishing returns.

\begin{figure}[t]
\centering
\includegraphics[width=0.9\linewidth]{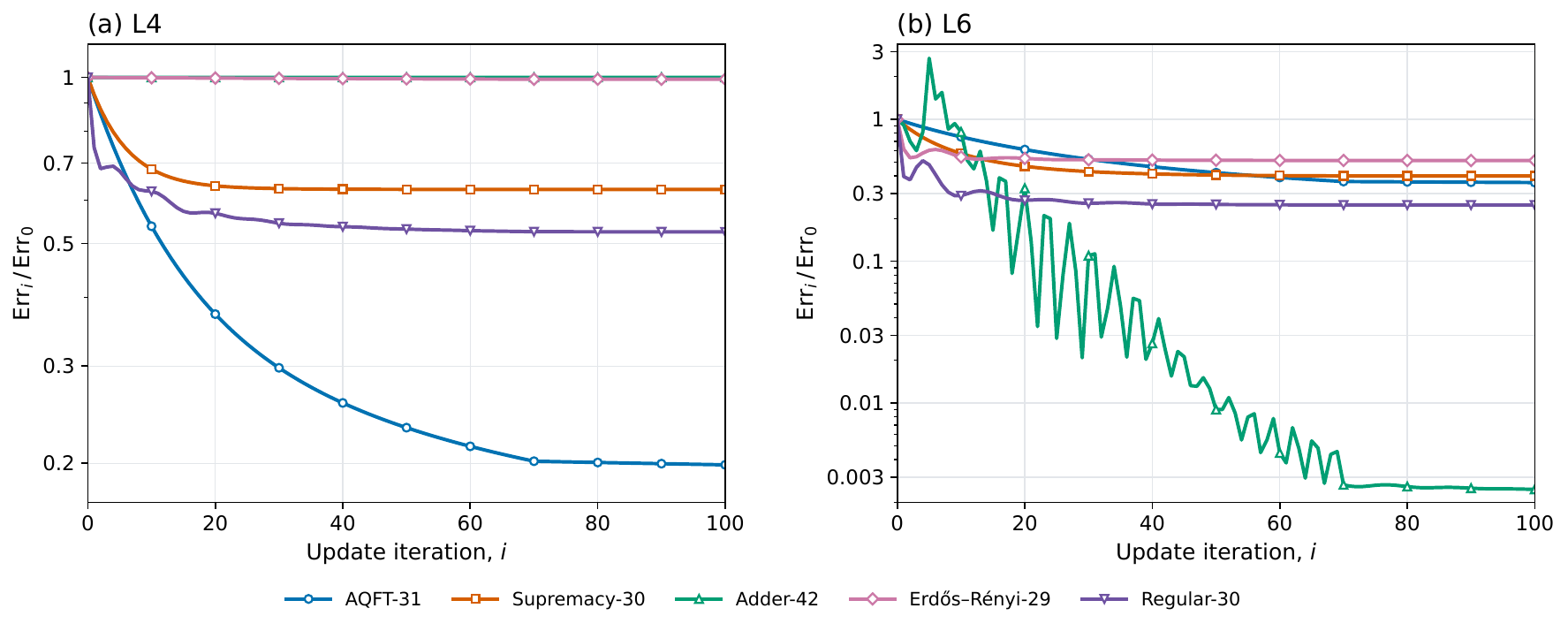}
\caption{Normalized estimated sampling variance $\mathrm{Err}_i / \mathrm{Err}_0$ across Adam update iterations during post-shot cut parameter (PSCP) optimization for the (a) L4 and (b) L6 cutting schemes.}
\label{fig:convergence}
\end{figure}

\section{Comparison with Related Work}\label{sec:related_works}
The overhead associated with quantum circuit cutting can be mitigated through several orthogonal strategies. Existing paradigms primarily focus on optimizing the initial circuit partition, redesigning the underlying measurement protocol, or pruning unnecessary decomposition terms. 

In contrast, ShotQC addresses overhead reduction from a complementary perspective: given an arbitrary circuit partition and a standard measure-and-prepare cutting scheme, it optimizes the allocation of finite sampling resources via adaptive shot distribution and dynamic parameterization refinement. Below, we delineate the distinct merits of ShotQC relative to existing techniques:

\begin{enumerate}
    \item \textbf{Circuit Partitioning and Cut Placement}:
    CutQC~\cite{Tang_2021} employs mixed-integer programming to identify cut locations that minimize cutting overhead subject to subcircuit qubit constraints, whereas subsequent frameworks jointly consider wire and gate cuts~\cite{10374226}. These algorithms optimize \emph{where and how} a circuit is partitioned. ShotQC operates downstream of this step; it takes an established circuit partition as input and optimizes the subsequent execution and reconstruction stages. Consequently, cut-finding algorithms and ShotQC are fully complementary and can be combined directly in a pipeline.

    \item \textbf{Alternative Cutting and Sampling Protocols}:
    Several approaches alter the physical execution protocol to reduce theoretical sampling complexity, including communication-assisted cutting~\cite{brenner2023optimalwirecuttingclassical} and randomized measurement techniques~\cite{Lowe_2023}. These strategies modify the asymptotic scaling of the cutting protocol itself. In contrast, ShotQC operates entirely within the standard non-communicating setting with independently executable subcircuits. Rather than altering asymptotic protocol bounds, ShotQC reduces practical, instance-dependent sampling overhead by exploiting probability mass concentration across subcircuit configurations and adaptively reallocating a fixed shot budget.

    \item \textbf{Pruning Reconstruction Terms}:
    Frameworks such as RICCO~\cite{10025537} introduce local rotations at cut locations to suppress specific observable terms, while other approaches leverage structural symmetries to discard unnecessary basis elements prior to execution~\cite{10196555,10313822}. ShotQC does not explicitly prune subcircuit configurations. Instead, it empirically assesses the statistical weight of each retained configuration from prior-shot data and dynamically allocates shots accordingly. Furthermore, its cut parameter optimization directly exploits redundancies in the reconstruction space to minimize variance propagation. As such, structural pruning techniques and ShotQC address distinct bottlenecks and can be integrated synergistically.
\end{enumerate}

\section{Conclusion}
In this work, we presented ShotQC, an enhanced quantum circuit cutting framework designed to dramatically reduce sampling overhead in practical execution settings. By combining adaptive, multi-segment shot distribution optimization, dynamic cut parameterization, and post-shot cut parameter (PSCP) updates, ShotQC systematically exploits the structural characteristics of subcircuit probability distributions. 

Through extensive benchmarking across five distinct circuit families---representing a broad spectrum of output probability profiles from uniform (AQFT) to diffuse non-uniform (QAOA, Quantum Supremacy) and extremely concentrated (Adder)---we demonstrated that ShotQC delivers substantial sampling overhead reductions across all tested regimes. Crucially, the performance gains scale with the degree of probability mass concentration: while uniform distributions yield moderate improvements, highly concentrated output profiles achieve extraordinary variance reductions, with improvement factors $R$ exceeding $80{,}000\times$. 

Furthermore, our results reveal that the L6 wire-cutting scheme routinely outperforms L4 in terms of absolute reconstruction variance---countering conventional wisdom that fewer state preparations necessarily yield lower overhead—because L6's richer preparation basis better exposes and exploits probability mass concentration. Overall, ShotQC provides flexible operational schedules balancing optimization runtime and sampling efficiency, establishing adaptive shot allocation and parameter optimization as critical components for scalable quantum circuit cutting.

It should be noted that the benchmark evaluations in this work were conducted using noiseless classical simulations on small-to-medium-scale instances, serving as a necessary baseline to isolate the theoretical variance-reduction mechanics of ShotQC. In real noisy hardware environments, sampling overhead is intimately tied to gate errors, decoherence, and readout mitigation costs. Because ShotQC reduces total sampling variance substantially faster than conventional uniform allocation, it is reasonable to conjecture that this accelerated statistical convergence could indirectly improve overall fidelity on NISQ devices by requiring fewer physical executions before noise dominates. However, the exact interplay between adaptive shot optimization and hardware noise remains an open question, and further research---including noisy backend simulations and real-hardware deployments---is required to fully characterize the practical fidelity gains in realistic execution settings.


\begin{acknowledgments}
The authors would like to thank Tain-Fu Chen and Yu-Hung Pan for fruitful discussions, as well as the anonymous reviewers for their valuable comments and suggestions, which significantly improved the quality of this work.
This work was supported in part by the National Science and Technology Council, Taiwan, under grants 14-2119-M-002-02, 115-2119-M-002-006, and 115-2811-M-002-039; by the NTU Center of Data Intelligence: Technologies, Applications, and Systems under grants NTU-113L900903 and NTU-114L900903; and by the Intelligent Computing for Sustainable Development Research Center at National Taiwan Normal University under the Higher Education Sprout Project of the Ministry of Education, Taiwan.
\end{acknowledgments}

\section*{AI tool usage statement}
The authors used ChatGPT and Gemini to assist in refining the manuscript text and to support the implementation and debugging of Python code. All content, code, and results were carefully reviewed and validated by the authors, who assume full responsibility for the integrity and accuracy of the final work.

\appendix

\section{Quantum Circuits as Tensor Networks}\label{sec: back}
This section presents a brief introduction to the tensor network representation of quantum circuits and the proof of Lemma~\ref{lemma:cut}, based on the formulation in \cite{Peng_2020}.


Any quantum circuit can be represented as a directed graph $G=(V,E)$, where the directed edges indicate the flow of qubits, and the vertices are categorized into three types: input states (denoted by $\lhd$), quantum gates (denoted by $\Box$), and observables (denoted by $\rhd$). This graph can then be converted into a tensor network $\mathcal{A} = \{A(v): v \in V\}$, where each vertex $v$ is associated with a tensor $A(v)$. 

The set of $k$ incoming or outgoing edges attached to a tensor is indexed by $\alpha = (\alpha_1, \alpha_2, \dots, \alpha_k)$, where each $\alpha_j \equiv(\alpha_j^1,\alpha_j^2)\in \{0,1\}^2$ is a pair of binary indices. Let $M(\alpha) = |\alpha_1^1, \dots, \alpha_k^1\rangle \langle \alpha_1^2, \dots, \alpha_k^2|$ be the elementary matrix corresponding to $\alpha$. The tensor entries of the input state $\rho$, the quantum gate $U$, and the observable $O$ are given by
\setlength{\abovedisplayskip}{5pt}  
\setlength{\belowdisplayskip}{5pt}  
\begin{subequations}
\begin{align}
     A(\rho)_\alpha &= \Tr[\rho  M(\alpha)^\dagger], \label{eq:rho}\\
     _\alpha A(U)_\beta &= \Tr[U M(\alpha) U^\dagger  M(\beta)^\dagger], \label{eq:U}\\
     _\beta A(O) &= \Tr[M(\beta)  O], \label{eq:O}
\end{align}
\end{subequations}
where the left index of a tensor represents the incoming edges and the right index represents the outgoing edges. Using these formulae, we can evaluate the quantum circuit through its equivalent tensor network.

This tensor network representation allows us to formulate Lemma~\ref{lemma:cut}, the foundational result for wire cutting. The proof proceeds as follows.

\begin{proof}[Proof of Lemma~\ref{lemma:cut}]
Since every edge of the tensor network is equivalent to an identity ($I$) channel,
we have to verify that
\begin{equation}
    _\alpha A(I)_\beta = \sum^\ell_{i=1}r_i\, _\alpha A(O_i)\, A(\rho_i)_\beta.
    \label{eq:verify}
\end{equation}
Obviously, the l.h.s.\ of \eqref{eq:verify} is given by
\begin{equation}\label{eq:lhs}
    _\alpha A(I)_\beta = \Tr[M(\alpha)\, M(\beta)^\dagger] = \delta_{\alpha, \beta}.
\end{equation}
Meanwhile, the r.h.s.\ is given by
\begin{align}\label{eq:rhs}
    &\sum^\ell_{i=1}r_i\, _\alpha A(O_i)\, A(\rho_i)_\beta 
    \equiv \sum^\ell_{i=1}r_i\, \Tr[M(\alpha)O_i]\Tr[\rho_iM(\beta)^\dagger] \nonumber \\
    &\quad= \Tr\left[\sum^\ell_{i=1} r_i\Tr[M(\alpha)O_i]\rho_i M(\beta)^\dagger\right] \nonumber\\
    &\quad= \Tr[M(\alpha)M(\beta)^\dagger] = \delta_{\alpha,\beta},
\end{align}
where we have used the identity
\begin{equation}
    M(\alpha) = \sum_{i=1}^\ell r_i \Tr[M(\alpha) O_i]\rho_i,
\end{equation}
as a consequence of \eqref{eq:decomp} with $M(\alpha)$ substituted for $A$.

The equality of \eqref{eq:lhs} and \eqref{eq:rhs} proves \eqref{eq:verify}.
\end{proof}

\section{Detailed Benchmark Results}\label{sec:sup_data}

For completeness, the detailed plots---including statistical error bars---corresponding to the benchmark results summarized in \tabref{tab:shotqc_small_results}, \tabref{tab:shotqc_medium_results}, and \tabref{tab:shotqc_large_results} are provided in this appendix: \figref{fig:adder_small}--\figref{fig:adder_large} for Adder, \figref{fig:aqft_small}--\figref{fig:aqft_large} for AQFT, \figref{fig:erdos_small}--\figref{fig:erdos_large} for Erd{\"o}s--R{\'e}nyi QAOA, \figref{fig:regular_small}--\figref{fig:regular_large} for Random Regular QAOA, and \figref{fig:supremacy_small}--\figref{fig:supremacy_large} for Quantum Supremacy circuits. These figures illustrate performance trends across all 12 optimization configurations for both the L4 and L6 wire-cutting schemes.

In each plot, circular and square markers denote the sample means of $\mathrm{Err}$ for the L4 and L6 schemes, respectively, while horizental error bars indicate the standard error of the mean (SE) calculated across $5$ independent sampling runs. Owing to the large total shot budget, the standard errors are generally too small to be visually discernible, except for various data points in the Adder instances (\figref{fig:adder_small}, \figref{fig:adder_medium}, and \figref{fig:adder_large}).

Together, these plots facilitate fine-grained visual comparisons of the performance gains achieved by ShotQC across varying circuit scales, benchmark families, and optimization strategies---with logarithmic scales enabling intuitive assessment of variance reductions spanning multiple orders of magnitude.

\begin{figure}
\centering
\includegraphics[width=0.85\linewidth]{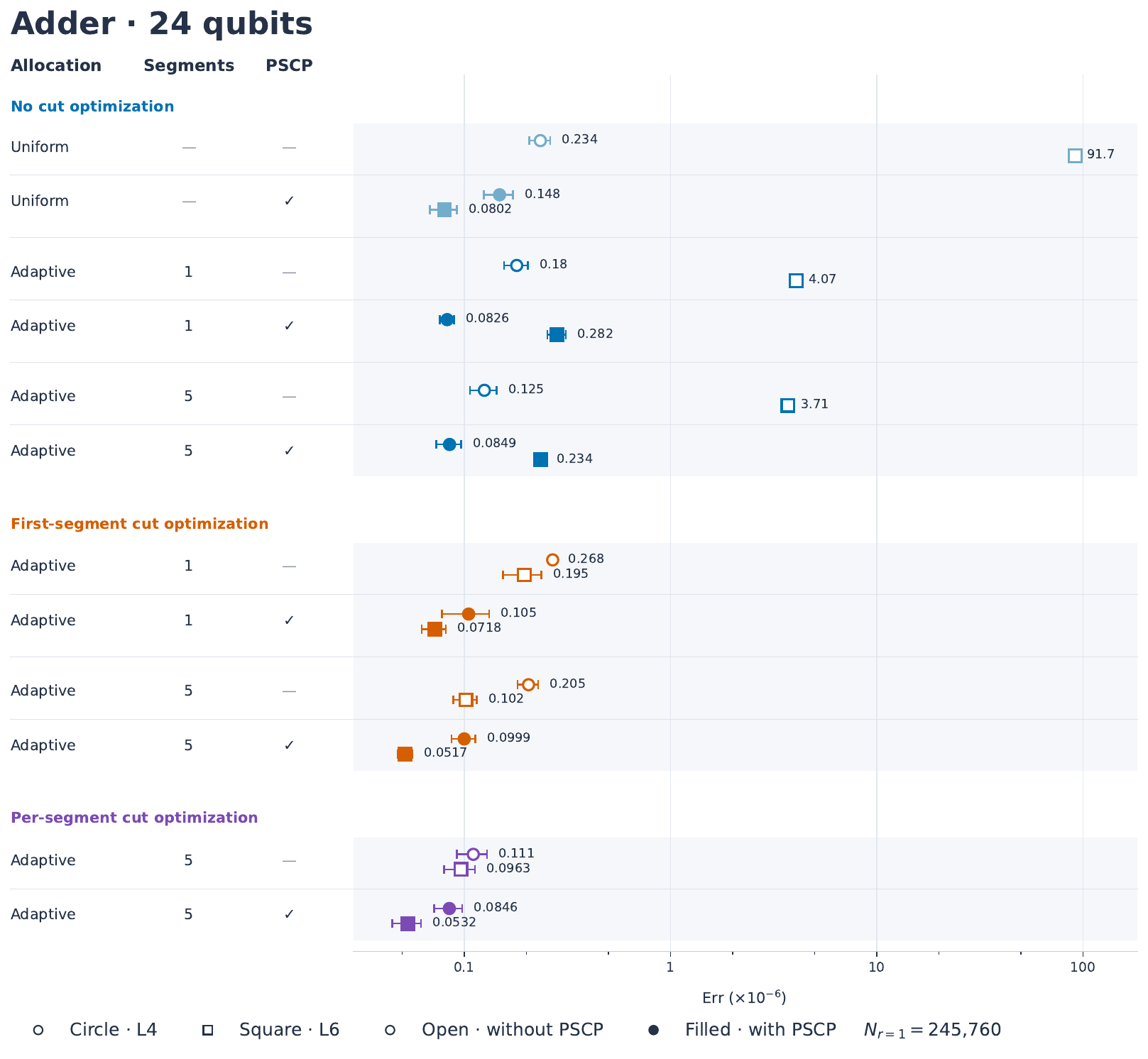}
\caption{Estimated sampling variance $\mathrm{Err}$ across different optimization settings for the small Adder circuit.}
\label{fig:adder_small}
\end{figure}

\begin{figure}
\centering
\includegraphics[width=0.85\linewidth]{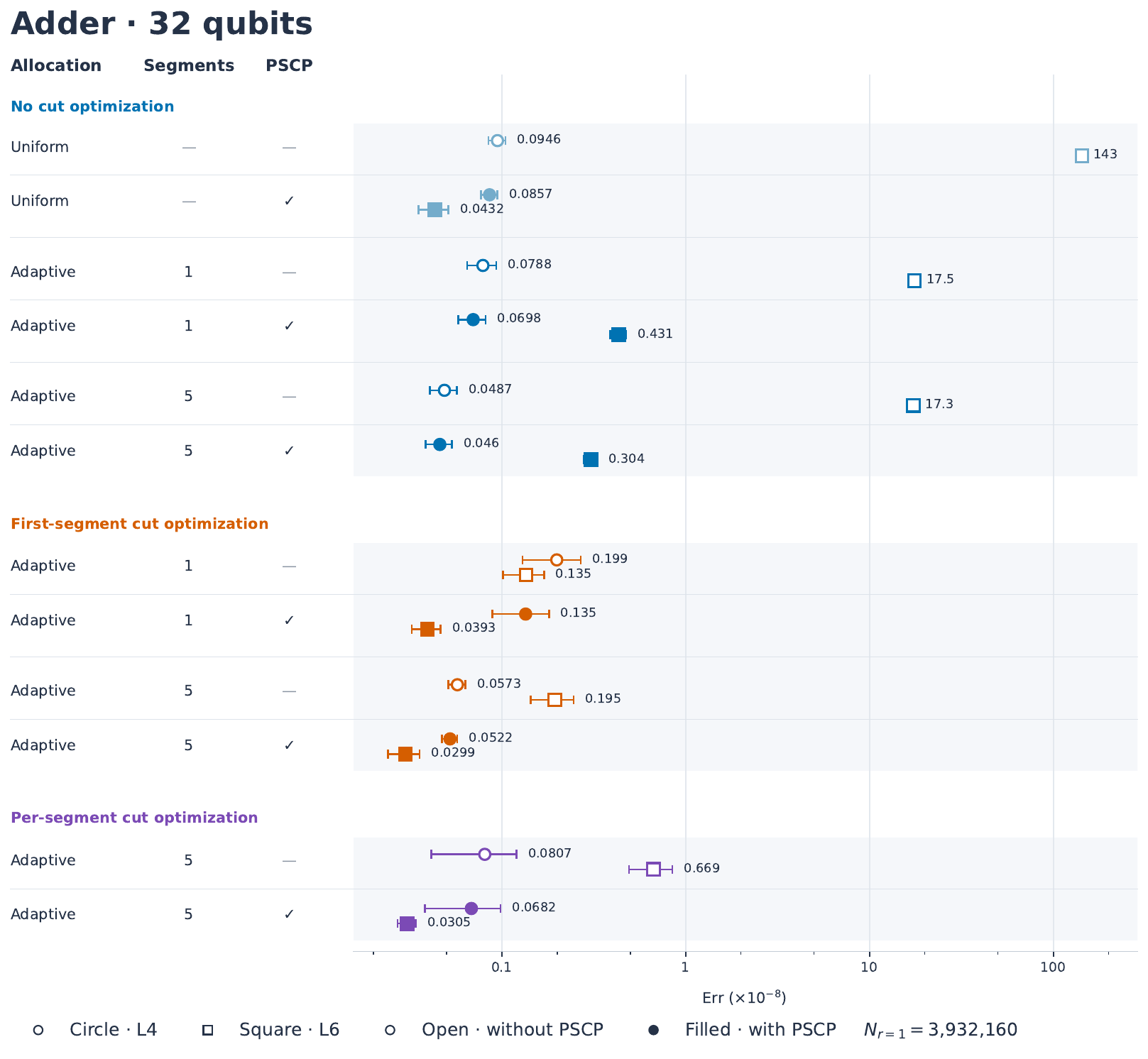}
\caption{Estimated sampling variance $\mathrm{Err}$ across different optimization settings for the medium Adder circuit.}
\label{fig:adder_medium}
\end{figure}

\begin{figure}
\centering
\includegraphics[width=0.85\linewidth]{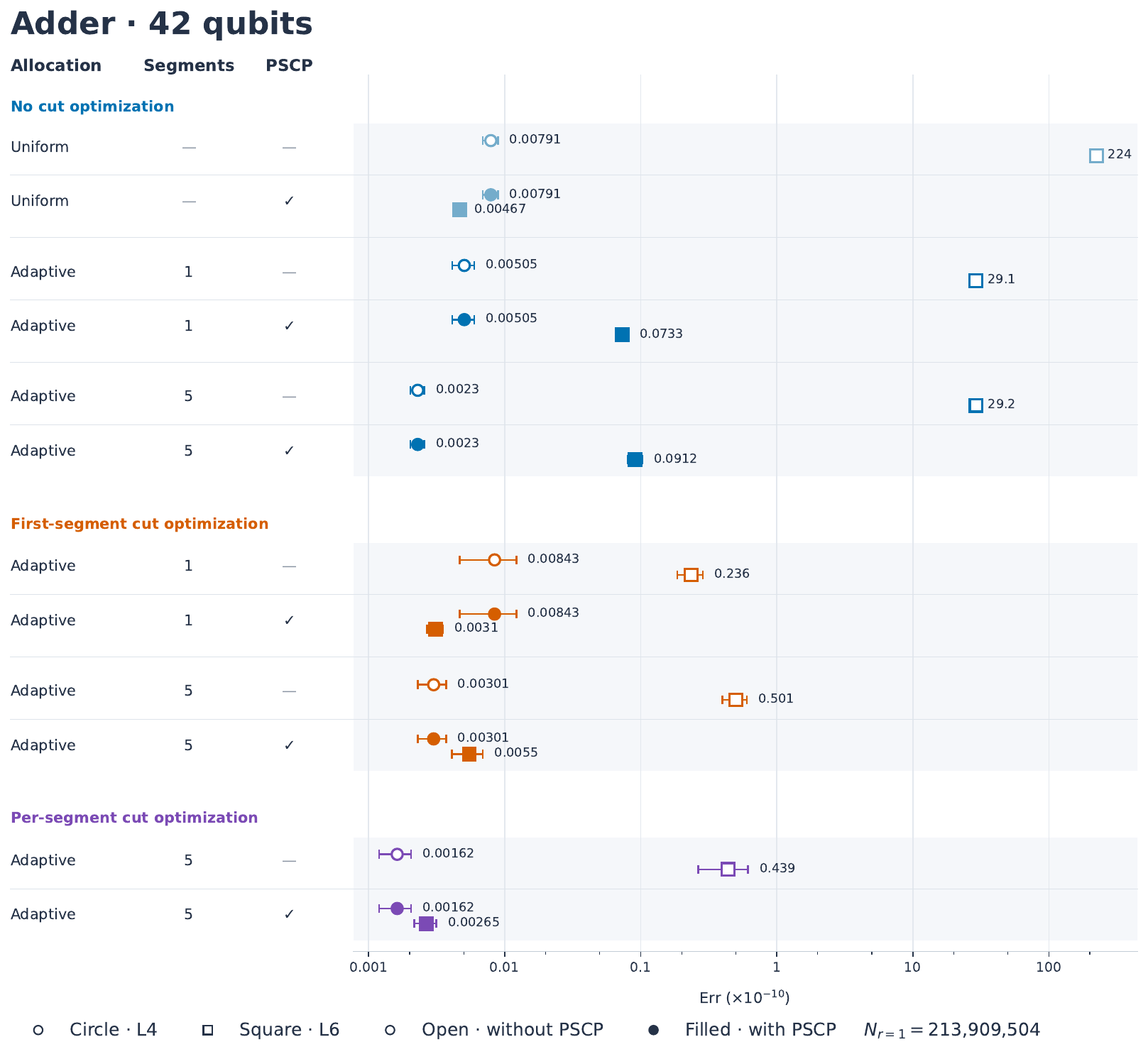}
\caption{Estimated sampling variance $\mathrm{Err}$ across different optimization settings for the large Adder circuit.}
\label{fig:adder_large}
\end{figure}

\begin{figure}
\centering
\includegraphics[width=0.85\linewidth]{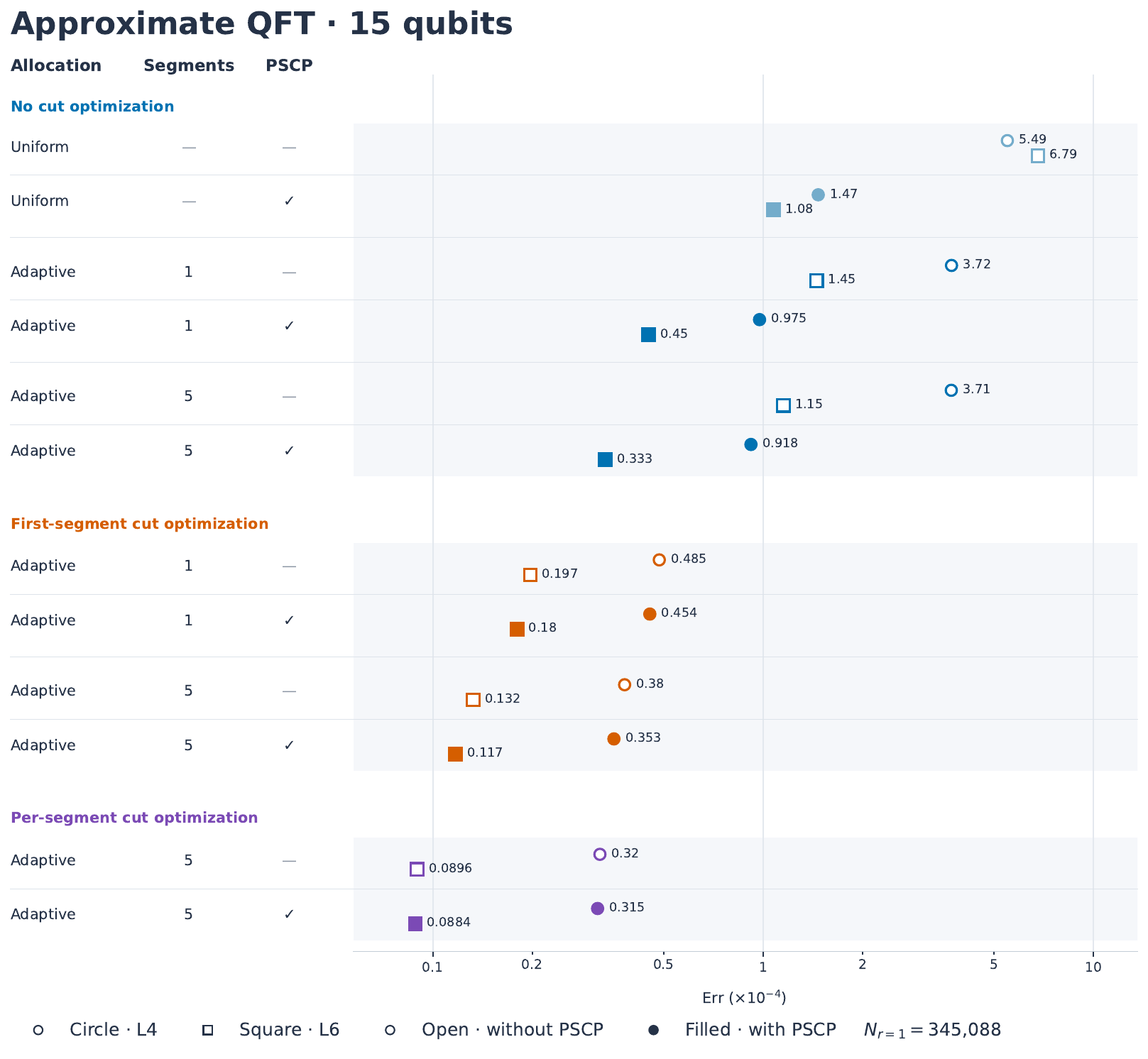}
\caption{Estimated sampling variance $\mathrm{Err}$ across different optimization settings for the small AQFT circuit.}
\label{fig:aqft_small}
\end{figure}

\begin{figure}
\centering
\includegraphics[width=0.85\linewidth]{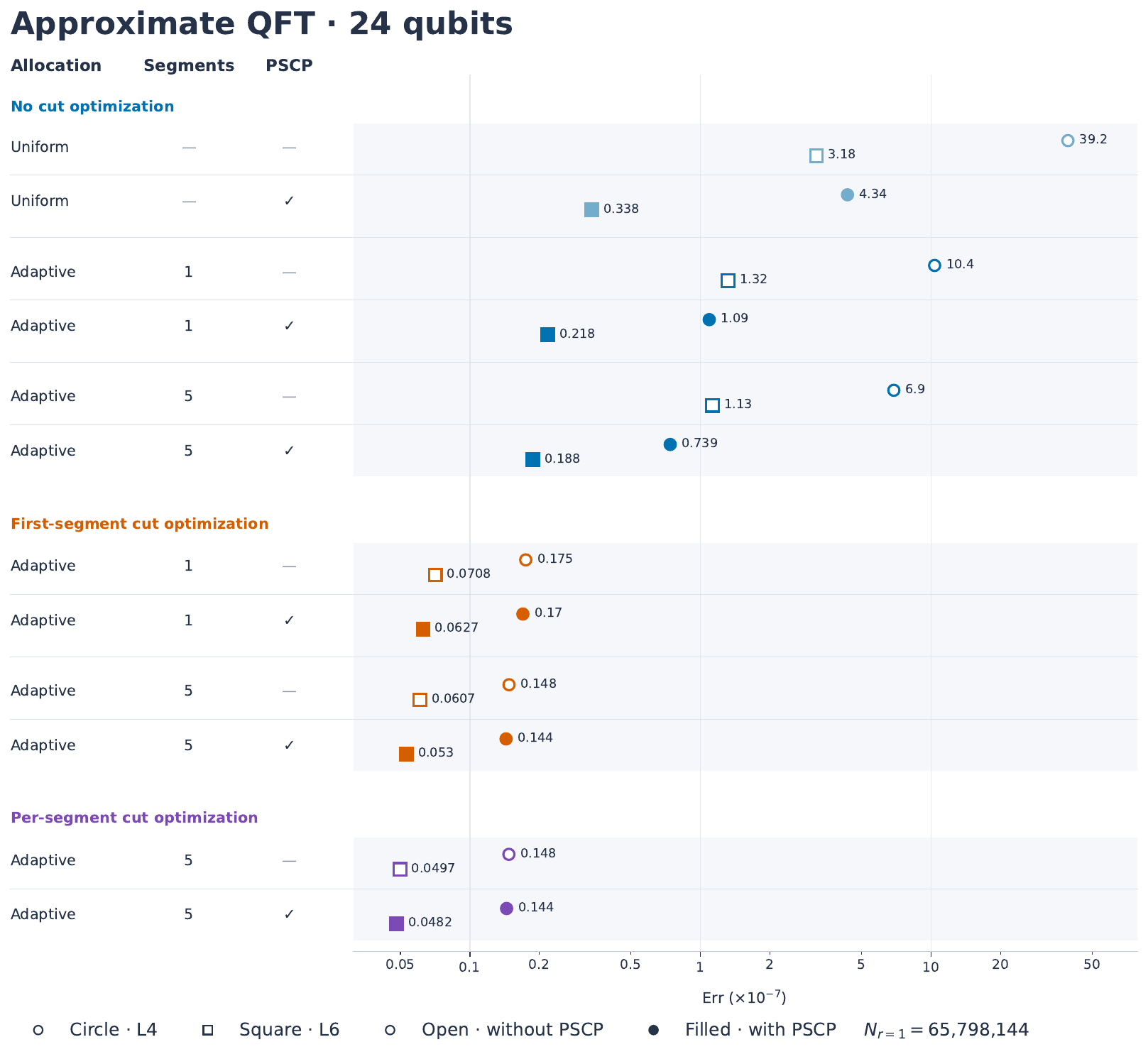}
\caption{Estimated sampling variance $\mathrm{Err}$ across different optimization settings for the medium AQFT circuit.}
\label{fig:aqft_medium}
\end{figure}

\begin{figure}
\centering
\includegraphics[width=0.85\linewidth]{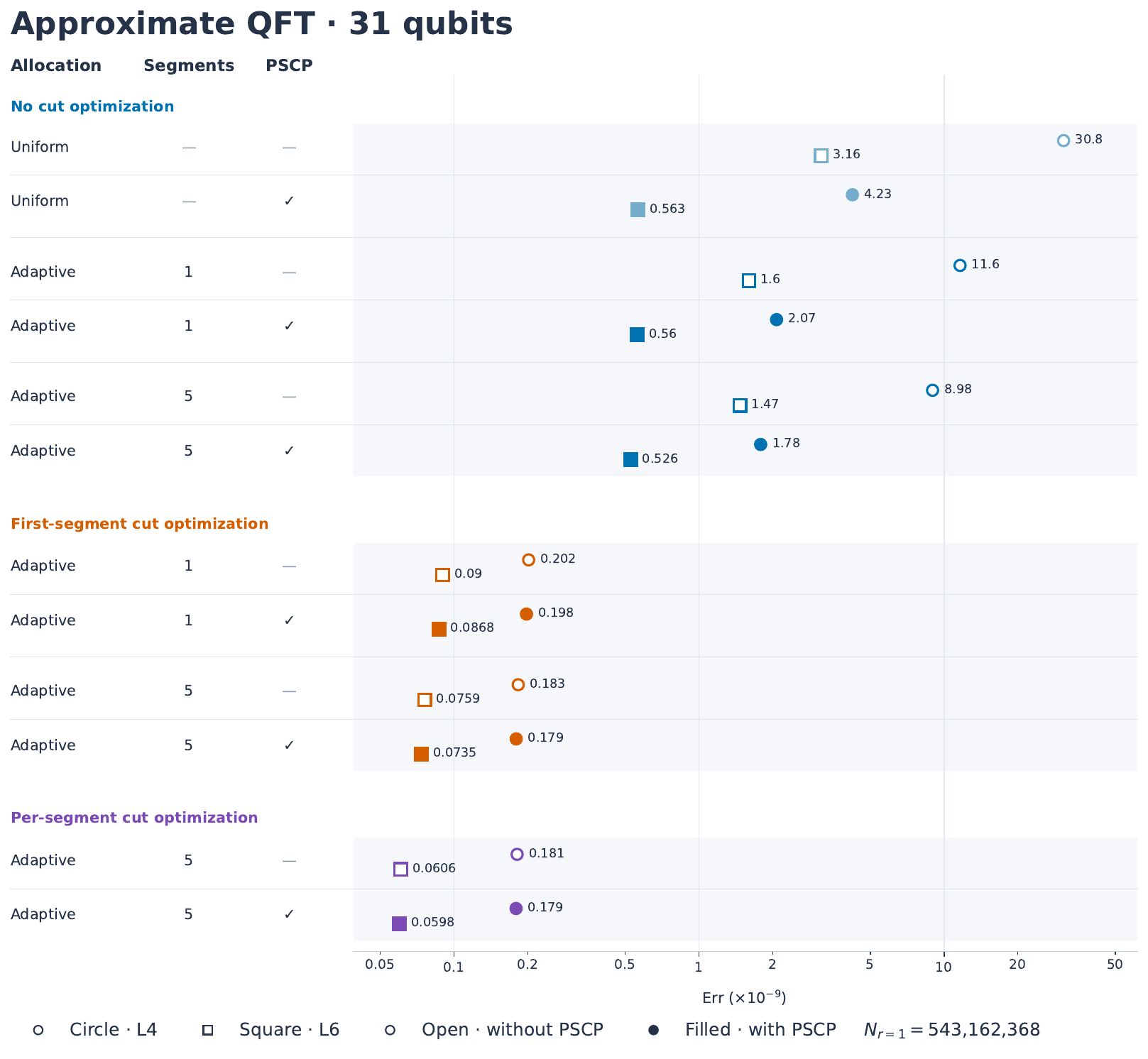}
\caption{Estimated sampling variance $\mathrm{Err}$ across different optimization settings for the large AQFT circuit.}
\label{fig:aqft_large}
\end{figure}

\begin{figure}
\centering
\includegraphics[width=0.85\linewidth]{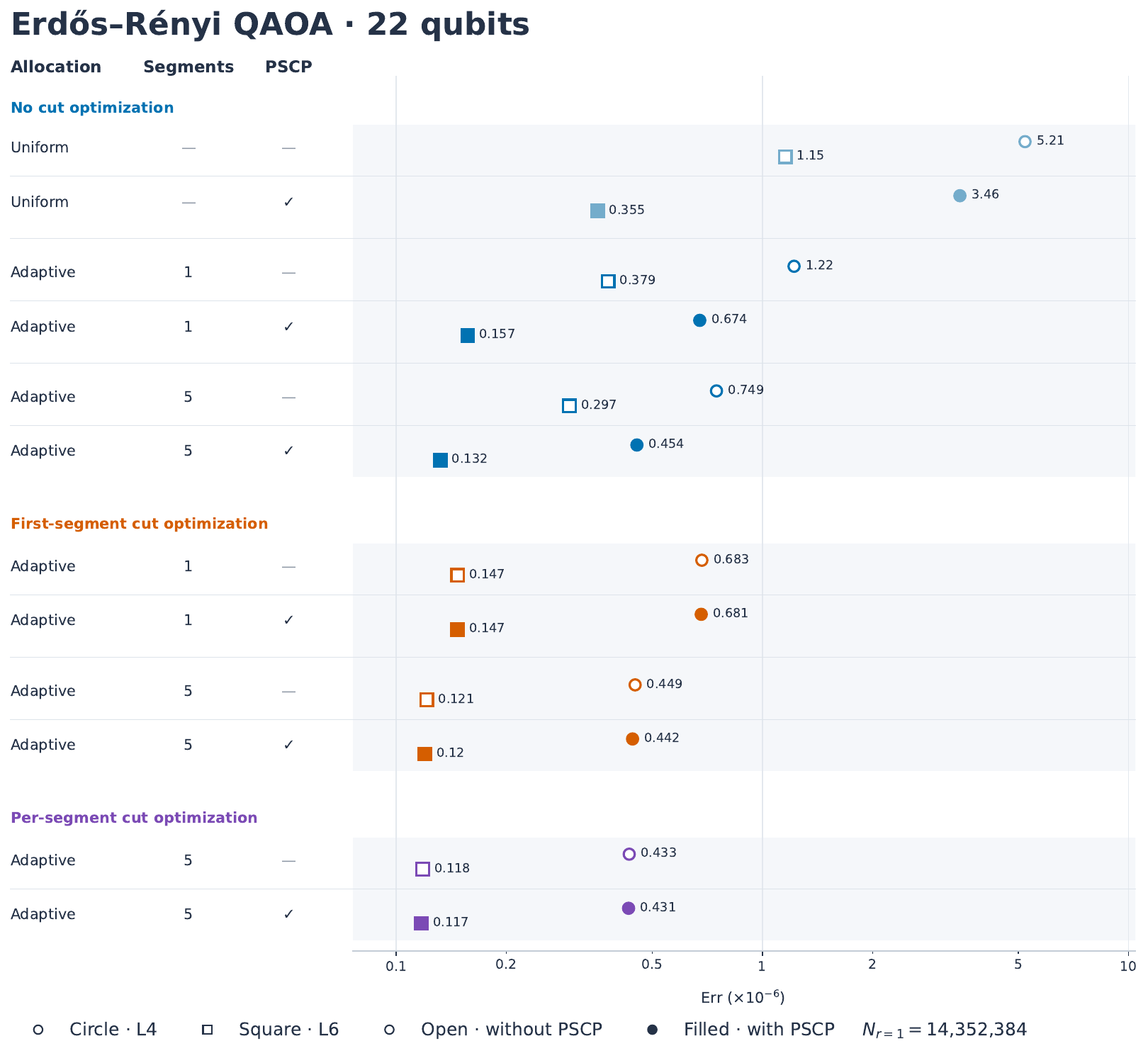}
\caption{Estimated sampling variance $\mathrm{Err}$ across different optimization settings for the small Erd\"{o}s--R\'{e}nyi circuit.}
\label{fig:erdos_small}
\end{figure}

\begin{figure}
\centering
\includegraphics[width=0.85\linewidth]{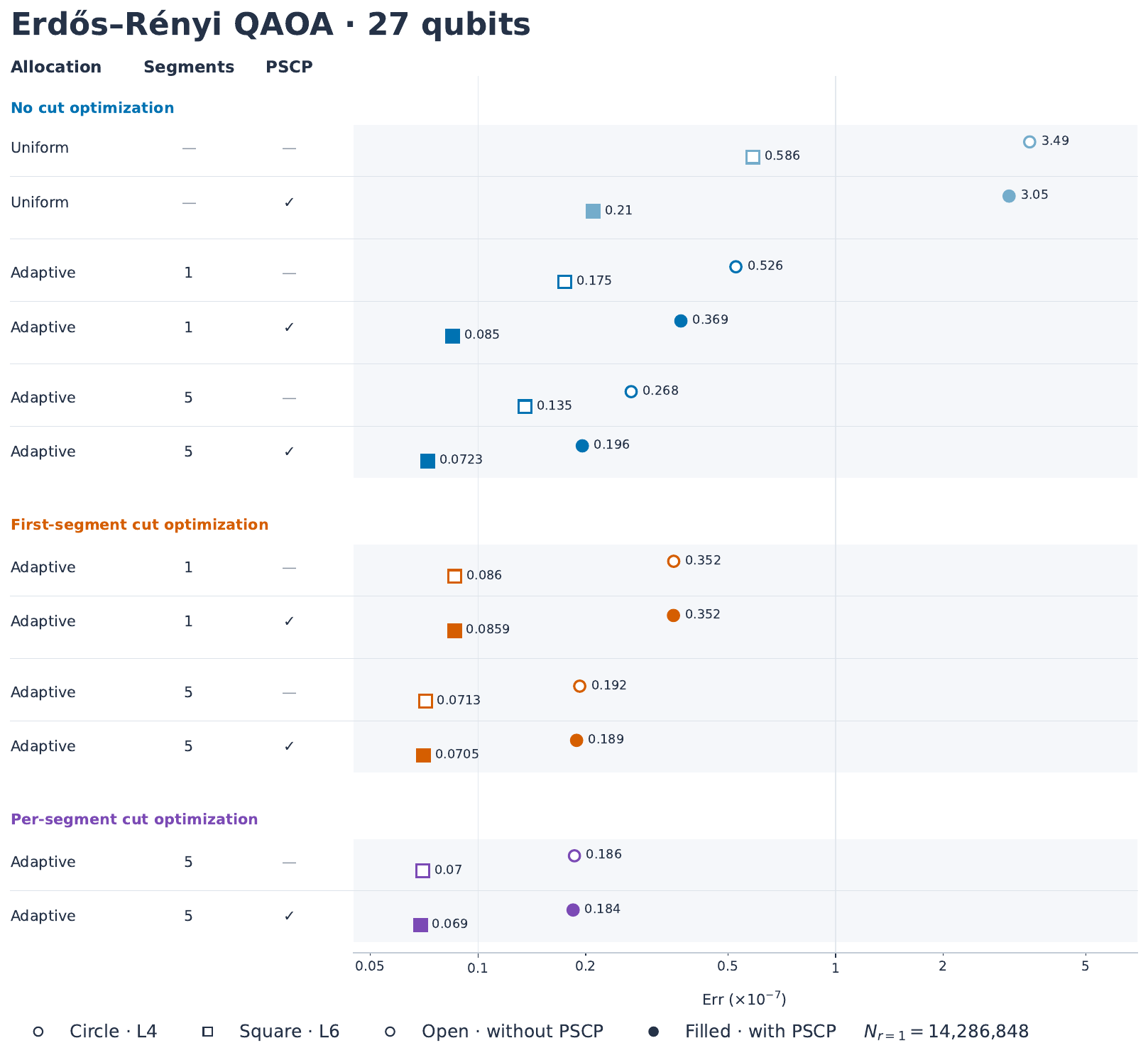}
\caption{Estimated sampling variance $\mathrm{Err}$ across different optimization settings for the medium Erd\"{o}s--R\'{e}nyi circuit.}
\label{fig:erdos_medium}
\end{figure}

\begin{figure}
\centering
\includegraphics[width=0.85\linewidth]{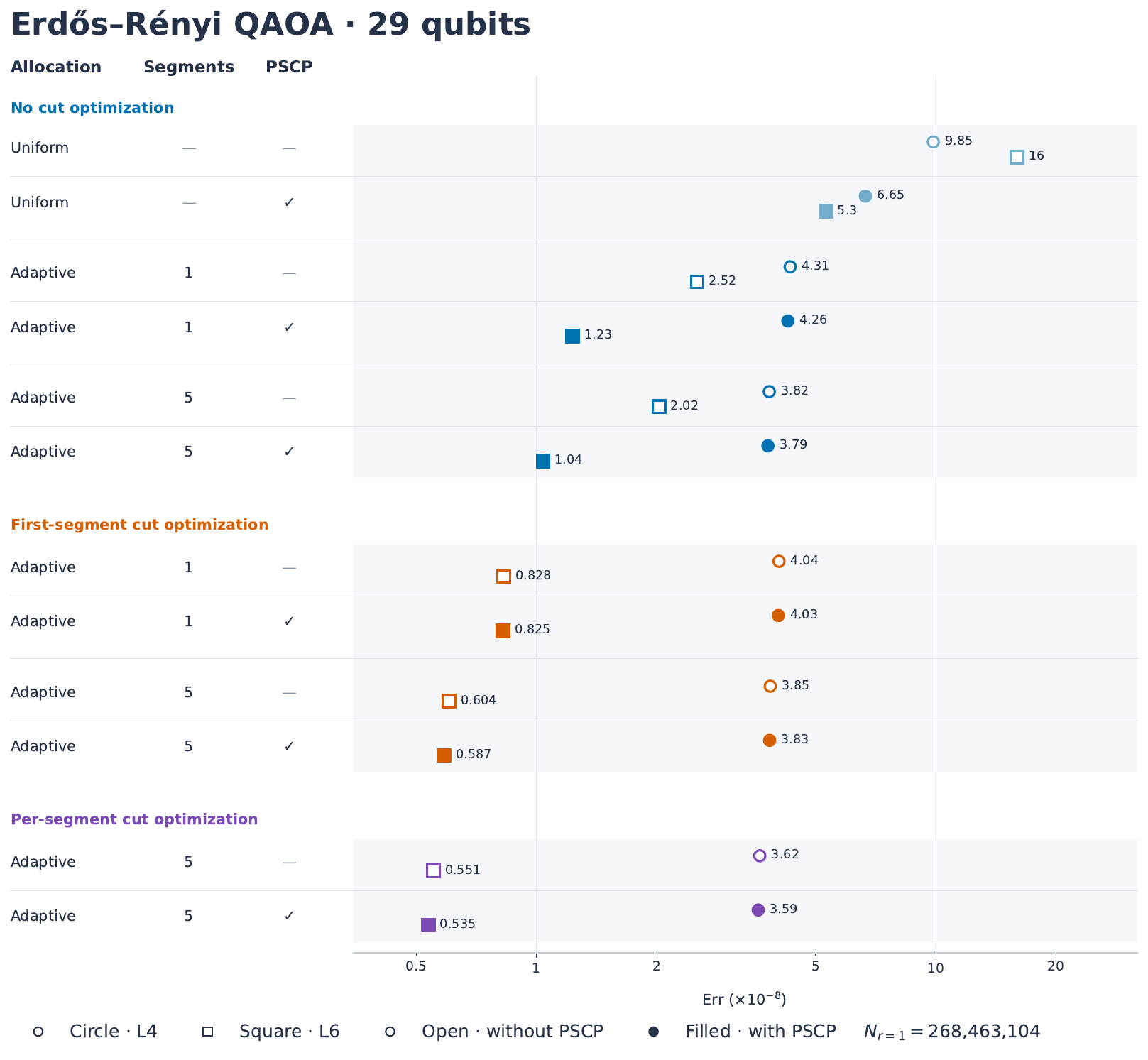}
\caption{Estimated sampling variance $\mathrm{Err}$ across different optimization settings for the large Erd{\"o}s--R{\'e}nyi QAOA circuit.}
\label{fig:erdos_large}
\end{figure}

\begin{figure}
\centering
\includegraphics[width=0.85\linewidth]{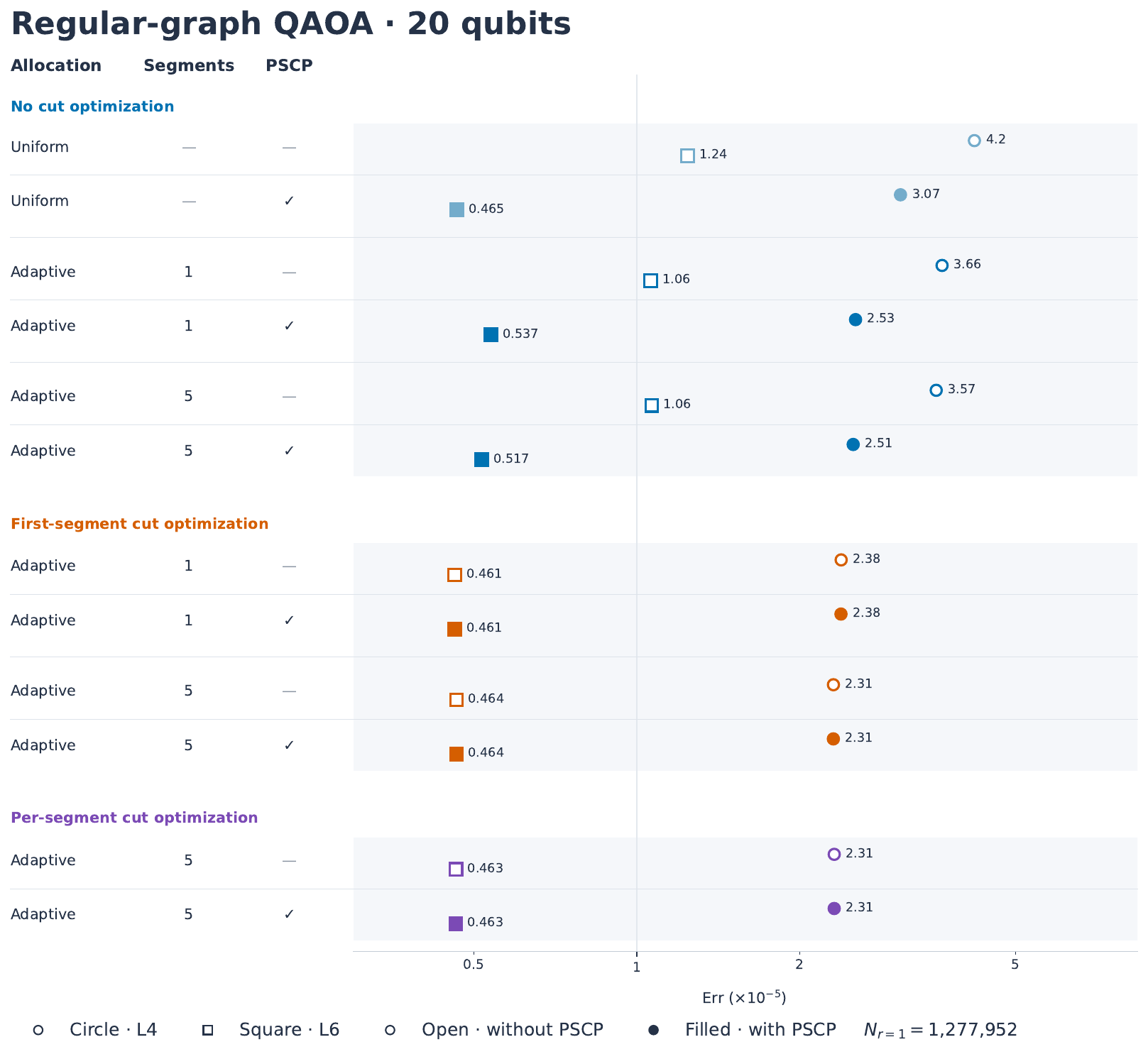}
\caption{Estimated sampling variance $\mathrm{Err}$ across different optimization settings for the small Random Regular circuit.}
\label{fig:regular_small}
\end{figure}

\begin{figure}
\centering
\includegraphics[width=0.85\linewidth]{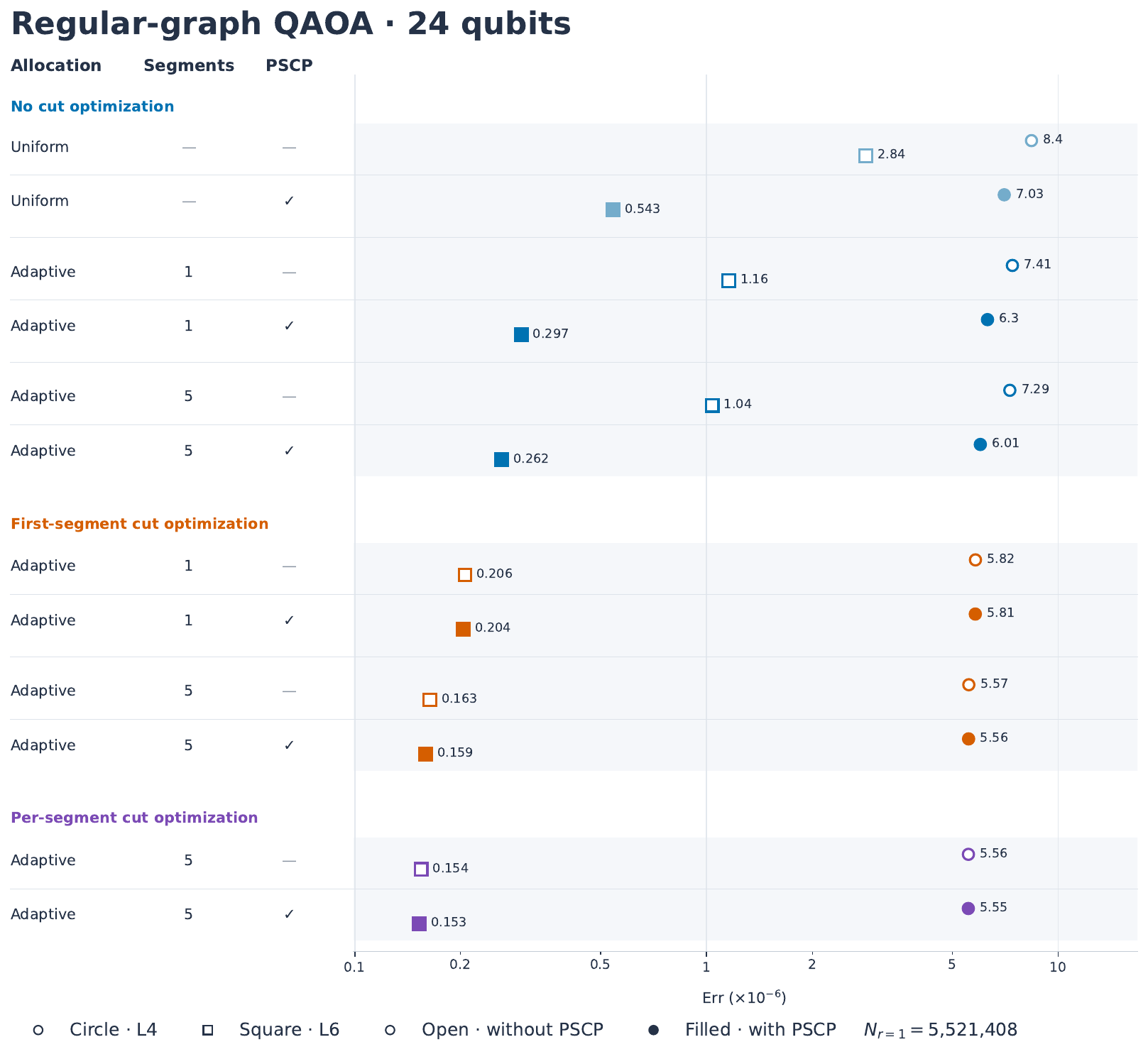}
\caption{Estimated sampling variance $\mathrm{Err}$ across different optimization settings for the medium Random Regular circuit.}
\label{fig:regular_medium}
\end{figure}

\begin{figure}[t]
\centering
\includegraphics[width=0.85\linewidth]{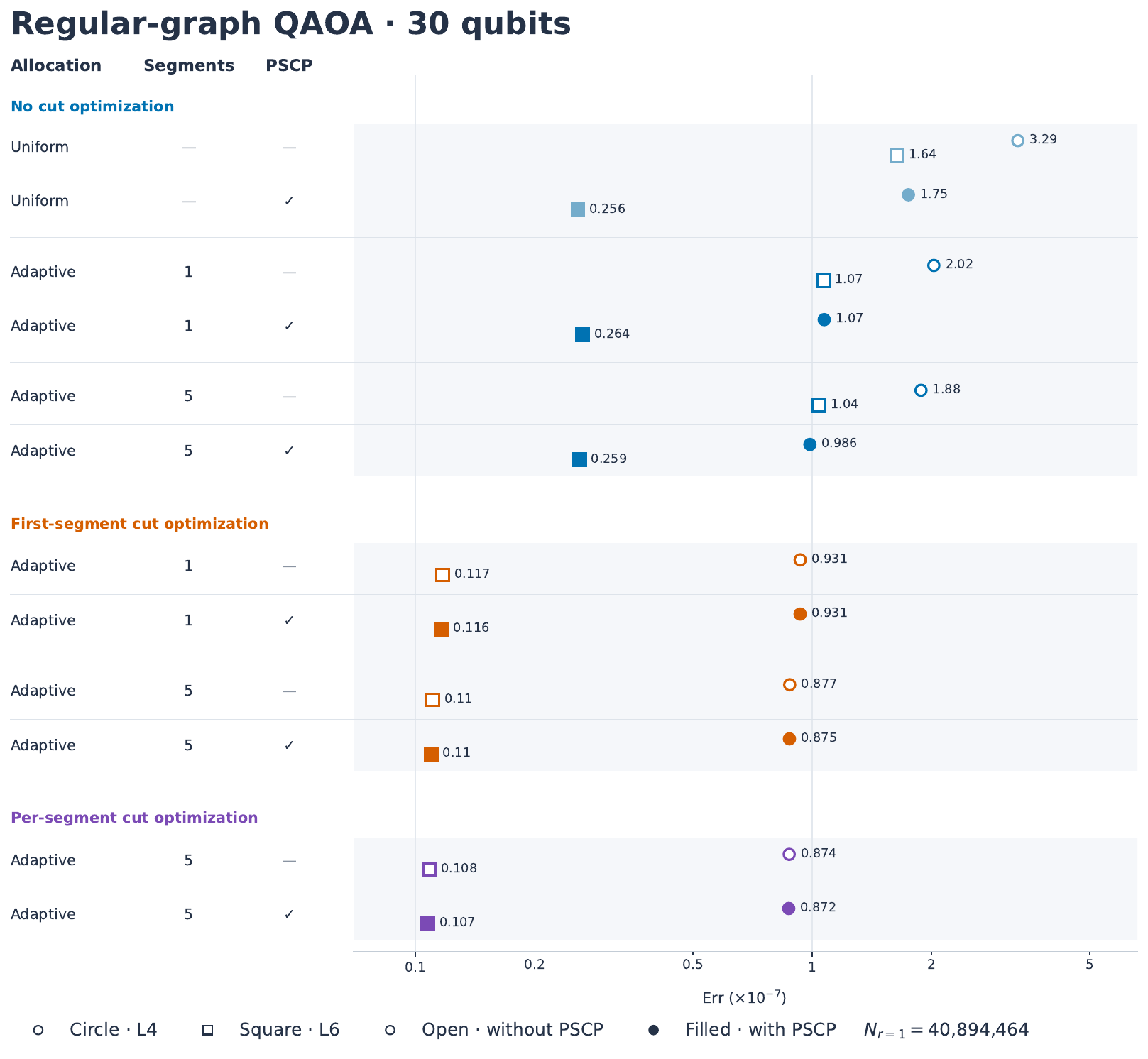}
\caption{Estimated sampling variance $\mathrm{Err}$ across different optimization settings for the large Random Regular QAOA circuit.}
\label{fig:regular_large}
\end{figure}

\begin{figure}
\centering
\includegraphics[width=0.85\linewidth]{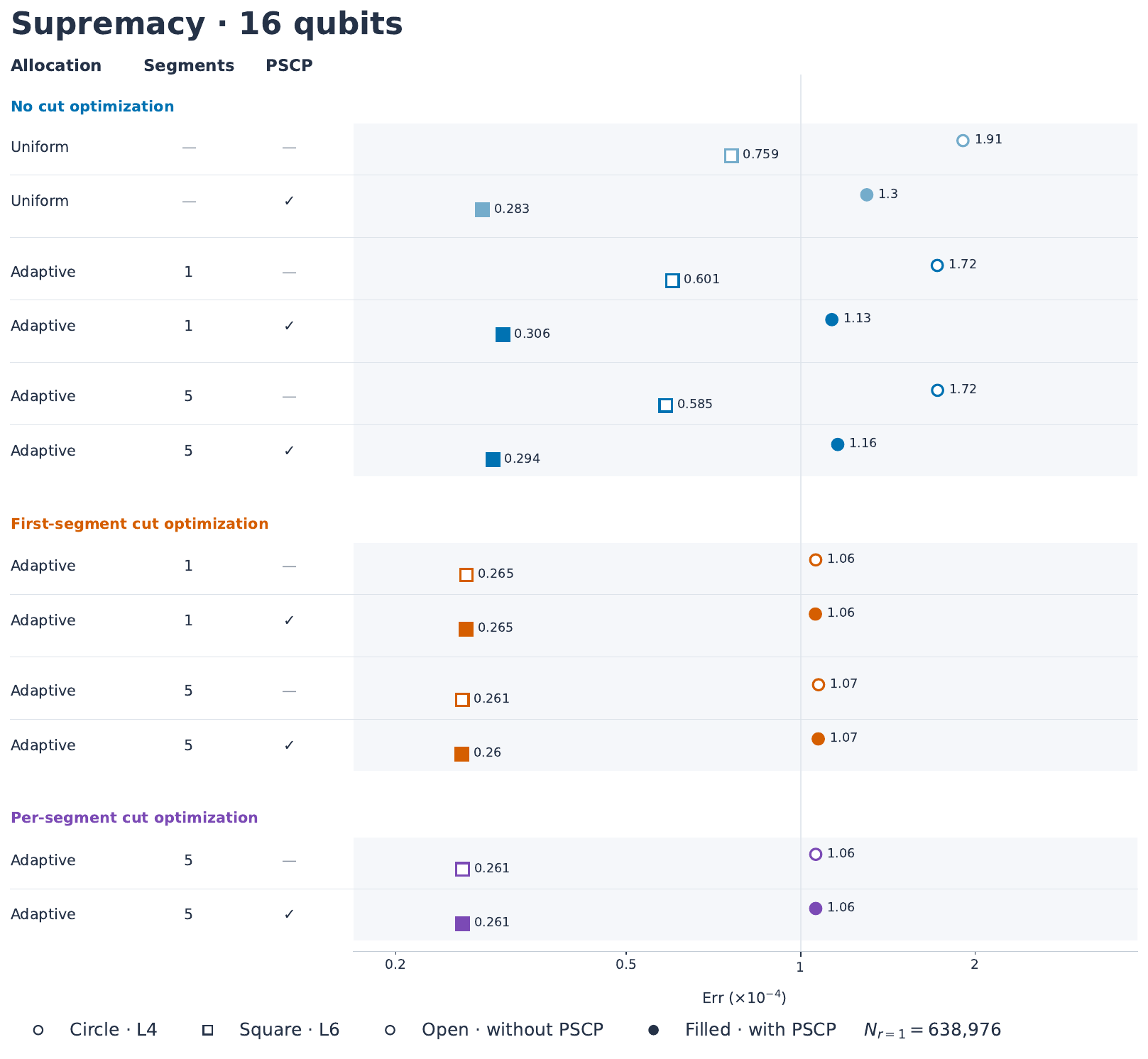}
\caption{Estimated sampling variance $\mathrm{Err}$ across different optimization settings for the small Quantum Supremacy circuit.}
\label{fig:supremacy_small}
\end{figure}

\begin{figure}
\centering
\includegraphics[width=0.85\linewidth]{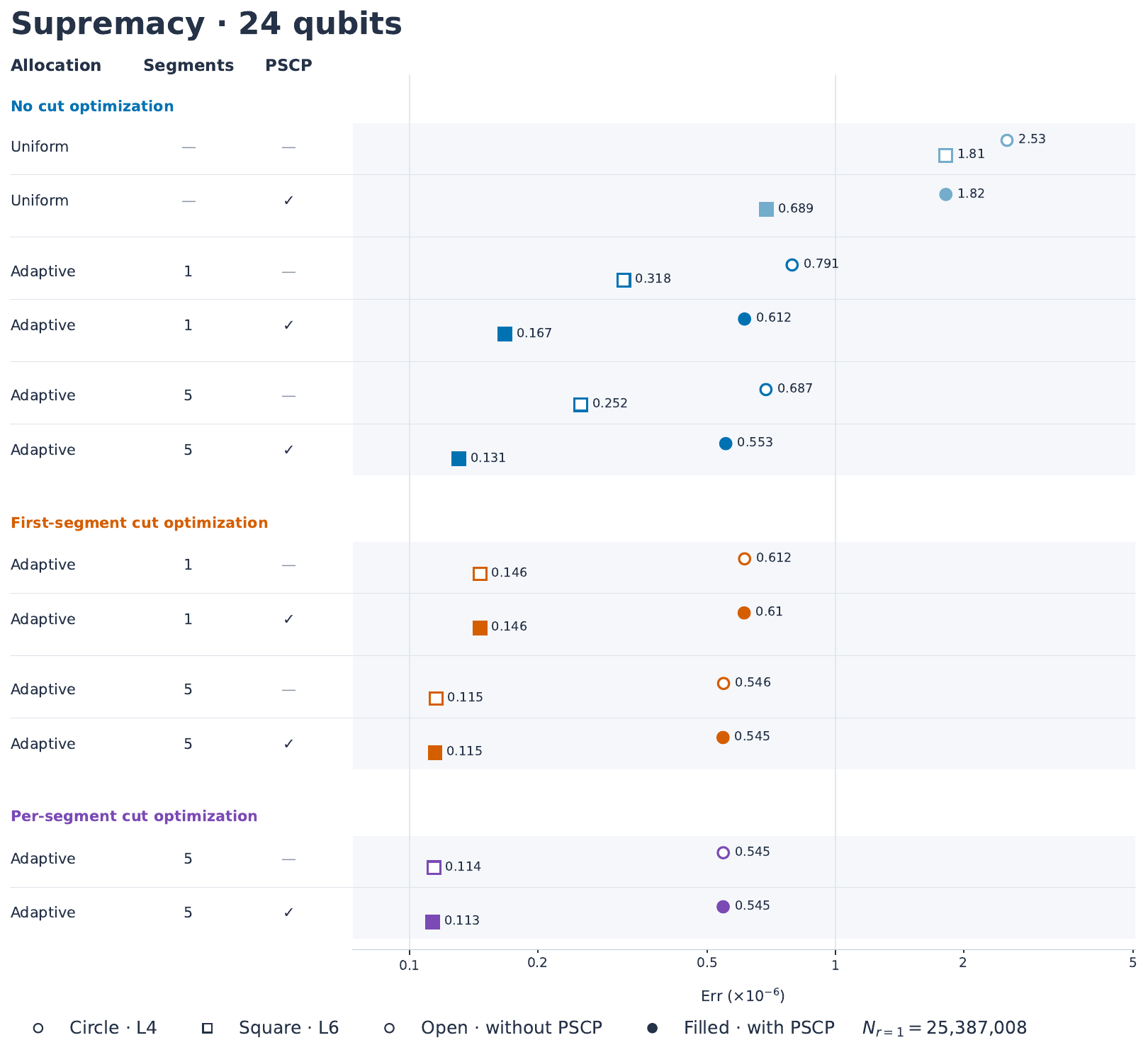}
\caption{Estimated sampling variance $\mathrm{Err}$ across different optimization settings for the medium Quantum Supremacy circuit.}
\label{fig:supremacy_medium}
\end{figure}

\begin{figure}
\centering
\includegraphics[width=0.85\linewidth]{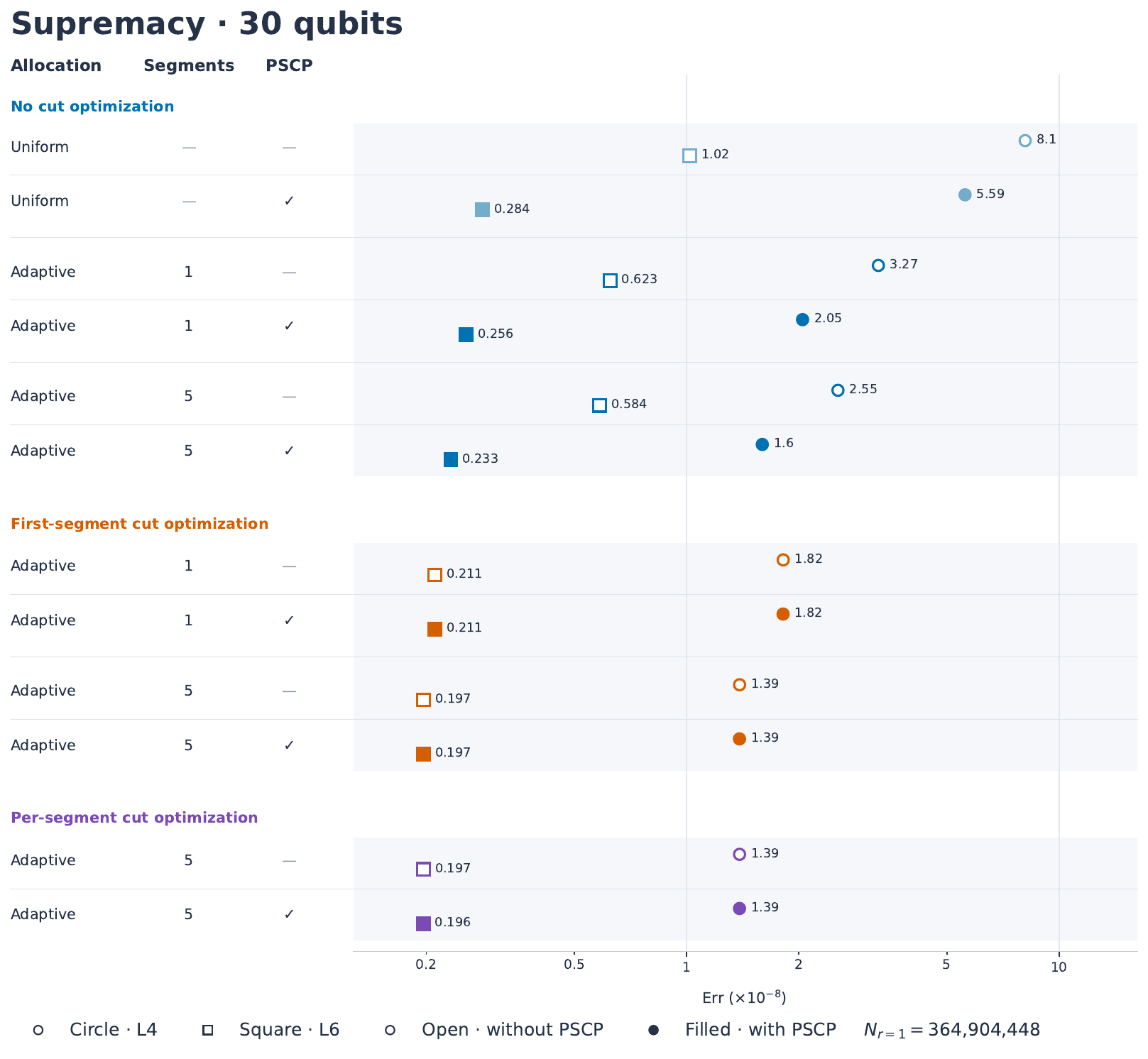}
\caption{Estimated sampling variance $\mathrm{Err}$ across different optimization settings for the large Quantum Supremacy circuit.}
\label{fig:supremacy_large}
\end{figure}

\section{Analysis of Probability Mass Concentration}\label{sec:profile analysis}
Our primary benchmark evaluation demonstrates that the performance of ShotQC is intrinsically linked to the structural characteristics of the underlying circuits---most notably their \emph{probability mass concentration}. We define probability mass concentration as the degree to which the output probability distribution, $\{p_a \mid a = 0, \dots, 2^n - 1\}$ for an $n$-qubit circuit, is concentrated over a sparse subset of bitstrings as opposed to being uniformly distributed across the full computational basis.

To elucidate the circuit-dependent mechanisms driving these variance-reduction gains, we analyze the ground-truth output probability distributions of both the uncut global circuits and their constituent cut subcircuits (denoted as ``Subcircuit 1'' and ``Subcircuit 2'') across all five benchmark families. To visualize mass concentration independently of system size, we sort the output probabilities $\{p_a\}$ in descending order and plot them against their normalized rank (where ranks are divided by $2^n - 1$, placing the lowest-probability bitstring at $2^n/2^n$). These ranked probability profiles provide a concise visual diagnostic, enabling a clear distinction between uniform, broad/diffuse, and highly concentrated (sparse) distributions.

We first examine the mass concentration of the uncut circuits for the small benchmark instances, where exact classical simulation can easily yield full global probability profiles. As shown in \figref{fig:circuit_profiles}, the five benchmark families span markedly different concentration regimes:
\begin{itemize}
    \item \textbf{AQFT}: Exhibits a strictly uniform distribution, appearing as a flat horizontal profile where every basis state is equiprobable.
    \item \textbf{Adder}: Lies at the opposite extreme, presenting a deterministic profile with all probability mass concentrated on a single output bitstring.
    \item \textbf{Erd{\"o}s--R{\'e}nyi QAOA, Random Regular QAOA, and Quantum Supremacy}: Occupy the intermediate regime, where probability mass is distributed across many basis states with non-equal weights, producing broad, smoothly decaying profiles.
\end{itemize}

\begin{figure}[t]
    \centering
    \includegraphics[width=0.7\linewidth]{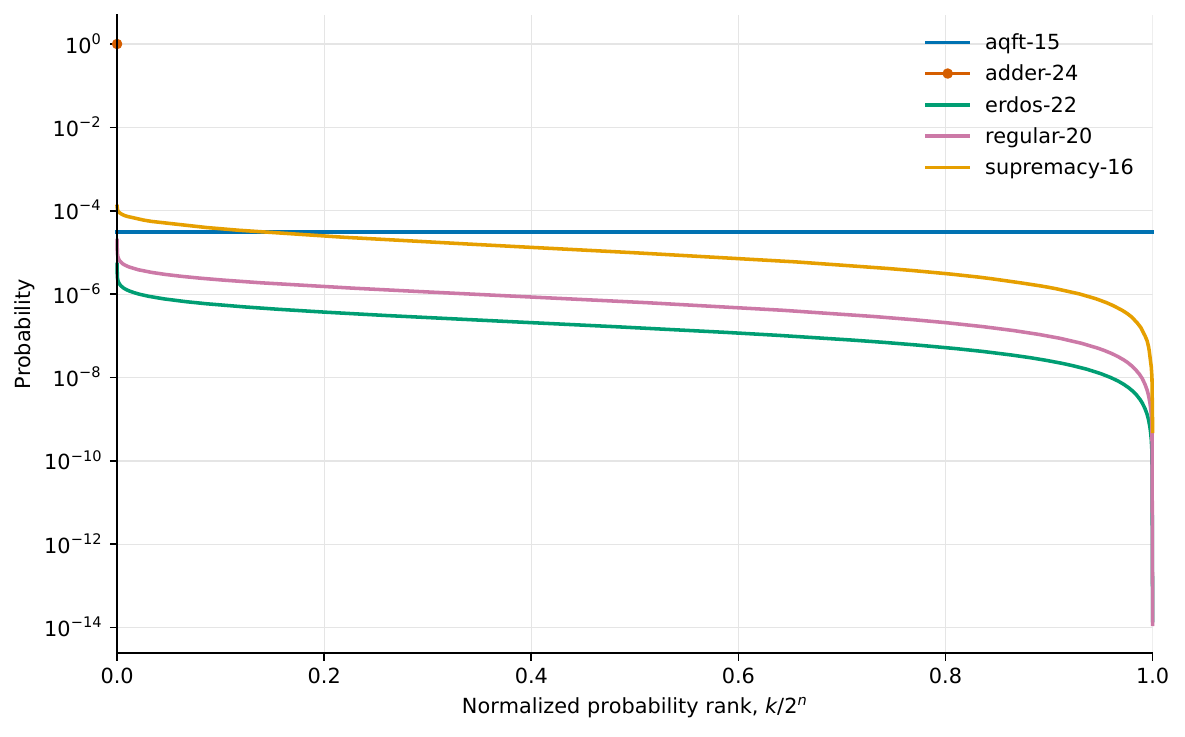}
    \caption{Ranked output probability distributions of the uncut small benchmark circuits plotted against normalized descending rank.}
    \label{fig:circuit_profiles}
\end{figure}

Next, we analyze the output probability distributions of the cut subcircuits in the large benchmark instances. Unlike global uncut circuits, each cut subcircuit is executed under a variety of preparation and measurement basis settings, producing a family of distinct probability distributions rather than a single profile. To capture this variability compactly, \figref{fig:subcircuit_profiles} plots the profiles displaying the \emph{most}, \emph{median}, and \emph{least} mass concentration for Subcircuit 1 and Subcircuit 2 across each benchmark family.

While subcircuits generally inherit the overarching concentration characteristics of their global parents, individual preparation and measurement basis choices introduce substantial profile variations. Notably, Adder subcircuits remain extremely sparse, concentrating their mass into only one, two, or four dominant bitstrings. Conversely, AQFT subcircuits exhibit configuration-dependent profiles ranging from strongly concentrated to nearly uniform. For the remaining families, Erd{\"o}s--R{\'e}nyi QAOA subcircuits exhibit stronger concentration than Random Regular QAOA, while Quantum Supremacy subcircuits remain the most diffuse.

\begin{figure}[t]
    \centering
    \includegraphics[width=0.85\linewidth]{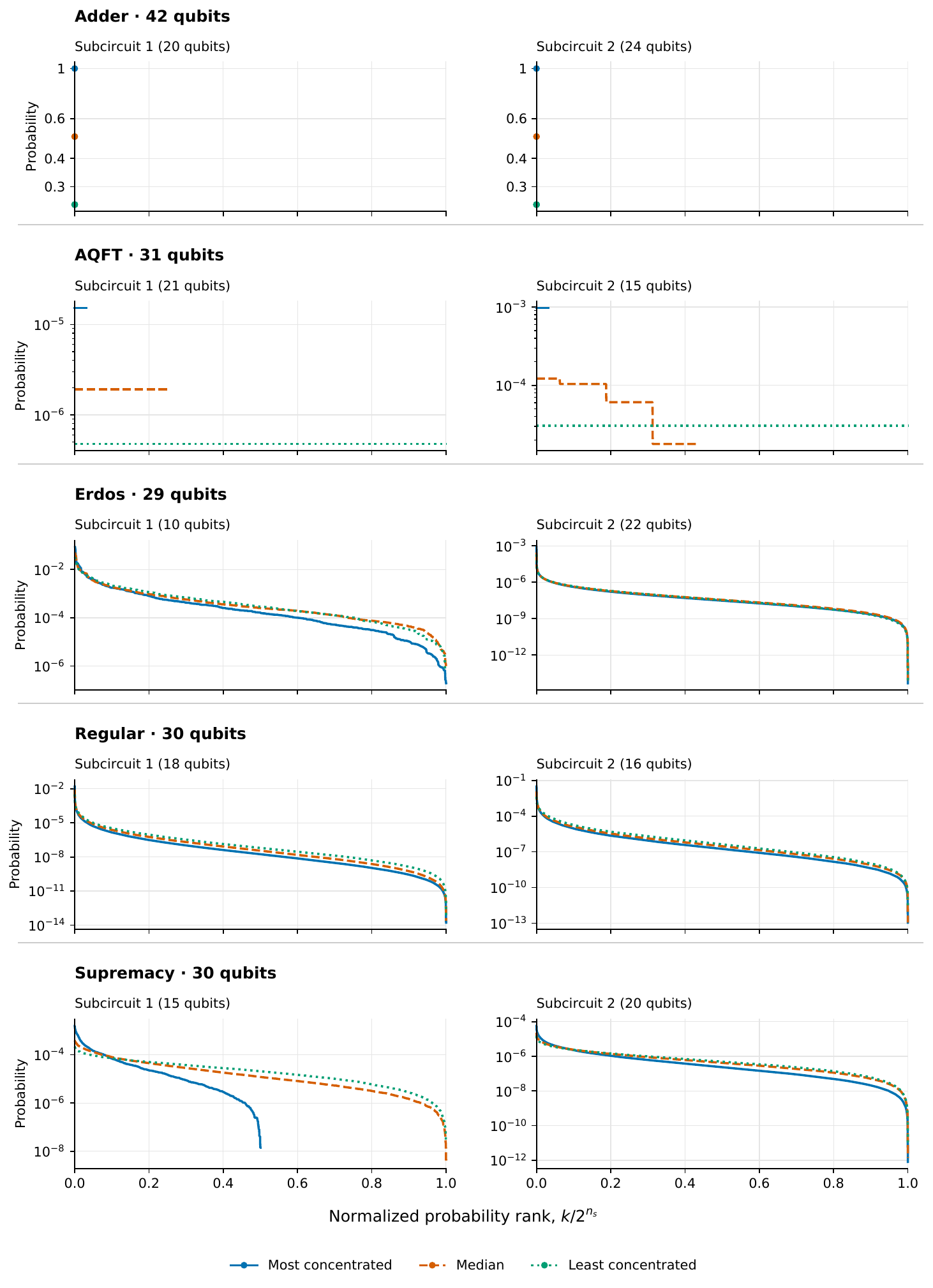}
    \caption{Ranked output probability distributions of Subcircuit 1 (left column) and Subcircuit 2 (right column) against normalized descending rank for the large benchmark instances. Profiles corresponding to the most, median, and least mass-concentrated configurations are shown for each subcircuit.}
    \label{fig:subcircuit_profiles}
\end{figure}

These subcircuit profiles shed light on the family-dependent performance trends observed in \secref{sec:benchmark results}. The statistical variance of a subcircuit estimator depends directly on its underlying output profile: higher mass concentration inherently yields a lower variance for a given shot allocation and presents greater structural leverage for adaptive optimization routines to exploit. 

This profile analysis also provides a physical explanation for why the L6 wire-cutting scheme routinely outperforms L4. Although L6 requires evaluating a larger total number of subcircuit configurations, its richer preparation basis can expose subcircuit settings with significantly higher mass concentration. When coupled with adaptive shot allocation or cut parameter optimization, ShotQC naturally channels sampling resources into these highly concentrated, low-variance configurations, lowering the overall reconstruction error.

However, mass concentration alone does not fully determine total reconstruction variance. As formulated in \eqref{eq:linear variance}, the contribution of each configuration to the total variance is governed by $f^{(e)}(P, \Theta)$, which depends not only on the local covariance $\Sigma^{(e)}$ encoded by $P^{(e)}$, but also on the reconstruction sensitivity coefficients and cut parameters $\Theta$. This distinction is particularly evident in the Adder benchmark: despite both L4 and L6 subcircuits remaining extremely sparse across all configurations, their sensitivities to cut parameter optimization diverge significantly due to differences in the underlying reconstruction matrices.

\section{Performance Gains of Adaptive Shot Allocation in Edge Cases}\label{sec:variance reduction scaling}

To illustrate that adaptive shot allocation can yield dramatic, order-of-magnitude variance reductions in extreme regimes, we examine the Bernstein--Vazirani (BV) algorithm~\cite{doi:10.1137/S0097539796300921} as a canonical benchmark for edge-case performance. Specifically, we consider a deliberate cutting scheme wherein wire cuts are placed immediately preceding the final layer of Hadamard gates. 

While this cutting configuration is deliberately impractical---as it cuts gates only in the final layer and thus offers no subcircuit execution benefit---the BV circuit possesses a highly uniform, repeated structure. This symmetry enables exact analytical expressions for all subcircuit state probability distributions across an arbitrary number of cuts $K$. The BV circuit thus serves as an ideal theoretical sandbox to isolate how probability mass concentration interacts with optimal shot allocation as the cut count scales.

In the ideal setting, the BV algorithm identifies a secret bitstring $s \in \{0,1\}^n$ with certainty using a single query to a phase oracle $U_s$, defined as
\begin{equation}
U_s \ket{x} = (-1)^{x \cdot s}\ket{x},
\end{equation}
where $x \cdot s$ denotes the bitwise inner product modulo two. The circuit structure is illustrated in \figref{fig:bv_circuit}. An ancilla qubit prepared in $\ket{-} \equiv (\ket{0} - \ket{1})/\sqrt{2}$ is utilized to implement phase kickback. Immediately prior to the final layer of Hadamard gates, the joint state of the system is given by
\begin{equation}
\ket{\psi} = \left( \bigotimes_{j=1}^{n} Z^{s_j}\ket{+} \right) \ket{-} = \left( \bigotimes_{j=1}^{n} H\ket{s_j} \right) \ket{-}.
\label{eq:bv_pre_final_h}
\end{equation}
Applying the final layer of Hadamard gates transforms $\ket{\psi}$ directly into $\ket{s}\ket{-}$ up to a global phase.

\begin{figure}[t]
\centering
\begin{quantikz}
\lstick{$q_1:\ket{0}$}
    & \gate{H}
    & \ctrl{4} 
    \gategroup[
        wires=5,
        steps=4,
        style={inner sep=0.5pt},
        label style={label position=above,anchor=south}
      ]{$U_{s=1111}$}& & &
    & \gate[style={draw=none,fill=none,text=red}]{\times}
    \gategroup[
        wires=2,
        steps=1,
        style={dashed,inner sep=0.5pt},
        label style={label position=above,anchor=south}
        ]{$K=2$}        
    & \gate{H} & \meter{}
\\
\lstick{$q_2:\ket{0}$}
    & \gate{H}
    & & \ctrl{3} & &
    & \gate[style={draw=none,fill=none,text=red}]{\times}
    & \gate{H} & \meter{}
\\
\lstick{$q_3:\ket{0}$}
    & \gate{H}
    & & & \ctrl{2} &
    &
    & \gate{H} & \meter{}
\\
\lstick{$q_4:\ket{0}$}
    & \gate{H}
    & & & & \ctrl{1}
    &
    & \gate{H} & \meter{}
\\
\lstick{$a:\ket{1}$}
    & \gate{H}
    & \targ{} & \targ{} & \targ{} & \targ{}
    & & & \rstick{$\ket{-}$}
\end{quantikz}
\caption{The Bernstein--Vazirani circuit encoding a $4$-bit secret string $s=1111$, evaluated under $K=2$ wire cuts placed immediately prior to the final Hadamard gate layer.}
\label{fig:bv_circuit}
\end{figure}

We place $K$ wire cuts immediately upstream of the final Hadamard gates on the first $K$ qubits, as indicated by the dashed region in \figref{fig:bv_circuit}. From \eqref{eq:bv_pre_final_h}, the single-qubit expectation values associated with the $j$-th cut wire in the upstream component satisfy
\begin{equation}
\langle X_j\rangle = (-1)^{s_j}, \qquad \langle Y_j\rangle = \langle Z_j\rangle = 0.
\label{eq:bv_cut_expectations}
\end{equation}
Consequently, an $X$-basis measurement on the upstream subcircuit is strictly deterministic (yielding outcome $s_j$), whereas $Y$- and $Z$-basis measurements yield outcomes $0$ and $1$ with equal probability. Because $\ket{\psi}$ in \eqref{eq:bv_pre_final_h} is a product state, joint probabilities across the $K$ cut wires factorize completely. The probability distributions of all subcircuit configurations (both the upstream component and the downstream single-qubit Hadamard circuits) can therefore be derived analytically.

For a downstream single-qubit subcircuit injected with a preparation state $\rho \in S_{\mathrm{prep}}$, the probability of observing outcome $b \in \{0,1\}$ following the Hadamard gate is
\begin{equation}
p(b \mid \rho) = \Tr \big( \ket{b}\!\bra{b} H\rho H \big) = \frac{1 + (-1)^b \Tr(\rho X)}{2}.
\label{eq:bv_downstream_probability}
\end{equation}
Increasing $K$ replicates this identical, analytically tractable local structure across the cut boundary without altering the qualitative nature of the subcircuits.

Under the L4 scheme with secret string $s=11\cdots1$, substituting these exact probability distributions into the variance objective $f^{(e)}(P,\Theta)$ yields the ground-truth improvement factor $R_{\mathrm{gt}}(K)$:
\begin{equation}
R_{\mathrm{gt}}(K) = \frac{|E|\left(3^K - 2^K\right)}{\left[ \displaystyle\sum_{u=1}^{K} \binom{K}{u} 2^u \sqrt{1 - 2^{-u}} \right]^2}, \qquad |E| = 3^K + 4^K.
\label{eq:bv_ground_truth_R}
\end{equation}
Equation~\eqref{eq:bv_ground_truth_R} demonstrates that the intrinsic variance reduction factor $R_{\mathrm{gt}}(K)$ grows exponentially with the cut count $K$, independently of the total bitstring register size $n$. This proves that exponential gains in sampling efficiency are theoretically achievable when probability mass concentration is extreme.

To examine how adaptive allocation behaves in practical finite-shot settings, we compare two two-stage sampling scenarios under the L4 scheme:
\begin{enumerate}
    \item \textbf{Uniform Allocation}: A total budget $N$ is divided equally across two stages ($N_{r=0.5}$ each), with shots distributed uniformly among all configurations in both phases.
    \item \textbf{Adaptive Allocation}: An initial budget $N_{r=0.5}$ is distributed uniformly in the prior phase, after which the remaining $N_{r=0.5}$ shots are reallocated adaptively based on the empirical estimates gathered in the first phase.
\end{enumerate}

Both scenarios share an identical first phase, yielding an initial baseline variance $\mathrm{Err}_b$. Let $\mathrm{Err}_0$ and $\mathrm{Err}_{\mathrm{opt}}$ denote the final estimated variance after the second phase for the uniform and adaptive strategies, respectively. The ratios $R_1 = \mathrm{Err}_b / \mathrm{Err}_0$ and $R_2 = \mathrm{Err}_b / \mathrm{Err}_{\mathrm{opt}}$ quantify the variance reduction achieved by adding the second shot budget. Their relative ratio
\begin{equation}
R_{\mathrm{rel}} \equiv \frac{R_2}{R_1} = \frac{\mathrm{Err}_{0}}{\mathrm{Err}_{\mathrm{opt}}}
\label{eq:bv_relative_reduction}
\end{equation}
directly measures the performance gain of adaptive allocation over uniform allocation at equal total shot budgets.

Numerical simulation results for a $(7+1)$-qubit BV circuit ($K=1,\dots,6$) and a $(24+1)$-qubit BV circuit ($K=1,\dots,4$) are shown in \figref{fig:scalability}. For the $8$-qubit system, the adaptive reduction factor dips slightly from $K=1$ to $K=2$ before rising monotonically to $5.99$ at $K=6$. This aligns closely with the theoretical ground-truth profile in \eqref{eq:bv_ground_truth_R}, which exhibits a shallow minimum at $K=2$ ($R_{\mathrm{gt}}(2) \simeq 3.16$) and reaches $R_{\mathrm{gt}}(6) \simeq 6.64$. In contrast, the uniform reduction factor remains near $2.0$ across all $K$, matching the expected $1/N$ scaling of variance when doubling the sample size. For the $25$-qubit circuit, the uniform factor remains near $2.0$, whereas the empirical adaptive improvement factor reaches extraordinary magnitudes, ranging from $4.73 \times 10^4$ at $K=1$ to $3.10 \times 10^4$ at $K=4$.

\begin{figure}[t]
\centering
\includegraphics[width=0.85\linewidth]{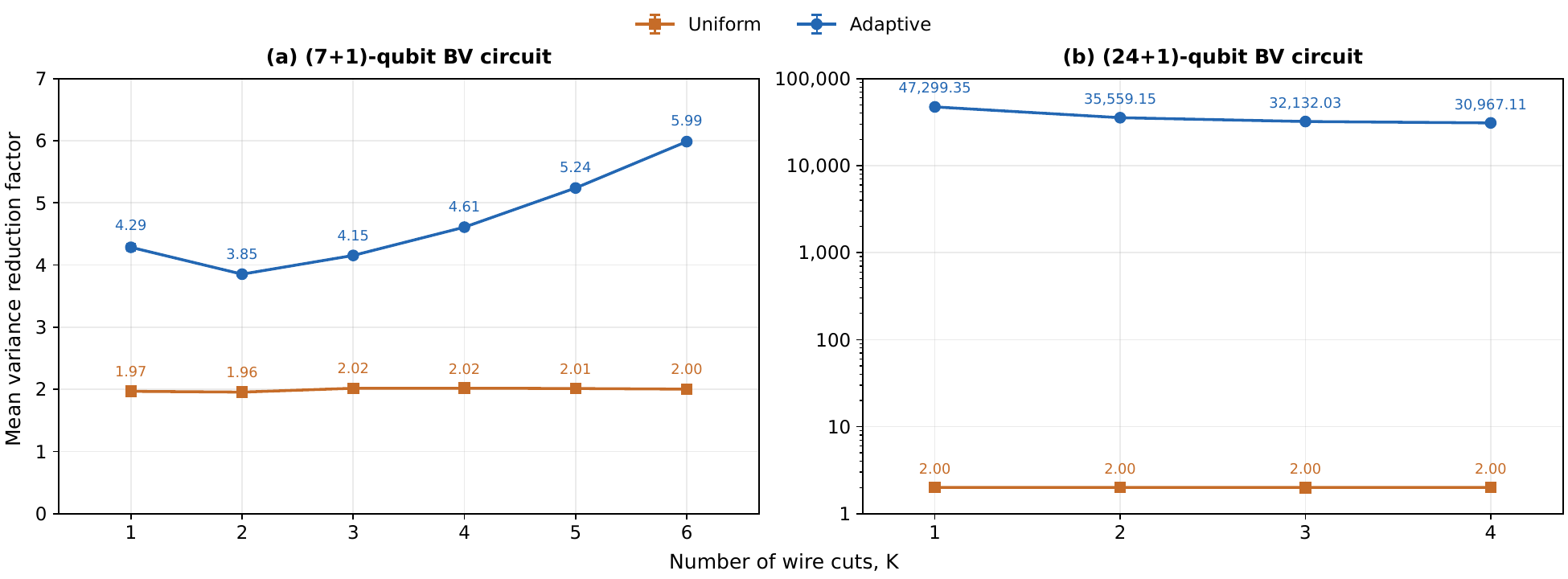}
\caption{Variance reduction factors $R_1$ (uniform) and $R_2$ (adaptive) for Bernstein--Vazirani circuits across cut numbers $K$. (a) $(7+1)$-qubit circuit on a linear scale. (b) $(24+1)$-qubit circuit on a logarithmic scale.}
\label{fig:scalability}
\end{figure}

To confirm that reductions in estimated variance translate directly to improved output distribution fidelity, we evaluate the half-$L_1$ norm (total variation distance) between the reconstructed distribution $\tilde{P}\equiv\{\tilde{p}_a \mid a=0,\dots,2^n-1\}$ and the ground-truth distribution $\hat{P}\equiv\{\hat{p}_a \mid a=0,\dots,2^n-1\}$:
\begin{equation}
D_{L_1}\left(\tilde{P}, \hat{P}\right) = \frac{1}{2} \sum_{a=0}^{2^n-1} \left| \tilde{p}_a - \hat{p}_a \right|.
\end{equation}
Because $\tilde{P}$ is reconstructed from subcircuit estimators without explicit probability simplex constraints, entries may occasionally be negative, allowing $D_{L_1}$ to exceed $1.0$.

As illustrated in \figref{fig:bv_tvd}, adaptive shot allocation consistently yields reconstructed probability distributions closer to the ground truth than uniform sampling. For the $25$-qubit circuit, the reconstruction error under adaptive allocation is several orders of magnitude smaller than under uniform allocation. This confirms that variance reductions achieved through adaptive shot distribution translate into tangible improvements in output accuracy.

\begin{figure}
\centering
\includegraphics[width=0.85\linewidth]{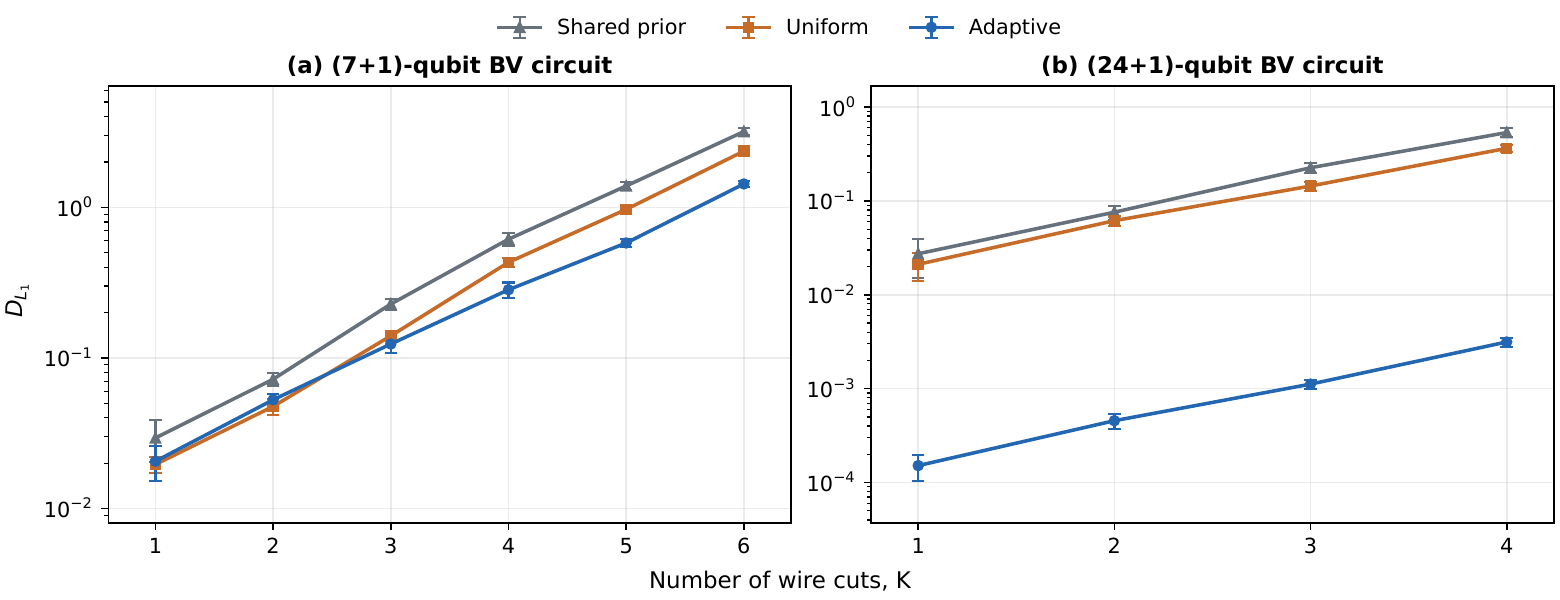}
\caption{Reconstruction error $D_{L_1}$ between the reconstructed output distribution $\tilde{P}$ and the ground-truth BV distribution $\hat{P}$ for (a) $(7+1)$-qubit and (b) $(24+1)$-qubit circuits. The shared prior curve depicts error after the initial $N_{r=0.5}$ shots.}
\label{fig:bv_tvd}
\end{figure}

In summary, these analytical and empirical evaluations illuminate the mechanics of adaptive shot distribution in highly concentrated regimes. Equation~\eqref{eq:bv_ground_truth_R} confirms that an intrinsic, exponential scaling advantage $R_{\mathrm{gt}}(K)$ exists for mass-concentrated circuits. The even larger performance gains observed empirically in large system sizes reflect finite-shot effects: when ground-truth distributions are extremely concentrated, initial prior samples rapidly identify dominant configurations, allowing posterior adaptive iterations to suppress variance far more aggressively than static theoretical bounds predict.

\bibliography{references}

\end{document}